\documentclass[12pt]{amsart}
\usepackage[margin=0.94in]{geometry}
\usepackage{amssymb, comment,hyperref,slashed,tensor}
\usepackage{ mathrsfs }
\usepackage{ mathtools}
\usepackage{ wasysym , todonotes, amsmath, tikz}
\usetikzlibrary{intersections}
\usetikzlibrary{decorations.pathmorphing}

\newtheorem{theorem}{Theorem}[section]
\newtheorem{lemma}[theorem]{Lemma}
\newtheorem{proposition}[theorem]{Proposition}

\newtheorem{cor}[theorem]{Corollary}

\theoremstyle{definition}
\newtheorem{definition}[theorem]{Definition}

\theoremstyle{remark}
\newtheorem{remark}[theorem]{Remark}

\numberwithin{equation}{section}

\newcommand{\R}{\mathbb{R}}  

\newcommand{\N}{\mathbb{N}}
\newcommand{\de}{\mbox{ d}}

\newcommand{\norm}[1]{\left\|#1\right\| }

\newcommand{\desude}[2]{\frac{\partial #1}{\partial #2}}
\newcommand{\desu}[2]{\frac{d #1}{d #2}}

\newcommand{\fara}{\mathcal{F}}

\newcommand{\snabla}{\slashed{\nabla}}
\newcommand{\sdelta}{\slashed{\Delta}}

\newcommand{\lbar}{{\underline{L}}}
\newcommand{\curl}{\slashed{\mbox{curl }}}
\newcommand{\dive}{\slashed{\mbox{div }}}

\newcommand{\alphabar}{{\underline{\alpha}}}

\newcommand{\gbar}{\slashed{g}}
\newcommand{\fdual}{{^\star}F}
\newcommand{\svol}{\slashed{\varepsilon}}
\newcommand{\duu}{\mathfrak{D}_{u_1}^{u_2}}
\newcommand{\ddue}{{\slashed{\mathcal{D}}}^\star_2}
\newcommand{\duno}{{\slashed{\mathcal{D}}}^\star_1}
\newcommand{\daiv}{\text{div}\,}
\newcommand{\cherl}{\text{curl}\,}
\newcommand{\phibar}{\underline{\phi}}
\newcommand{\desphere}{\de \mathbb{S}^2}

\newcommand{\spheq}{\stackrel{\mathbb{S}^2}{=}}
\newcommand{\slie}{\slashed{\mathcal{L}}}
\newcommand{\ttildeu}{\tilde{\tilde u}_n}
\newcommand{\ttildev}{\tilde{\tilde v}_n}
\newcommand{\mall}{m_{\text{all}}[\phi]}

\newcommand{\snorm}[4]{\norm{#1}_{C_{u_0, v_0}; #2; #3, #4}}
\newcommand{\fluxob}{\fara^\smallsetminus}

\newcommand{\fackip}{Fackerell--Ipser }
\newcommand{\pvar}{\stackrel{_{(1)}}{P}}
\newcommand{\pbarvar}{\stackrel{_{(1)}}{\underline P}}
\newcommand{\Psivar}{\stackrel{_{(1)}}{\Psi}}
\newcommand{\Psibarvar}{\stackrel{_{(1)}}{\underline \Psi}}
\newcommand{\psivar}{\stackrel{_{(1)}}{\psi}}
\newcommand{\avar}{\stackrel{_{(1)}}{\alpha}}
\newcommand{\abarvar}{\stackrel{_{(1)}}{\alphabar}}
\newcommand{\sigmavar}{\stackrel{_{(1)}}{\sigma}}
\newcommand{\rhovar}{\stackrel{_{(1)}}{\rho}}
\newcommand{\chivar}{\stackrel{_{(1)}}{\hat \chi}}
\newcommand{\chibarvar}{\stackrel{_{(1)}}{\hat {\underline \chi}}}
\newcommand{\conplus}{{\overline C}}
\newcommand{\conminus}{{\underline C}}
\newcommand{\future}{J^+}
\begin{document}

\title{The spin $\pm$1 Teukolsky equations and the Maxwell system on Schwarzschild}

\author{Federico Pasqualotto}
\address{\small Princeton University, Department of Mathematics, Fine~Hall,~Washington~Road,~Princeton,~NJ~08544,~United~States\vskip.2pc \small University of Cambridge, Department of Pure Mathematics and Mathematical
Statistics, Wilberforce~Road,~Cambridge~CB3~0WA,~United~Kingdom\vskip.2pc }
\email{fp2@princeton.edu}

\maketitle

\begin{abstract}
In this note we prove decay for the spin $\pm$1 Teukolsky Equations on the Schwarzschild spacetime. These equations are those satisfied by the extreme components ($\alpha$ and $\alphabar$) of the Maxwell field, when expressed with respect to a null frame. The subject has already been addressed in the literature, and the interest in the present approach lies in the connection with the recent work by Dafermos, Holzegel and Rodnianski on linearized gravity  [M. Dafermos, G. Holzegel and I. Rodnianski, \emph{The linear stability of the Schwarzschild solution to gravitational perturbations},  \hyperlink{http://arxiv.org/abs/1601.06467}{preprint (2016)}]. In analogy with the spin $\pm2$ case, it seems difficult to directly prove Morawetz estimates for solutions to the spin $\pm1$ Teukolsky Equations. By performing a differential transformation on the extreme components $\alpha$ and $\alphabar$, we obtain quantities which satisfy a \fackip Equation, which \emph{does} admit a straightforward Morawetz estimate, and is the key to the decay estimates. This approach is exactly analogous to the strategy appearing in the aforementioned work on linearized gravity. We achieve inverse polynomial decay estimates by a streamlined version of the physical space $r^p$ method of Dafermos and Rodnianski. Furthermore, we are also able to prove decay for all the components of the Maxwell system. The transformation that we use is a physical space version of a fixed-frequency transformation which appeared in the work of Chandrasekhar \cite{chandrateu}. The present note is a version of the author's master thesis and also serves the ``pedagogical'' purpose to be as complete as possible in the presentation.
\end{abstract}

\section{Introduction}

The subject of black hole stability has received a good amount of attention lately, as research efforts are focused on proving the full nonlinear stability of the Kerr family of black holes. See the lecture notes \cite{lecturenotes} for a comprehensive introduction on the topic. The interest in the aforementioned problem stems from the fundamental question of whether the Kerr solution indeed provides an appropriate description of physical reality.

In an attempt to address the fully nonlinear problem, researchers have been following a natural path: first, one studies the covariant scalar wave equation (spin 0). Then, one studies the Maxwell equations (in which the extreme components satisfy spin $\pm$1 Teukolsky equations). Finally, one seeks to study the linearized Einstein equations (in which the extreme curvature components satisfy spin $\pm$2 Teukolsky equations). Eventually, one hopes that this process would lead to a deeper understanding of the nonlinear structure, in order to indeed address the nonlinear stability of the Kerr family.

Hence, the subject of decay of linear waves on a black-hole background has been recently studied, with many contributions by different research groups. For the first step of the ``linear program'', these efforts culminated in the proof of decay of scalar waves on a Kerr background for $|a|< M$, by Dafermos, Rodnianski and Shlapentokh-Rothman \cite{fullkerr}.

To proceed in the outlined program, the Maxwell equations have also been studied. Boundedness and decay for solutions to the Maxwell equations has been first proved by Blue (\cite{blue}) on the Schwarzschild background. There were advances in extending these results to the Kerr setting in the slowly rotating case ($a \ll M$) by Andersson and Blue in \cite{blue2}.
Furthermore, Sterbenz and Tataru proved local energy decay for the Maxwell field on a large class of spherically symmetric spacetimes in \cite{stetat}. In addition, Metcalfe, Tataru and Tohaneanu proved pointwise decay for the Maxwell field under the assumption of local energy decay, for a fairly general class of asymptotically flat spacetimes, in the paper \cite{mettattoh}. See also \cite{sari}.

Finally, Andersson, B\"ackdahl and Blue found a new way of producing robust energy estimates on the Schwarzschild background, exploiting a super-energy tensor. The relevant paper is \cite{newblue}.

Recently, there has also been a substantial advance in the last step of this ``linear program'', i.e.~the proof of linear stability of the Schwarzschild metric under gravitational perturbations. Indeed, a recent result of Dafermos, Holzegel and Rodnianski \cite{linearized} shows that the Schwarzschild metric is stable under linearized gravitational perturbations.

\subsection{Maxwell: why another proof?}

In this paper, we return to the topic of decay of the Maxwell field on a curved background. The aim of the note is threefold: first, we provide a simple proof of pointwise decay of the Maxwell field on the Schwarzschild background. Second, we adopt the further ``didactic'' aim to be as detailed as possible in our exposition. The third and main motivation for this work, though, lies in the connection with the aforementioned problem of linear stability of the Schwarzschild metric. It has become evident that a very similar approach to the one in \cite{linearized} can be adopted to address the decay properties of the spin $\pm$1 Teukolsky equations and of the Maxwell system on Schwarzschild. 

We briefly recall the strategy followed by the authors in \cite{linearized}. Given a solution to the equation for the extreme curvature components (i.e., the spin $\pm$2 Teukolsky equation), the authors find a second order differential transformation that performs the following: if we apply the transformation on the extreme curvature component, the resulting expression satisfies a ``good'' equation (i.e., an equation for which Morawetz and energy estimates can be proved.) The resulting equation is called the Regge--Wheeler equation.
The relevant quantity enables us to estimate all the components of the field, just by using transport equations.

We wish to follow the same path in the context of the Maxwell system, which we study on the Schwarzschild background. We start from the spin $\pm$1 Teukolsky equations, which are satisfied by the extreme components $\alpha$ and $\alphabar$. We exploit an elementary transformation (which can be found, in its fixed-frequency form, in the work of Chandrasekhar \cite{chandrateu}). The transformation takes the spin $\pm$1 Teukolsky equations into a ``good'' equation, called the \fackip equation. The transformed quantity has the required property of vanishing on the zeroth mode and of satisfying good integrated decay estimates. We finally use the transformed quantity to estimate, via transport equations, first the extreme components $\alpha$ and $\alphabar$, and subsequently the remaining components of the Maxwell field.

Let us finally note that earlier versions of the results contained in this note were originally obtained in my master thesis at ETH Z\"urich \cite{mythes}.

\subsection{Outline of the note} We will first motivate our analysis in Section~\ref{sec:motiv}, giving an outline of the work \cite{linearized}, as well as a sketch of the argument in the present note. 

We subsequently introduce some necessary notation and the null decomposition of the Maxwell system in Section~\ref{sec:notation}, as well as the crucial transformation which takes the spin $\pm$1 Teukolsky equations into the \fackip equation.

We then proceed to state the main results of this note in Section~\ref{sec:results}:
\begin{itemize}
\item decay for solutions to the spin $\pm$ 1 Teukolsky equations (Theorem~\ref{prop:decayteu}),
\item decay for solutions to the Maxwell system (Theorem~\ref{prop:decaymax}).
\end{itemize}

Subsequently, we prove integrated decay estimates for solutions to the \fackip equation, in Section~\ref{sec:rpmethod}. We follow the $r^p$-method approach by Dafermos--Rodnianski, as in~\cite{newmihalis}.

From these integrated estimates, via a combination of Sobolev embedding and the Gronwall inequality, we obtain decay for solutions to the spin $\pm$1 Teukolsky equations in Section~\ref{sec:decayteu}.

We improve the decay for the $\alpha$ component in Section~\ref{sec:alphaimproved}.

We then extend the decay to the full Maxwell system in Section~\ref{sec:decaymid}.

We conclude the note with the Appendices, in which we collect important lemmas and calculations. In the last Appendix \ref{sec:compare} we finally compare the Morawetz estimate obtained in this paper with that achieved in previous work by Andersson, B\"ackdahl and Blue \cite{newblue}.

\subsection{Acknowledgements}

I would like to thank my advisor, Prof.~Mihalis Dafermos, for suggesting the problem to me and for his guidance. I would also like to thank Stefanos Aretakis, Georgios Moschidis and Yakov Shalpentokh-Rothman for very insightful discussions.

\section{Motivation and main idea of the work}\label{sec:motiv}

\subsection{Introducing the work on linearized gravity}

It is sensible to recall here the strategy the authors follow in the paper \cite{linearized}.  We warn the reader that this brief subsection does not have any claim of completeness, and we refer to the paper \cite{linearized} for the full details.

In said work, the authors consider the vacuum Einstein equations:
\begin{equation}\label{eq:einstein}
\textbf{Ric}_{\mu\nu} = 0,
\end{equation}
with respect to the unknown metric $\boldsymbol{g}$, which is considered to be ``near'' the Schwarzschild metric. In the following, bold typeface will always denote quantities associated to the metric $\boldsymbol{g}$, whereas quantities in unbolded typeface will always be associated to the Schwarzschild metric $g$.

The authors perform a suitable linearization as follows. Let $(\boldsymbol{\hat L}, \boldsymbol{\hat{\lbar}},\boldsymbol{e}_1, \boldsymbol{e}_2)$ be a normalized null frame with respect to the metric $\boldsymbol{g}$. Also, $\boldsymbol{{}^\star}$ denotes the Hodge dual with respect to two indices:
\begin{equation}
\boldsymbol{{}^\star \! R}_{\mu\nu\kappa\lambda} := \frac 1 2 \boldsymbol{\varepsilon}_{\mu \nu \alpha \beta} \boldsymbol{R}\indices{_\kappa_\lambda^\alpha^\beta}.
\end{equation}
Note that, here, $\boldsymbol{{}^\star}$ is in bold typeface, hence the Hodge dual is calculated using the (bold) volume form $\boldsymbol{\varepsilon}$, which is the natural volume form induced by the metric $\boldsymbol{g}$.

In order to write the linearized Einstein equations, the authors decompose the field in null frame. They consider the Ricci coefficients:
\begin{equation}\label{eq:riccicoeff}
\begin{array}{ll}
\boldsymbol{\chi}_{AB} := \boldsymbol{g}(\boldsymbol{\boldsymbol{\nabla}}_A \boldsymbol{\hat L}, \boldsymbol{e}_B), & \boldsymbol{\underline{\chi}}_{AB} := \boldsymbol{g}(\boldsymbol{\nabla}_A \boldsymbol{\hat{\lbar}}, \boldsymbol{e}_B),\\
\boldsymbol{\eta}_A := - \frac 1 2 \boldsymbol{g}(\boldsymbol{\nabla}_{\boldsymbol{\hat{\lbar}}}\boldsymbol{e}_A, \boldsymbol{\hat L}), & \boldsymbol{\underline{\eta}}_A := - \frac 1 2 \boldsymbol{g}(\boldsymbol{\nabla}_{\boldsymbol{\hat{L}}}\boldsymbol{e}_A, \boldsymbol{\hat{\lbar}}) \\
\boldsymbol{\hat{\omega}} := \frac 1 2 \boldsymbol{g}(\boldsymbol{\nabla}_{\hat L} \boldsymbol{\hat{\lbar}}, \boldsymbol{\hat L}), & \boldsymbol{\hat{\underline{\omega}}}:= \frac 1 2 \boldsymbol{g}(\boldsymbol{\nabla}_{\boldsymbol{\hat{\lbar}}}\boldsymbol{\hat L}, \boldsymbol{\hat{\underline{L}}}), \\
\boldsymbol{\zeta} := \frac 1 2 \boldsymbol{g}(\boldsymbol{\nabla}_A \boldsymbol{\boldsymbol{\hat{L}}}, \boldsymbol{\hat{\lbar}}).
\end{array}
\end{equation}
The authors consider the null components of the Riemann tensor as well, i.e.~the set of components
\begin{equation}\label{eq:nullcomp}
\begin{array}{ll}
\boldsymbol{\alpha}_{AB} := \boldsymbol{R}(\boldsymbol{e}_A, \boldsymbol{\hat L}, \boldsymbol{e}_B, \boldsymbol{\hat L}), & \boldsymbol{\alphabar}_{AB} := \boldsymbol{R}(\boldsymbol{e}_A, \boldsymbol{\hat{\lbar}}, \boldsymbol{e}_B, \boldsymbol{\hat{\lbar}}), \\
\boldsymbol{\beta}_A := \frac 1 2 \boldsymbol{R}(\boldsymbol{e}_A, \boldsymbol{\hat L}, \boldsymbol{\hat{\lbar}}, \boldsymbol{\hat L}), & {\boldsymbol{\underline{\beta}}}_A := \boldsymbol{R}(\boldsymbol{e}_A, \boldsymbol{\hat{\lbar}}, \boldsymbol{\hat{\lbar}}, \boldsymbol{\hat L}), \\
\boldsymbol{\rho} := \frac 1 4 \boldsymbol{R}(\boldsymbol{\hat L}, \boldsymbol{\hat{\lbar}}, \boldsymbol{\hat L}, \boldsymbol{\hat{\lbar}}), & \boldsymbol{\sigma} := \frac 1 4 \boldsymbol{{}^\star} \! \boldsymbol{R}(\boldsymbol{\hat L}, \boldsymbol{\hat{\lbar}}, \boldsymbol{\hat L}, \boldsymbol{\hat{\lbar}}).
\end{array}
\end{equation}
They proceed to linearize the Einstein vacuum equations~(\ref{eq:einstein}) around the Schwarzschild metric using this null framework. In other words, they write the unknown metric
\begin{equation*}
\boldsymbol{g}= g \ + \stackrel{_{\tiny (1)}}{g},
\end{equation*}
where $g$ is the Schwarzschild metric, and $\stackrel{_{\tiny (1)}}{g}$ is the variation of the metric. They plug this expression for $\boldsymbol{g}$ into Equation~(\ref{eq:einstein}), decompose in null components and eliminate the nonlinear terms in $\stackrel{_{(1)}}{g}$. They therefore obtain a suitable linearization of the Einstein equations around Schwarzschild.

The resulting equations are \emph{coupled}, meaning that in every equation we have more than one component of the field. In the following, we adopt the notation from \cite{linearized}. In particular, quantities with the superscript $^{(1)}$ correspond to the variation of the metric, whereas quantities without subscript or superscript correspond to their original Schwarzschild values.

Remarkably, the perturbed extreme components $\avar$ and $\abarvar$ were shown by Teukolsky in \cite{teukolsky} to satisfy \emph{decoupled} equations. In the notation of \cite{linearized}, the equation for $\avar$ is
\begin{equation}\label{eq:gteualpha}
\begin{aligned}
\snabla_4 \snabla_3 \avar &+ \left( \frac 1 2 \text{tr} \underline \chi + 2 \underline{\hat \omega}\right)\snabla_4 \avar+ \left( \frac 5 2\text{tr}\chi-\hat \omega \right) \snabla_3 \avar- \sdelta \avar\\
&+\avar(5 \underline{\hat \omega} \, \text{tr}\chi - \hat \omega \, \text{tr}\underline{\chi}-4 \rho +2r^{-2}+ \text{tr}\chi \text{tr}\underline{\chi}-4 \hat \omega \underline{\hat \omega}) = 0.
\end{aligned}
\end{equation}
On functions, we have the definition $\Omega \snabla_3 = \partial_u$, $\Omega \snabla_4 = \partial_v$. Here, $\Omega = \sqrt{1-2M/r}$. $\snabla$ indicates the induced connection on the spheres of constant $(u,v)$ Schwarzschild coordinates, and $\sdelta$ indicates the corresponding covariant Laplacian. Also, the definition of unbolded Ricci coefficients and unbolded null components is exactly as in (\ref{eq:riccicoeff}) and (\ref{eq:nullcomp}), replacing all the boldface quantities by the unbolded ones. Furthermore, $e_3 = \hat \lbar = \Omega^{-1} \partial_u$, $e_4 = \hat L = \Omega^{-1}\partial_v$.

The core of the proof is the following: starting from Equation~(\ref{eq:gteualpha}), the authors find a quantity $\Psivar$ which satisfies a ``good'' equation. Such quantity $\Psivar$ (and analogously its companion $\Psibarvar$) can be defined in terms of the sole extreme component $\avar$ (resp.~$\abarvar$), which satisfies a spin $\pm$2 Teukolsky equation. Here are some definitions
\begin{align*}
&\psivar := - \frac 1 2 r^{-1}\Omega^{-2} \snabla_{3}(r \Omega^2 \avar),
& \stackrel{_{(1)}}{P} := r^{-3}\Omega^{-1}\snabla_{3}(r^3 \Omega \psivar),\\
&\stackrel{_{(1)}}{\underline{\psi}} := \frac 1 2 r^{-1}\Omega^{-2} \snabla_{4}(r \Omega^2 \abarvar), 
&\stackrel{_{(1)}}{\underline{P}} := r^{-3}\Omega^{-1}\snabla_{4}(r^3 \Omega \stackrel{_{(1)}}{\underline \psi}).
\end{align*}
As usual, the superscript ${}^{(1)}$ indicates that the quantity is the one relative to the perturbed metric. As before, on functions, $\Omega \snabla_3 = \partial_u$, $\Omega \snabla_4 = \partial_v$.
We finally define the rescaled versions of $\pvar$ and $\pbarvar$:
\begin{equation*}
\boxed{\Psivar:= r^5 \pvar, \qquad \Psibarvar := r^5 \pbarvar.}
\end{equation*}
In this setting, $\Psivar$ satisfies the \textbf{Regge--Wheeler equation}:
\begin{equation}\label{eq:rwforpsi} \boxed{
\Omega \snabla_3 (\Omega \snabla_4 \Psivar) - (1-\mu)\slashed \Delta \Psivar + \left(\frac 4 {r^2}- \frac {6M}{r^3}\right)(1-\mu)\Psivar = 0.}
\end{equation}
Here, as before, $\Omega = \sqrt{1-2M/r}$. Also, on functions, $\Omega \snabla_3 = \partial_u$, $\Omega \snabla_4 = \partial_v$. The same equation is satisfied by $\Psibarvar$. Furthermore, $\mu = 2M/r$.
\begin{remark}
Most importantly, solutions to Equation~(\ref{eq:rwforpsi}) satisfy an energy conservation inequality and a Morawetz estimate.
\end{remark}
\begin{remark}\label{rmk:tens}
Equation~(\ref{eq:rwforpsi}) is a \emph{tensorial} equation. This equation is typically stated by other authors in the corresponding scalar form. See Remark 7.1 in \cite{linearized} for the form of the corresponding scalar equation.
\end{remark}
\begin{remark}
We further notice that $\Psivar$ can also be defined solely in terms of the middle components $\sigmavar$ and $\rhovar$ via angular derivation:
\begin{equation}\label{eq:pdefinition}
\Psivar_{AB} := r^5 \left( {\slashed{\mathcal{D}}}^\star_2 \duno (- \rhovar, \sigmavar)+ \frac 3 4 \rho \text{tr}\chi (\chivar - \chibarvar)\right).
\end{equation}
Here, the definition of the spherical operators $\duno$ and $\ddue$ is as follows:
\begin{equation}\label{eq:dunodduedef}
\duno(\rho, \sigma):= -\snabla_A \rho + \svol_{AB}	\snabla^B \sigma \qquad \ddue \xi := - \frac 1 2 (\snabla_B \xi_A + \snabla_A \xi_B -( \dive \xi) \gbar_{AB}).
\end{equation}
\end{remark}
\begin{remark}\label{rmk:gauge}
An additional difficulty, in the linearized gravity case, is the existence of \emph{pure gauge solutions}. These are solutions of the linearized Einstein equations arising from changes of coordinates which preserve the null structure of the metric. For the full formulation, we refer to the paper \cite{linearized}, especially Sections 2.1.4 and 6.1.
\end{remark}

\begin{remark}
There is a connection between the kernel of $\ddue \duno$ and solutions of the linearized gravity equations corresponding to ``infinitesimal perturbations towards Kerr'' of the Schwarzschild solution. As we shall see, this has an analogy in the Maxwell case.
\end{remark}

\subsection{Non-radiating modes}

Let us now turn our attention to the Maxwell system on the Schwarzschild spacetime. We remark that the Maxwell equations possess non-trivial stationary solutions, whose null components decay at spacelike infinity.

It is an easy computation to show that the following expression gives stationary solutions to the Maxwell system on Schwarzschild:
\begin{equation}\label{eq:statsol}
F = q_B r^{-2} \svol_{AB} + q_E r^{-2} \left(1-\frac{2M}{r}\right) \de t \wedge \de r^*.
\end{equation}
Here, $q_E$ and $q_B$ are two real parameters, respectively the ``electric charge'' and the ``magnetic charge''.

Excluding such ``stationary modes'' is a crucial element of every proof of decay of the Maxwell field on the Schwarzschild manifold.

Let us make a key remark: these stationary solutions have vanishing $\alpha$ and $\alphabar$ components. In other words, the extreme components do not ``see'' the stationary modes. Hence, we seek to define a quantity starting from $\alpha$ and $\alphabar$.

\subsection{Key to the proof for the spin $\pm 1$ Teukolsky Equations and the Maxwell system} 

Let us consider the Maxwell equations in a null frame, with the Schwarzschild metric as a background.

As was proved by Bardeen and Press in \cite{bardeen1973}\,(in its scalar, fixed-frequency version), the extreme components ($\alpha$ and $\alphabar$) satisfy the so-called \textbf{spin $\pm$1 Teukolsky equations}:
\begin{align}\label{eq:teuuno}
\snabla_\lbar \snabla_L (r \alpha_A) + \frac 2 r \left(1-\frac{3M}{r}\right) \snabla_\lbar (r \alpha_A) - (1-\mu)\slashed{\Delta}(r\alpha_A) + \frac{1-\mu}{r^2} r \alpha_A &= 0,
\\\label{eq:teudue}
\snabla_\lbar \snabla_L (r \alphabar_A) - \frac 2 r \left(1-\frac{3M}{r}\right) \snabla_L (r \alphabar_A) - (1-\mu)\slashed{\Delta}(r\alphabar_A) + \frac{1-\mu}{r^2} r \alphabar_A &= 0.
\end{align}
For the $\snabla$ notation, and for the definition of $\sdelta$, refer to Section~\ref{sec:notation}.

\begin{remark}
Notice that these are \emph{tensorial} equations, cf.~Remark~\ref{rmk:tens}.
\end{remark}

Here, as usual, $\mu = 2M/r$. These equations are badly behaved from the point of view of energy estimates, due to the first-order term. We now consider the quantities
\begin{equation}\label{eq:defmwquant}
\begin{aligned}
\phi_A &:= \frac{r^2}{1-\mu}\snabla_\lbar(r \alpha_A),\\
\phibar_A &:= \frac{r^2}{1-\mu}\snabla_L(r \alphabar_A).
\end{aligned}
\end{equation}
\begin{remark}
A fixed-frequency version of this transformation appeared in its scalar form in the work by Chandrasekhar \cite{chandrateu}. Similarly, a fixed-frequency version of the transformation relating Equations~(\ref{eq:gteualpha}) and (\ref{eq:rwforpsi}) appeared in the work of Chandrasekhar \cite{Chandraschw}.
\end{remark}
\begin{remark}\label{rm:angular}
Notice that, in view of the Maxwell equations, in the notation of the previous subsection,
$$
\phi_A = r^3 \duno(\rho, \sigma).
$$
\end{remark}
It can be shown by direct calculation from (\ref{eq:teuuno}) and (\ref{eq:teudue}) that $\phi_A$ and $\phibar_A$ both satisfy the following tensorial \textbf{\fackip equation}:
\begin{equation}\label{eq:rwphi}
\snabla_L \snabla_\lbar \phi - (1-\mu) \slashed{\Delta} \phi + V \phi = 0,
\end{equation}
with $V = \frac{1-\mu}{r^2}$.
\begin{remark}
Notice that this equation is analogous to (\ref{eq:rwforpsi}). As in the case of linearized gravity, this equation is also usually stated in the literature as a scalar equation. The scalar form of the \fackip equation for the unknown $u$ is the following:
\begin{equation}
(1-\mu)^{-1} L \lbar u - \sdelta u = 0.
\end{equation}
In particular, this is the wave equation satisfied by $r^2 \rho$ and by $r^2 \sigma$. Commuting the equation with the angular operator $r \snabla$ leads to the appearance of the additional zeroth order term. This mimics closely the case of linearized gravity, cf.~Remark 7.1 in \cite{linearized}, and is consistent with Remark~\ref{rm:angular}.
\end{remark}
Equation~(\ref{eq:rwphi}) is now \textbf{the key to the argument}: it admits robust energy estimates, and furthermore the quantity $\phi$ enables us to estimate the quantities $\alpha$ and $\alphabar$.

\begin{remark}\label{rmk:gauge2}
Let us notice here that the argument contained in this note (in particular, the proof of Theorem~\ref{prop:decayteu}) does not suffer from the difficulty arising from the ``pure gauge solutions'', which are present in the linearized gravity case \cite{linearized}.
\end{remark}

\section{Preliminaries and notation}\label{sec:notation}
Having settled the heuristics, we proceed to the actual setup.
	\begin{itemize}
	\item Let $g_{\mathbb{S}^2}$ be the standard metric on the sphere $\mathbb{S}^2$.
		\item 
		Let $\mathcal{S}_e$ be the following smooth Lorentzian manifold without boundary: $\mathcal{S}_e := (t, r, \omega) \in \R \times (2M, \infty) \times \mathbb{S}^2$. The metric tensor $g_e$ on $\mathcal{S}_e$ is defined as follows:
		\begin{equation*}
		g_e := - (1-\mu) \de t \otimes \de t + (1-\mu)^{-1} \de r \otimes \de r + r^2 g_{\mathbb{S}^2}.
		\end{equation*}
		Here,
		\begin{equation*}
		\mu := \frac{2M} r.
		\end{equation*}
		We call $\mathcal{S}_e$ the \textbf{open exterior Schwarzschild spacetime}, or simply \textbf{Schwarzschild exterior}. Note that this set does not have a boundary.

		\item When denoting subsets of $\mathcal{S}_e$ determined by some property, we shorten the notation in the following way:
		\begin{equation*}
		\{\text{property}\}:=\{(u,v,\omega) \in \mathcal{S}_e: \text{property}\}.
		\end{equation*}
		Hence, for instance, the set $\{(u,v,\omega) \in \mathcal{S}_e: u \geq u_0\}$ is denoted by $\{u \geq u_0\}$.
		\item Let $\nabla$ the Levi-Civita connection associated to $g_e$.
		\item Other coordinates:
		\begin{itemize}
		\item $(t, r_*, \omega)$, with $r_*:= r+2M\log(r-2M) -3M-2M \log M$,
		\item $(t^*, r_1, \omega)$ with $t^* := t + 2M \log(r-2M)$, $r_1 := r$,
		\item $(u,v,\omega)$ with $u := t-r_*$, $v := t+r_*$.
		\end{itemize} 
		\item Define a \textbf{local framefield}: $(L, \lbar, \partial_{\theta^A}, \partial_{\theta^B})$ such that 
		\begin{align*}
		L :&= \partial_t + \partial_{r_*},\\
		\lbar :&= \partial_t - \partial_{r_*},
		\end{align*}
		and such that $\partial_{\theta^A}$ and $\partial_{\theta^B}$ are local vectorfields induced by a system of local coordinates $(\theta^A, \theta^B)$ for $\mathbb{S}^2$.
		\item $\mathcal{S}_e$ embeds isometrically in $\mathcal{M}$, the maximally extended Schwarzschild spacetime. Let us call such isometric embedding $i: \mathcal{S}_e \to \mathcal{M}$. For details about the precise definition of $\mathcal M$ and the form of $i$, see \cite{lecturenotes}, Section 2.3.
		\item Consider the Kruskal coordinates $(T, R, \theta, \varphi)$ on $\mathcal{M}$, as in \cite{lecturenotes}, Section 2.3. We denote by $\mathfrak{V}$ the set of vectorfields on $T \mathcal{M}$:
		$$
		\mathfrak{V}:= \left\{\desude{}{T}, \desude{}{R}, \Omega_1, \Omega_2, \Omega_3\right\},
		$$
		where the $\{\Omega_i\}$ are rotation Killing fields, such that $\text{span}(\mathfrak{V}) = T\mathcal{M}$.
		\item We convene that a function $f$ is smooth on an open set $U \subset \mathcal{S}_e$ (denoted by $f \in \mathcal{C}^\infty(U)$) if there exists an open set $O \subset \mathcal{M}$, such that $O \supset \overline {i(U)}$ and there exists a smooth function $\tilde f \in \mathcal{C}^{\infty}(O)$ which restricts to $f$ on $i(U)$.
		\item Let again be $U \subset \mathcal{S}_e$ an open set, and let $\mathcal{V}$ be either $TU$ or a derived bundle of it (i.e. a tensor product of some copies of $T U$ with some copies of its dual). We say that a section $V$ of $\mathcal V$ is smooth if the following holds. We push $V$ forward via $i$ to obtain a section $V'$ of $\mathcal V'$, a derived bundle of $T \mathcal M$. We then express the components of $V'$ in the frame $\mathfrak{V}$ (or the corresponding derived frame), obtaining a collection of functions $(f_i)_{i=1,\ldots,5}:i(U)\to \R$, $n \in \N$. For $V$ to be smooth, we require that the $f_i$'s all be extendible to smooth functions $\bar f_i$ on an open set $O \supset \overline{i(U)}$, as in the previous bullet.
		We denote by $\Gamma(\mathcal V)$ the vectorspace of all smooth sections of such bundle.
		\begin{remark}
		Notice that this definition encodes the notion of ``smoothness up to the event horizon'', for instance, if $U = \mathcal{S}_e$, or if $U = \{t^* > a\}$, with $a > 0$.
		\end{remark}
		
		\begin{remark}
			$\lbar$ vanishes as $r \to 2M$, $\left(1-\frac{2M}{r}\right)^{-1} \lbar$ is a smooth vectorfield on the set $U:=\{t^* > a\}$, with $a \in \R$, according to our definition (it is a smooth section of $TU$).
		\end{remark}
		\item Let $k \in \N$. We denote by $\Lambda^k(\mathcal{V})$ the vectorspace of smooth antisymmetric $k$-forms, which is, the space of smooth sections of the bundle $\underbrace{\mathcal{V}^* \otimes \ldots \otimes \mathcal{V}^*}_{k \text{ times}}$, which are antisymmetric with respect to the permutation of any two arguments.
		\item We introduce tensorfields tangent to the spheres of constant $r$. 
		\begin{itemize}
		\item Let $S_{\tilde t,\tilde r} \subset \mathcal{S}_e$ be the set
		\begin{equation*}
		S_{\tilde t,\tilde r} := \{(t,r,\omega) \in \mathcal{S}_e, t = \tilde t, r = \tilde r\}.
		\end{equation*}
		\item Consider $TS_{\tilde t, \tilde r} \subset T\mathcal{S}_e$, and let $$\mathcal{B} := \bigcup_{\tilde t \in (-\infty, \infty), \tilde r \in (2M, \infty)} TS_{\tilde t, \tilde r} \subset T\mathcal{S}_e.$$
		Notice that $\mathcal{B}$ is the bundle tangent to each sphere of constant $t, r$.
		\item Sections of $\mathcal B$ can be seen as sections of $T\mathcal{S}_e$, due to the fact that $\mathcal B \subset T\mathcal{S}_e$. Hence, we say that a section $W$ of $\mathcal B$ (or a derived bundle thereof) is smooth if the corresponding section of $T \mathcal{S}_e$ (or a derived bundle thereof) is.
		\item If $\widetilde{ \mathcal B}$ is $\mathcal B$ or a derived bundle thereof, we denote by $\Gamma(\widetilde {\mathcal B })$ the vector space of all smooth sections of $\widetilde {\mathcal B}$. Similarly, $\Lambda^k(\widetilde{\mathcal{B}})$ is the space of alternating $k$-forms on $\widetilde B$.
		\item Indices for tensors in $T \mathcal{S}_e$ and derived bundles will be indicated by Greek letters $\mu, \nu, \kappa \ldots$ Indices for tensors in $\mathcal{B}$ will be indicated by uppercase Latin letters: $A, B, C, \ldots$.
		\item Let $\gbar$ be the induced metric on spheres of constant $r$. Technically, this is a smooth section $\gbar \in \Gamma(\mathcal{B}^* \otimes \mathcal{B}^*)$. On each sphere $S_{t,r}$, $\gbar$ is the round metric.
		\item Let $\svol_{AB} \in \Lambda^2 (\mathcal B)$ be the induced volume form on the spheres $S_{\tilde t, \tilde r}$.
		\item Let $(\cdot )^\perp: T\mathcal{S}_e \to \mathcal{B}$ be the orthogonal projection on the spheres $S_{\tilde t, \tilde r}$.
		
		\item Let $V, W \in \Gamma(\mathcal{B})$. We define a connection on $\mathcal{B}$ by
		\begin{align}
		\snabla_V W := (\nabla_V W)^\perp.
		\end{align}
		This connection coincides, on the spheres, with the Levi--Civita connection induced by the induced metric $\gbar$.
		\item We define two other differential operators on $\Gamma(\mathcal{B})$ in the following way:
		\begin{align*}
		\snabla_L V := (\nabla_L V)^\perp, \qquad \snabla_\lbar V := (\nabla_\lbar V)^\perp.
		\end{align*}
		\item The previous differential operators can be extended to derived bundles from $\mathcal B$ in the usual way, asking that they satisfy the Leibnitz rule.
		\item We define the induced covariant curl and divergence in the following way. Let $\omega \in \Gamma(\mathcal{B}^*)$:
		\begin{align*}
		\dive \omega := \gbar^{AB} \snabla_A \omega_B, \qquad 
		\curl \omega := \svol^{AB} \snabla_A \omega_B.
		\end{align*}
		\end{itemize}
		\item We introduce the foliation needed in the note, and relevant Sobolev norms on it.
		\begin{itemize}
		\item Let $\leo:=\{\Omega_1, \Omega_2, \Omega_3\}$ be a set of angular Killing fields of $\mathcal{S}_e$ whose elements, at each point of $\mathcal{S}_e$, span all directions in $\mathcal{B}$.
		Let $\tilde \leo$ be the renormalized version
		\begin{equation*}
		\tilde \leo := \{\Omega_1/r, \Omega_2/r, \Omega_3/r\}.
		\end{equation*}
		\item Let $k \geq 0$, let $\iota^{\tilde \Omega}_k$ (resp. $\iota^{\tilde \Omega}_{\leq k}$) be the set of all ordered lists of length $k$ (resp. $\leq k$) composed of elements of $\tilde \leo$, and analogously let $\iota^{\Omega}_k$ (resp. $\iota^{\tilde \Omega}_{\leq k}$) be the set of all ordered lists of length $k$ (resp. $\leq k$) composed of elements of $\leo$. Elements in $\iota^{\tilde \Omega}_k$ and $\iota^{\Omega}_k$ will be referred to as \textbf{multi-indices}.
		\item Let $\eta$ be a covariant tensorfield on $\mathcal{B}$. If $J = (V_1, \ldots, V_k)$ is a multi-index, and let $\slie$ be the Lie derivative induced by the connection $\snabla$. Let $X \in \Gamma(\mathcal{B})$. We let 
		\begin{equation}\label{eq:replie}
		\begin{aligned}
		\snabla^J \eta := \snabla_{V_1} \cdots \snabla_{V_k} \eta, \hspace{20pt}
		\slie^J \eta := \slie_{V_1} \cdots \slie_{V_k} \eta, \hspace{20pt}
		(\snabla_X)^k \eta := \underbrace{\snabla_X \cdots \snabla_X }_{k-\text{times}}\eta.
		\end{aligned}
		\end{equation}
		\item Let $n \geq 0$ be an integer, $\eta \in \Gamma((\mathcal{B}^*)^n)$. We define the angular norm of $\eta$ as
		\begin{equation*}
		|\eta|(t,r,\omega) :=  \sum_{J \in \iota^{\tilde \Omega}_n}|\eta(J)|,
		\end{equation*}
		where we assumed, if $J = (V_1, \ldots, V_n)$,
		\begin{equation*}
		\eta(J) = \eta(V_1, \ldots, V_n).
		\end{equation*}
		\item Given $\tilde u, \tilde v \in \R$, we define
		\begin{align*}
		&\conplus_{\tilde u, \tilde v} := \{(u, v, \omega)\in \mathcal{S}_e, u=\tilde u, v \geq \tilde v \},	
		\hspace{20pt} \conminus_{\tilde u, \tilde v} := \{(u, v, \omega)\in \mathcal{S}_e, u\geq \tilde u, v = \tilde v \},\\
		& C_{\tilde u, \tilde v} := \conplus_{\tilde u, \tilde v} \cup \conminus_{\tilde u, \tilde v},  \\
		& \conplus_{\tilde u} := \{(u, v, \omega)\in \mathcal{S}_e, u= \tilde u\}, \hspace{62pt}
		\conminus_{\tilde v} := \{(u, v, \omega)\in \mathcal{S}_e, v = \tilde v\}.
		\end{align*}
		\item Let $i$ be the inclusion of $C_{u_0, v_0}$ into $\mathcal{S}_e$. Let $\eta$ be a covariant section of $i^* \mathcal{B}$ or one of its derived bundles. We define the following fluxes
		\begin{align}\label{eq:fdefu1}
			\fara^T_u[\eta](v_1, v_2) &:= \int_{v_1}^{v_2} \int_{\mathbb{S}^2} [|\snabla_L \eta|^2+(1-\mu)|\snabla \eta|^2+V|\eta|^2](u,v,\omega) \de v  \desphere(\omega) ,\\ \label{eq:fdefv1}
			\fara^T_v[\eta](u_1, u_2) &:= \int_{u_1}^{u_2}\int_{\mathbb{S}^2} [|\snabla_\lbar \eta|^2+(1-\mu)|\snabla \eta|^2+V|\eta|^2](u,v,\omega) \de u  \desphere(\omega),\\
			\fara^N_v[\eta](u_1, u_2) &:= \\
			& \int_{u_1}^{u_2}\int_{\mathbb{S}^2} [(1-\mu)^{-1}|\snabla_\lbar \eta|^2+(1-\mu)|\snabla \eta|^2+V|\eta|^2](u,v,\omega) \de u  \desphere(\omega),\nonumber\\
			 \fara^\infty[\eta] (u_1, v_1) &:= \fara^T_{u_1} [\eta] (v_1, \infty) + \fara^N_{v_1} [\eta] (u_1, \infty).
		\end{align}
		
		\item Let $q,x,s \in \N_{\geq 0}$.
		Let $\eta$ as above.
		We define the weighted Sobolev norms
		\begin{align}\label{eq:defsobo}
		\snorm{\eta}{q}{x}{s} &:= \\
			 &\sum_{i =0}^x \sum_{J \in \iota_{\leq s}^\Omega} \left\{
			 \fara^\infty[(\snabla_T)^i \slie^J \eta](P(u_0)) \phantom{\int} \right. \nonumber\\
			 &\hspace{60pt}\left. + \int_{\conplus_{u_0} \cap \{r \geq R \}} r^q |\snabla_L (\snabla_T)^{\min\{i, (x-1)^+ \}} \slie^J \eta|^2\de v \desphere \right\}. \nonumber
		\end{align}
	
		\end{itemize}
\end{itemize}
\subsection{The Maxwell system}
Let $F$ be an antisymmetric 2-form on $\mathcal{S}_e$. Let us introduce the Maxwell Equations:
\begin{equation}\label{eq:mweffe}
d F = 0, \ \ d \star F = 0.
\end{equation}
Here, $\star$ denotes the Hodge dual operator. More explicitly, if $G$ is a two-form,
\begin{equation}\label{eq:hodgedual}
(\star G)_{\mu\nu} = \frac 1 2 \varepsilon_{\alpha \beta \gamma \delta} G^{\gamma \delta}.
\end{equation}
Equivalently, the system can be written as
\begin{equation*}
\nabla_{[\mu} F_{\kappa \lambda]} = 0, \qquad \nabla^\mu F_{\mu\nu} = 0.
\end{equation*}
Here, square brackets denote antisymmetrization of indices.

\subsection{The null decomposition of the Maxwell system}
\begin{definition}
Let $F \in \Lambda^2(\mathcal{S}_e)$. We define $\alpha, \alphabar \in \Gamma(\mathcal{B}^*)$, and $\rho, \sigma \in \mathcal{C}^\infty(\mathcal{S}_e)$ by the following relations:
\begin{equation}\label{eq:nulldecdef}
	\begin{aligned}
		&\alpha(V) := F(V,L), \\
		&\alphabar(V) := F(V, \lbar),\\
		&\rho := \frac 1 2 \left(1-\frac{2M}{r} \right)^{-1}F(\lbar,L), \\
		&\sigma := \frac 1 2 \svol^{CD} F_{CD}.
	\end{aligned}
\end{equation}
for all $V \in \Gamma(\mathcal{B})$.
\end{definition}
\begin{remark}
$\alpha$ and $\alphabar$ can also be viewed as one-forms in $\Gamma(T^*\mathcal{S}_e)$, by requiring that they vanish on $L$ and $\lbar$. Furthermore, in the definition of $\sigma$, we consider $F$ as an element of $\Lambda^2(\mathcal{B})$, by restriction.
\end{remark}
\begin{remark}
Note that all the previously defined quantities are smooth on $\mathcal{S}_e$ (up to the horizon), in the sense of our definition. This is due to the fact that the vectorfield $(1-\mu)^{-1}\lbar$ is smooth on $\mathcal{S}_e$, according to our definition.
\end{remark}
Having introduced these quantities, we write the Maxwell system with respect to them. We have the following proposition.
\begin{proposition}\label{prop:mwnull}
Let $F \in \Lambda^2(\mathcal{S}_e)$, and let $F$ satisfy the Maxwell system (\ref{eq:mweffe}) on $\mathcal{S}_e$. Then, defining the objects $\alpha, \alphabar, \rho, \sigma$ as in (\ref{eq:nulldecdef}) we have that
\begin{align}
 \frac 1 r \snabla_L (r\alphabar_A) + (1-\mu) (\snabla_A \rho - \svol_{AB} \snabla^B \sigma) = 0,\label{mw1}\\
 \frac 1 r \snabla_\lbar (r\alpha_A) - (1-\mu) (\snabla_A \rho + \svol_{AB} \snabla^B \sigma) = 0, \label{mw2}\\
 \curl \alphabar - 2 \frac{1-\mu}{r} \sigma + \snabla_\lbar \sigma = 0, \label{mw3}\\
 -\dive \alphabar + 2 \frac{1-\mu}{r} \rho - \snabla_\lbar \rho = 0,\label{mw4}\\
 \curl \alpha + 2 \frac{1-\mu}{r} \sigma + \snabla_L \sigma = 0,\label{mw5}\\
 \dive \alpha - 2 \frac{1-\mu}{r} \rho - \snabla_L\rho=0. \label{mw6}
\end{align}
Furthermore, the extreme components $\alpha$ and $\alphabar$ satisfy the \textbf{spin $\pm$1 Teukolsky equations}:
\begin{align}\label{eq:teua}
\snabla_\lbar \snabla_L (r \alpha_A) + \frac 2 r \left(1-\frac{3M}{r}\right) \snabla_\lbar (r \alpha_A) - (1-\mu)\slashed{\Delta}(r\alpha_A) + \frac{1-\mu}{r^2} r \alpha_A &= 0,
\\\label{eq:teuabar}
\snabla_\lbar \snabla_L (r \alphabar_A) - \frac 2 r \left(1-\frac{3M}{r}\right) \snabla_L (r \alphabar_A) - (1-\mu)\slashed{\Delta}(r\alphabar_A) + \frac{1-\mu}{r^2} r \alphabar_A &= 0.
\end{align}
\end{proposition}

\begin{proof}[Proof of Proposition~\ref{prop:mwnull}]
We postpone the relevant calculations to the Appendix, Section~\ref{sec:dernulldec}.
\end{proof}

\subsection{Derivation of the \fackip equation} \label{subsec:fackip}
We now proceed to introduce the crucial quantities $\phi$ and $\phibar$, and we prove that, if we only require the spin $\pm1$ Teukolsky equations to hold for $\alpha$ and $\alphabar$, then $\phi$ and $\phibar$ satisfy the so-called \fackip equation.
\begin{proposition}\label{prop:rw}
Let $\alpha$ satisfy the spin $+1$ Teukolsky equation (\ref{eq:teua}) and let $\alphabar$ satisfy the $-1$ Teukolsky equation (\ref{eq:teuabar}) on $\mathcal{S}_e$. Then, we define
\begin{equation}\label{eq:quantity}
\begin{aligned}
\phi_A &:= \frac{r^2}{1-\mu}\snabla_\lbar(r \alpha_A),\\
\phibar_A &:= \frac{r^2}{1-\mu}\snabla_L(r \alphabar_A).
\end{aligned}
\end{equation}
Under these hypotheses, $\phi$ and $\underline \phi$ satisfy the \fackip Equation:
\begin{align} \label{eq:rw1}
\snabla_\lbar \snabla_L \phi_A - (1-\mu)\slashed{\Delta}(\phi_A) + \frac{1-\mu}{r^2} \phi_A &= 0,
\\
\label{eq:rw2}
\snabla_\lbar \snabla_L \phibar_A - (1-\mu)\slashed{\Delta}(\phibar_A) + \frac{1-\mu}{r^2} \phibar_A &= 0.
\end{align}
\end{proposition}

\begin{remark}\label{rm:secret}
We remark that, if we further assume that $\alpha$ and $\alphabar$ are part of a solution $(\rho, \sigma, \alpha, \alphabar)$ of the Maxwell equations (\ref{mw1}) -- (\ref{mw6}) on $\mathcal{S}_e$, the following relations hold true:
\begin{align*}
\phi_A = r^3(\snabla_A \rho + \svol_{AB} \snabla^B \sigma), \qquad
\phibar_A = r^3(-\snabla_A \rho + \svol_{AB} \snabla^B \sigma).
\end{align*}
We also remark that, in this case, the tensorial \fackip Equation can be obtained from the wave equation (scalar Fackerell--Ipser) satisfied by the middle components, commuting with the projected covariant angular derivative $\snabla_A$.
\end{remark}

\begin{proof}[Proof of Proposition~\ref{prop:rw}]
It is a straightforward calculation from the Teukolsky Equation. We restrict to $\phibar$, the reasoning for $\phi$ being analogous. First of all, we notice that the Teukolsky Equation for $\alphabar$ is equivalent to
\begin{equation}\label{eq:forrw1}
\begin{aligned}
\frac{1-\mu}{r^2}\snabla_\lbar\left( \frac{r^2}{1-\mu}\snabla_L (r \alphabar_A)\right)  - (1-\mu)\slashed{\Delta}(r\alphabar_A) + \frac{1-\mu}{r^2} r \alphabar_A = 0.
\end{aligned}
\end{equation}
For,
\begin{equation*}
\lbar \left(\frac{r^2}{1-\mu} \right) =- \frac{r^2}{1-\mu} \frac{2}{r}\left(1-\frac{3M}r\right).
\end{equation*}
Multiply Equation~(\ref{eq:forrw1}) by $\frac{r^2}{1-\mu}$ and subsequently take the $\snabla_L$ derivative of both sides. We obtain, since $[\snabla_L , r^2 \slashed \Delta] = 0$,
\begin{equation*}
\snabla_L \snabla_\lbar \left( \frac{r^2}{1-\mu}\snabla_L (r \alphabar_A)\right) - r^2 \sdelta \snabla_L (r\alphabar_A)+ \snabla_L (r \alphabar_A) = 0.
\end{equation*}
This implies the claim.
\end{proof}

\section{Statements of the main results}\label{sec:results}
In this section, we state the main results of the present note. The first result, Theorem~\ref{prop:decayteu}, only deals with solutions to the spin $\pm 1$ Teukolsky Equations, and provides decay rates for them. The second result, Theorem~\ref{prop:decaymax}, concerns a solution $F$ of the full Maxwell system, and derives decay bounds for the relevant quantities using Theorem~\ref{prop:decayteu}.
\begin{theorem}[Decay for solutions to the spin $\pm$1 Teukolsky equations]\label{prop:decayteu}
There exist a positive real number $R_* > 0$ and a positive constant $C$ depending only on $M$ and $R_*$ such that, letting $(u_0, v_0)$ be real numbers such that $v_0 - u_0 = 2 R_*$, we have the following. Let $\alpha, \alphabar \in \Gamma(\mathcal{B}^*)$ be solutions to the spin $\pm1$ Teukolsky equations on $[u_0, \infty) \times [v_0, \infty) \times \mathbb{S}^2 \subset \mathcal{S}_e$:
\begin{align} 
\snabla_\lbar \snabla_L (r \alpha_A) + \frac 2 r \left(1-\frac{3M}{r}\right) \snabla_\lbar (r \alpha_A) - (1-\mu)\slashed{\Delta}(r\alpha_A) + \frac{1-\mu}{r^2} r \alpha_A &= 0,
\\
\snabla_\lbar \snabla_L (r \alphabar_A) - \frac 2 r \left(1-\frac{3M}{r}\right) \snabla_L (r \alphabar_A) - (1-\mu)\slashed{\Delta}(r\alphabar_A) + \frac{1-\mu}{r^2} r \alphabar_A &= 0.
\end{align}
Let $\phi$, $\phibar$ be the related quantities as in (\ref{eq:quantity}).
Under these assumptions, $\phi$ and $\phibar$ satisfy the Morawetz estimate (\ref{eq:fullen}) of Lemma~\ref{lem:moraw}, as well as hierarchy of integrated estimates (\ref{pdue}) -- (\ref{puno}).

Furthermore, let $\chi$ be a smooth cutoff function such that $\chi(r) =1$ for $r \geq 3M$, and $\chi(r) = 0$ for $r \in [2M, 3/2M]$. Let $\widetilde \alphabar = (1-\mu)^{-1}\alphabar$ and $\overline \Psi := \chi(r) (1-\mu)^{-1}r^3 \alpha$.

Under these conditions, we have the pointwise estimates:
\begin{align}\label{eq:da1}
|\alpha| &\leq C \frac{\snorm{\overline \Psi}{0}{0}{2}+ \snorm{\phi}{2}{2}{1}}{v} & \text{and}\\ 
\label{eq:dab1}
|\widetilde \alphabar| &\leq C \frac{\snorm{\widetilde{\alphabar}}{0}{0}{0} + \snorm{\phibar}{2}{2}{1} }{v} &\text{ on }\{r_* \leq R_*\} \cap \{u \geq u_0\} \cap \{v \geq v_0\}. \\ \label{eq:da3}
|\alpha| &\leq C \frac{\snorm{\overline \Psi}{0}{0}{2}+ \snorm{\phi}{2}{q}{1}}{ v^{q/2} r^{3/2}}
&\text{for } q \in \{0,1,2\}, \text{ and}\\\label{eq:da4}
|\alpha| &\leq C \frac{\snorm{\overline \Psi}{2}{0}{2}+\snorm{\phi}{2}{2}{2} }{(|u|+1)^{\frac 1 2} r^3}, &\text{ and}\\
\label{eq:dab2}
|\alphabar| &\leq C \frac{\snorm{\widetilde \alphabar}{0}{0}{0} + \snorm{\phibar}{2}{2}{1}}{(|u|+1)r }
 &\text{ on }  \{r_* \geq R_*\}  \cap \{u \geq u_0\} \cap \{v \geq v_0\}.
\end{align}
Here, we used the definition of norm in (\ref{eq:defsobo}).
\end{theorem}

\begin{theorem}[Decay for solutions to the Maxwell system on Schwarzschild]\label{prop:decaymax}
There exist a positive real number $R_* > 0$ and a positive constant $C$ depending only on $M$ and $R_*$ such that, letting $(u_0, v_0)$ be real numbers such that $v_0 - u_0 = 2 R_*$, we have the following. Let $F \in \Lambda^2(T\mathcal{S}_e)$ be a solution to the Maxwell system on $\{u \geq u_0\} \cap \{v \geq v_0\}$:
\begin{equation*}
\de F = 0, \qquad \de \star F = 0.
\end{equation*}
Recall the definition of the null components $\alpha, \alphabar, \rho, \sigma$ (\ref{eq:nulldecdef}).
We let
\begin{equation}\label{eq:rhosdef}
\rho_s := \frac {R^2} {4 \pi} \int_{\omega \in \mathbb{S}^2} \rho(u_0, v_0, \omega) \desphere(\omega), \qquad \sigma_s := \frac {R^2} {4 \pi} \int_{\omega \in \mathbb{S}^2} \sigma(u_0, v_0, \omega) \desphere(\omega).
\end{equation}
Recall the definition of $\phi$ and $\phibar$ from (\ref{eq:quantity}). Let $\chi$ be a smooth cutoff function such that $\chi(r) =1$ for $r \geq 3M$, and $\chi(r) = 0$ for $r \in [2M, 3/2M]$.
Let us furthermore set $\widetilde \alphabar = (1-\mu)^{-1}\alphabar$, $\overline \Psi := \chi(r) (1-\mu)^{-1}r^3 \alpha$, and
\begin{align}
M_\alpha &:= \snorm{\overline \Psi}{2}{0}{2}+\snorm{\phi}{2}{2}{1}, \\
\label{eq:mrhosdef}
M_{\rho, \sigma} &:= \norm{\phi}_{C_{u_0, v_0}; 2; 2,1} + \norm{\phibar}_{C_{u_0, v_0}; 2; 2,1},\\
M_\alphabar &:=\snorm{\widetilde \alphabar}{0}{0}{0} + \snorm{\phibar}{2}{2}{1}.
\end{align}
Then, we have
\begin{align}
|\alpha|, (1-\mu)^{-1}|\alphabar| &\leq C v^{-1} (M_\alpha+ M_\alphabar),& 
\\ 
\left|\rho-\frac{\rho_s}{r^2}\right|, \left|\sigma-\frac{\sigma_s}{r^2}\right| &\leq C v^{-1} M_{\rho,\sigma}&\text{on } \{r_* \leq R_*\} \cap \{u \geq u_0\} \cap \{v \geq v_0\}, \label{eq:rsnear}\\
|\alpha| &\leq C(|u|+1)^{-\frac 1 2} r^{-3} M_\alpha, & 
\\
|\alphabar| &\leq C (|u|+1)^{-1} r^{-1}M_\alphabar, & 
\\
\left|\rho-\frac{\rho_s}{r^2}\right|, \left|\sigma-\frac{\sigma_s}{r^2}\right| &\leq C (|u|+1)^{-\frac 1 2}r^{-2}M_{\rho,\sigma},
\\
\left|\rho-\frac{\rho_s}{r^2}\right|, \left|\sigma-\frac{\sigma_s}{r^2}\right| &\leq C (|u|+1)^{- 1 }r^{-3/2}M_{\rho,\sigma} &\text{on } \{r_* \geq R_*\} \cap \{u \geq u_0\} \cap \{v \geq v_0\}.
\end{align}
\end{theorem}

A few remarks are in order.

\begin{remark}
We will not delve into the issue of optimal well-posedness statements for the Maxwell system or for the spin $\pm$1 Teukolsky equations here. Let us just remark that, in the smooth category, well-posedness for the characteristic initial value problem follows, in both situations, from ideas contained in the work by Rendall \cite{rendallchar}.
\end{remark}

\begin{remark} We notice that Theorem~\ref{prop:decaymax} gives the decay rate $v^{-1}$ for all components of the field in the region $\{r \leq R\}$. Furthermore, for all components of the field, we have the \textbf{uniform peeling estimates}, on a fixed outgoing null cone $\conplus_{\tilde u} \cap \{r \geq R\}$:
\begin{equation}
\begin{aligned}
|\alpha| &\lesssim r^{-3} (|\tilde u|+1)^{-\frac 1 2}, \\
|\rho|, |\sigma| &\lesssim r^{-2} (|\tilde u|+1)^{-\frac 1 2}, \\
|\alphabar| &\lesssim r^{-1} (|\tilde u|+1)^{-1}. 
\end{aligned}
\end{equation}
Here, we supposed for simplicity that $M_{\rho, \sigma}$, $M_\alpha$, $M_{\alphabar}$ all be finite, and that $\rho_s = \sigma_s = 0$. The bound for $\alpha$ is the stronger one, corresponding to inequality (\ref{eq:da4}).

We underline the difference between estimates (\ref{eq:da3}) and (\ref{eq:da4}). In the former, we require less of the initial data to obtain a lower decay rate. In the latter, we have a larger weight on the $L$-derivative of $\Psi$, and we obtain a uniform peeling estimate for $\alpha$. Weaker requirements on initial data, though implying weaker decay, may be useful for applications to nonlinear problems, in view of a bootstrap argument. An example is the original proof of the nonlinear stability of the Minkowski spacetime by Christodoulou and Klainerman \cite{globalnon}, in which the authors do not need optimal decay rates in order to close the argument. In fact, the failure of peeling to hold has a physical interpretation \cite{grossmann}.

\end{remark}

\begin{remark}[On initial data]\label{rem:indata}
We remark that, in order to solve the full Maxwell system (\ref{mw1}) -- (\ref{mw6}), it is enough to impose initial data for $\alpha$ and $\alphabar$ on the set $\mathcal{C}_{u_0, v_0}$. For, then, all first derivatives of $\alpha$ and $\alphabar$ along $\mathcal{C}_{u_0, v_0}$ can be recovered via the spin $\pm 1$ Teukolsky Equations. Then, we can solve for $\alpha$ and $\alphabar$ in $\{u \geq u_0\} \cap \{v \geq v_0\}$, again from the Teukolsky Equations. Finally, we can use relations (\ref{mw1}), (\ref{mw2}) to recover all angular derivatives of $\rho$ and $\sigma$. This defines uniquely a solution up to the stationary solutions~(\ref{eq:statsol}). The resulting quantities $\alpha$, $\alphabar$, $\sigma$, $\rho$ then satisfy the full Maxwell system (\ref{mw1}) -- (\ref{mw6}).
\end{remark}

\begin{remark}
For notational convenience, the norms we defined in Equation~(\ref{eq:defsobo}) are not intrinsic to the surface $\mathcal{C}_{u_0, v_0}$. Nevertheless, using the \fackip equation for $\phi$ and $\phibar$, it can be shown that 
\begin{equation}\label{eq:inittransl}
\begin{aligned}
\snorm{\phi}{2}{2}{0}^2 &\leq \int_{\conplus_{u_0} \cap \{r\geq R\}}\sum_{I \in \iota^+_{\leq 3}}(|\snabla_L \snabla^I \phi|^2 + |\snabla \snabla^I \phi|^2 + r^{-2}|\snabla^I \phi|^2) \de v \desphere\\
&+\int_{\conplus_{u_0} \cap \{r \geq R\}}(r^2 |\snabla_L \snabla_L \phi|^2 + r^2 |\snabla_L \phi|^2 ) \de v \desphere\\
&+ \int_{\conminus_{u_0} \cap \{r \leq R\}}\sum_{J \in \iota^-_{\leq 4}}|\snabla^J \phi| (1-\mu)\de u \desphere.
\end{aligned}
\end{equation}
Here, $\iota^{+}_{\leq 2}$ (resp. $\iota^{-}_{\leq 2}$) is the set of all ordered lists of length $\leq 2$ composed of elements of $\leo \cup \{L\}$ (resp. $\leo \cup \{(1-\mu)^{-1}\lbar\}$).

Estimate (\ref{eq:inittransl}) implies in particular that the norm $\snorm{\phi}{2}{2}{0}^2$ can be controlled in terms of a norm intrinsic to the surface $\mathcal{C}_{u_0, v_0}$. Similar expressions hold for $\alpha$ and $\alphabar$.

Furthermore notice that estimate (\ref{eq:inittransl}) ``loses derivatives''. On the left hand side, the norm $\snorm{\phi}{2}{2}{0}^2$ depends on $3$ ``unweighted'' derivatives, and on $2$ ``weighted'' derivatives. The norm on the right hand side of (\ref{eq:inittransl}), on the other hand, depends on $4$ ``unweighted'' derivatives and $2$ ``weighted'' derivatives.
\end{remark}

\begin{remark}[On the propagation of decay from initial data]
Let $\omega$ be a non-trivial $1$-form on $\mathbb{S}^2$. Suppose for ease of exposition that $\sdelta_{\mathbb{S}^2}\omega = 2 \omega$. Let $f_1(r)$ be a smooth function of $r$. Following Remark \ref{rem:indata}, let us set initial data for $\alpha$ on $\mathcal{C}_{u_0, v_0}$ in the following way:
\begin{equation}
\alpha = f_1(r) \omega \qquad \text{on } \mathcal{C}_{u_0, v_0}.
\end{equation}
We then use the relation (spin $+1$ Teukolsky equation) to induce data for $\phi$:
\begin{equation*}
\begin{aligned}
\snabla_L \phi_A &= r^2 \sdelta (r \alpha_A) - r \alpha_A.
\end{aligned}
\end{equation*}
From the latter, it follows that initial data for $\phi$ satisfies
\begin{equation}\label{eq:phiinit}
\phi(u_0,v) = \left(\int_{v_0}^v r(u_0, \tilde v) f_1(r(u_0, \tilde v))  \de \tilde v \right)\omega + \eta,
\end{equation}
where $\eta$ is a fixed one-form on $\mathbb{S}^2$. Here, $r(u_0, \tilde v)$ denotes the $r$-coordinate of the point $(u_0, \tilde v)$ in $(u,v)$-coordinates. Let us denote 
\begin{equation*}
f_2(\bar r) := \int_{v_0}^{v} r(u_0, \tilde v) f_1(r(u_0, \tilde v))  \de \tilde v,
\end{equation*}
whenever the $r$-coordinate of the point $(u_0, v)$ is $\bar r$.

Now, let $s \in (\frac 3 2, 1)$, and let us suppose the following on the function $f_1$:
\begin{equation}\label{eq:condalpha}
|f_1(r)| \sim r^{-1-s}, \qquad |f_1'(r)| \sim r^{-2-s}, \qquad |f_1''(r)| \sim r^{-3-s},
\end{equation}
as $r \to \infty$.

It then follows that
\begin{equation}\label{eq:condf}
|f_2(r)| \sim 1, \qquad |f_2'(r)|\sim r^{-s}, \qquad |f_2''(r)| \sim r^{-s-1}, \qquad  |f_2'''(r)| \sim r^{-s-2},
\end{equation} 
Now, from (\ref{eq:inittransl}) and the form of $\phi$ (\ref{eq:phiinit}), we obtain
\begin{equation*}
\snorm{\phi}{2}{2}{1} < \infty.
\end{equation*}
Under conditions (\ref{eq:condalpha}), we furthermore have that, recalling $\overline \Psi_A = \chi(r) r^3(1-\mu)^{-1}\alpha_A$, with $\chi(r)$ smooth supported away from $r = 2M$, such that $\chi(r) = 1$ for $r \geq 3M$,
\begin{equation*}
 \snorm{\overline \Psi}{0}{0}{2} < \infty.
\end{equation*}
In this case, the norm on the right hand side of estimate (\ref{eq:da3}) with $q = 2$ is bounded, and we obtain the bound $|\alpha|\leq C r^{-\frac 5 2}$ for $\alpha$ along any fixed outgoing cone, whereas we supposed that $|\alpha|$ is asymptotic to $r^{-1-s}$ on $\mathcal{C}_{u_0, v_0}$, with $s + 1 > \frac 5 2$. In this case, we do not recover the initial decay.

On the other hand, if we impose $s \geq 1$, we obtain
\begin{equation*}
 \snorm{\overline \Psi}{2}{0}{2} < \infty,
\end{equation*}
and the right hand side of estimate (\ref{eq:da4}) is finite. We then have $|\alpha| \leq C r^{-3}$ for some constant $C$, along a fixed outgoing null cone.

In particular, for $s=1$, we are able to propagate the $r^{-3}$ initial decay. Similar statements hold for $\alphabar$, $\rho$, $\sigma$.

We finally remark that, if we were to assume a sharper decay than $r^{-3}$ for $\alpha$ on the initial cone $\mathcal{C}_{u_0,v_0}$, generically, it would not propagate.
\end{remark}

\section{Estimates on the \fackip Equation}\label{sec:rpmethod}
In this section, we prove integrated decay estimates for solutions to the \fackip equation. The estimates and the methods to obtain them are very similar to those in \cite{linearized}.
The results contained in this section are of independent interest, and the section can be read independently from the rest of the paper, starting from the assumption that $\phi$ only satisfies the \fackip equation.

We do not prove pointwise decay for $\phi$, as it clearly follows from the ideas in the proof of Theorem~\ref{prop:decaymax}, cf.~Remark~\ref{rm:secret}.

Let us now proceed to the setup. Let $v_2 \geq v_1 \geq v_0$, and $u_2 \geq u_1 \geq u_0$. Let $V = (1-\mu)/r^2$. Recall the definition of the null fluxes and of the Sobolev norms:
\begin{align}\label{eq:fdefu}
	\fara^T_u[\phi](v_1, v_2) &:= \int_{v_1}^{v_2} \int_{\mathbb{S}^2} [|\snabla_L \phi|^2+(1-\mu)|\snabla \phi|^2+V|\phi|^2](u,v,\omega) \de v  \desphere(\omega) ,\\ \label{eq:fdefv}
	\fara^T_v[\phi](u_1, u_2) &:= \int_{u_1}^{u_2}\int_{\mathbb{S}^2} [|\snabla_\lbar \phi|^2+(1-\mu)|\snabla \phi|^2+V|\phi|^2](u,v,\omega) \de u  \desphere(\omega),\\
	\fara^N_v[\phi](u_1, u_2) &:= \\
	& \int_{u_1}^{u_2}\int_{\mathbb{S}^2} [(1-\mu)^{-1}|\snabla_\lbar \phi|^2+(1-\mu)|\snabla \phi|^2+V|\phi|^2](u,v,\omega) \de u  \desphere(\omega),\nonumber\\
	 \fara^\infty[\phi] (u_1, v_1) &:= \fara^T_{u_1} [\phi] (v_1, \infty) + \fara^N_{v_1} [\phi] (u_1, \infty)\\
	\snorm{\phi}{q}{x}{s} &:= \\
	 	 &\sum_{i =0}^x \sum_{J \in \iota_{\leq s}^\Omega} \left\{
	 	 \fara^\infty[(\snabla_T)^i \slie^J \phi](P(u_0)) \phantom{\int} \right. \nonumber\\
	 	 &\hspace{60pt}\left. + \int_{\conplus_{u_0} \cap \{r \geq R \}} r^q |\snabla_L (\snabla_T)^{\min\{i, (x-1)^+ \}} \slie^J \phi|^2\de v \desphere \right\}. \nonumber
	\end{align}
Here, $P(u_0) = (u_0, 2R_* + u_0)$.

Let $R_* > 0$. We define the spacetime regions:
\begin{equation}\label{eq:ddef}
\begin{aligned}
 \mathfrak{D}_{u_1}^{u_2} &:= \left\{r \geq R, u \in [u_1, u_2] \right\},\\
 \mathfrak{E}_{v_1}^{v_2} &:= \left\{r \leq R, v \in [v_1, v_2] \right\},\\
 \mathfrak{F}_{u_1}^{u_2} &:= \mathfrak{D}_{u_1}^{u_2} \cup \mathfrak{E}_{v_1}^{v_2}, \mbox{ such that } v_1-u_1 = 2 R_* \mbox{ and } v_2-u_2 = 2 R_*.
\end{aligned}
\end{equation}
\begin{figure}
\centering
\begin{tikzpicture}	

\node (I)    at ( 0,0) {};

\path  
  (I) +(90:4)  coordinate[label=90:$i^+$]  (top)
       +(-90:4) coordinate[label=-90:$i^-$] (bot)
       +(0:4)   coordinate                  (right)
       +(180:4) coordinate (left)
       ;
\path 
	(top) + (180:4) coordinate  (acca)
		+ (-45: 3) coordinate (nulluno)
		+ (-45: 4) coordinate (nulldue)
	;

\path 
	(left) + (-45: 3) coordinate (correspuno)
		+ (-45: 4) coordinate (correspdue)
	;

\draw [name path = rconst] (top) to [bend left = 15](bot);
\draw [name path = nc, opacity= 0] (nulluno) to (correspuno);
\draw [name path = nc1, opacity= 0] (nulldue) to (correspdue);
\draw [name intersections={of=rconst and nc}](nulluno) to node[midway, sloped, above]{ \tiny $\{u = u_2\}$} (intersection-1);
\draw [name intersections={of=rconst and nc1}] (nulldue) to node[midway, sloped, below]{\tiny $\{u = u_1\}$} (intersection-1);

\path (top) + (-135: 2.2) coordinate (horiuno)
			+ (-135: 3.2) coordinate (horidue);

\draw [name path = taudue, name intersections={of=rconst and nc}] (horiuno) to node[midway,sloped, above ]{\tiny $\{v = v_2\}$}(intersection-1);
\draw [name path = tauno, name intersections={of=rconst and nc1}] (horidue) to node[midway, sloped, below]{\tiny $\{v = v_1\}$} (intersection-1);

\coordinate [name intersections={of=tauno and nc}] (bam) at (intersection-1);
\coordinate [name intersections={of=rconst and tauno}] (bamm) at (intersection-1);
\coordinate [name intersections={of=rconst and taudue}] (rbam) at (intersection-1);
\coordinate [name intersections={of=nc and rconst}] (topz) at (intersection-1);
\fill[gray!10] (bam) -- (bamm) -- (nulluno) -- (nulldue) -- cycle;
\fill[gray!10] (bamm) -- (topz) -- (horiuno) -- (horidue) -- cycle;

\path
	(nulluno) + (-100:1.3) coordinate[label=90:$\mathfrak{D}_{u_1}^{u_2}$] (d12);

\draw [name path = taudue, name intersections={of=rconst and nc}] (horiuno) to node[midway,sloped, above ]{\tiny $\{v = v_2\}$}(intersection-1);
\draw [name path = tauno, name intersections={of=rconst and nc1}] (horidue) to node[midway, sloped, below]{\tiny $\{v = v_1\}$} (intersection-1);

\draw [name intersections={of=rconst and nc}] (nulluno) -- (intersection-1);
\draw [name intersections={of=rconst and nc1}] (nulldue) -- (intersection-1);

\draw (left) -- 
          node[midway, above left, sloped]    {$\mathcal{H}^+$}
      (top) --
          node[midway, above, sloped] {$\mathcal{I}^+$}
      (right) -- 
          node[midway, above, sloped] {$\mathcal{I}^-$}
      (bot) --
          node[midway, above, sloped]    {$\mathcal{H}^-$}    
      (left) -- cycle;

\draw [name path = rconst, densely dotted] (top) to [bend left = 15](bot);

\draw[decorate,decoration=zigzag] (top) -- (acca)
      node[midway, above, inner sep=2mm] {$r=0$};

\path (horiuno) +(-59:2)  coordinate[label=center: $\mathfrak{E}_{v_1}^{v_2}$]  (top);
      
\end{tikzpicture}
\caption{Penrose diagram of the considered regions.}\label{figuno}
\end{figure}
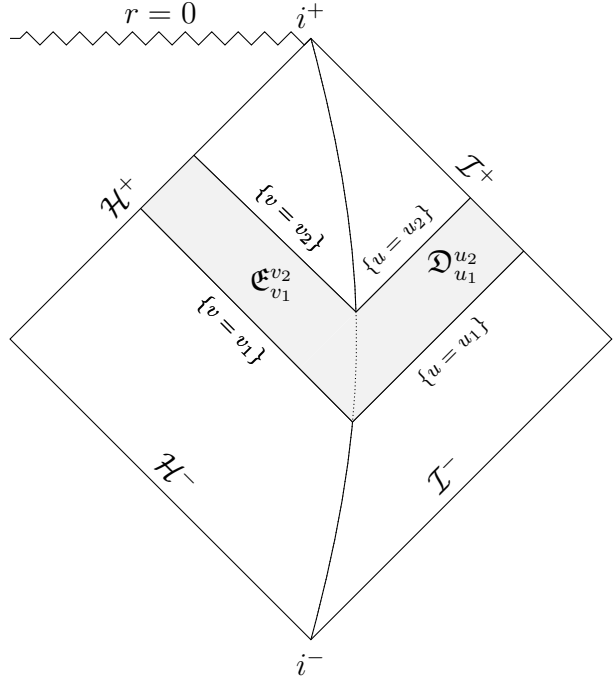

\subsection{Energy conservation}

\begin{lemma}
Let $\phi$ be a smooth solution to the \fackip Equation~(\ref{eq:rw1}) on $\{u \geq u_0\} \cap \{v \geq v_0\}$.
Let $v_2 \geq v_1 \geq v_0$, and $u_2 \geq u_1 \geq u_0$. Defining the fluxes as in (\ref{eq:fdefu}), (\ref{eq:fdefv}), we have that
\begin{equation} \label{eq:encons}
	\fara^T_{u_2}[\phi](v_1, v_2) + \fara^T_{v_2}[\phi](u_1, u_2) = \fara^T_{v_1}[\phi](u_1, u_2) + \fara^T_{u_1}[\phi](v_1, v_2).
\end{equation}
\end{lemma}

\begin{proof}
First, we notice that the \fackip Equation implies:
\begin{equation*}
\begin{aligned}
	(\snabla_\lbar + \snabla_L) \int_{\mathbb{S}^2} \left\{|\snabla_L \phi|^2+|\snabla_\lbar \phi|^2	+ 2 \frac{1-\mu}{r^2} |r \snabla \phi|^2 + 2V|\phi|^2\right\}\desphere\\
	+ (\snabla_\lbar - \snabla_L) \int_{\mathbb{S}^2}\left\{
	|\snabla_L \phi|^2-|\snabla_\lbar \phi|^2\right\}\desphere = 0.
\end{aligned}
\end{equation*}
Integrating with respect to $\de u \de v$ yields the claim.
\end{proof}

\subsection{Morawetz estimate}

\begin{lemma}\label{lem:moraw}
There exists a positive constant $C$ such that the following holds. Let $\phi$ be a smooth solution to the \fackip Equation~(\ref{eq:rw1}) on $\{u \geq u_0\} \cap \{v \geq v_0\}$. Defining the fluxes as in (\ref{eq:fdefu}), (\ref{eq:fdefv}), we have that
\begin{equation} \label{eq:fullen}
\begin{aligned}
 \int_{u_1}^{u_2}\int_{v_1}^{ v_2}\int_{\mathbb{S}^2} 
&\left\{\frac 1 {r^2}|\snabla_L \phi - \snabla_\lbar \phi|^2+ \frac{(r-3M)^2}{r^3}\left(|\snabla \phi|^2 + \frac 1 {r} |\snabla_L \phi + \snabla_\lbar \phi|^2\right)\right.\\& \left.+\frac 1 {r^3} |\phi|^2\right\} (1-\mu) \de u \de v  \desphere\\ &\leq
 C(\fara^T_{v_1}[\phi](u_1,u_2)+ \fara^T_{u_1}[\phi](v_1,v_2)).
\end{aligned}
\end{equation}
\end{lemma}

\begin{proof}[Proof of Lemma~\ref{lem:moraw}]
We consider the following identities, which follow from the \fackip equation. Let $f: \mathcal{S}_e \to \R$ be a smooth radial function. Let $(\cdot)'$ denote differentiation with respect to $\partial_{r_*}$.
\begin{equation}\label{en1}
\begin{aligned}
	(\snabla_\lbar + \snabla_L)&\int_{\mathbb{S}^2} f\left\{|\snabla_L \phi|^2-|\snabla_\lbar \phi|^2\right\} \desphere\\
	+ (\snabla_\lbar - \snabla_L) &\int_{\mathbb{S}^2}f\left\{|\snabla_L \phi|^2+|\snabla_\lbar \phi|^2	- 2 \frac{1-\mu}{r^2} |r \snabla \phi|^2 - 2V|\phi|^2\right\} \desphere\\
	+ &\int_{\mathbb{S}^2} \left\{2f'(|\snabla_L \phi|^2+ |\snabla_\lbar\phi|^2)-4 \partial_{r^\star}\left(f\frac{1-\mu}{r^2} \right)|r \snabla \phi|^2\right. \\&\left.\phantom{ \frac 2 3} -4 \partial_{r^\star}(fV)|\phi|^2\right\}\desphere = 0.
\end{aligned}
\end{equation}
We also have
\begin{equation}\label{en2}
\begin{aligned}
	(\snabla_\lbar + \snabla_L)& \int_{\mathbb{S}^2} \left(f' \phi \cdot (\snabla_\lbar + \snabla_L) \phi \right)\desphere\\
	-(\snabla_\lbar - \snabla_L)& \int_{\mathbb{S}^2} \left(f' \phi \cdot (\snabla_\lbar - \snabla_L)\phi+f''|\phi|^2 \right)\desphere\\
	+&\int_{\mathbb{S}^2} \left\{-2f'''|\phi|^2-4f'\snabla_\lbar \phi \cdot \snabla_L \phi \phantom{\frac 2 1} \right. \\
	&\left. + 4 f' \left( \frac{1-\mu}{r^2} |r\snabla \phi|^2+V|\phi|^2\right)\right\}\desphere = 0.
\end{aligned}
\end{equation}
In the previous equation, the dot $\cdot$ indicated that we are contracting with $\gbar$.
Let us now add the previous Equations (\ref{en1}) and (\ref{en2}), to get
\begin{equation}\label{mora}
\begin{aligned}
 	(\snabla_\lbar + \snabla_L) &\int_{\mathbb{S}^2} \left( f\left\{|\snabla_L \phi|^2-|\snabla_\lbar \phi|^2\right\}+f' \phi \cdot (\snabla_\lbar + \snabla_L)\phi \right)\desphere\\+
 	(\snabla_\lbar - \snabla_L)&\int_{\mathbb{S}^2}  \left(f\left\{|\snabla_L \phi|^2+|\snabla_\lbar \phi|^2 + 2 \frac{1-\mu}{r^2} |r \snabla \phi|^2 + 2V|\phi|^2\right\}\right.\\& \left.\phantom{\frac 1 2 }+ f' \phi \cdot (\snabla_\lbar - \snabla_L)\phi+f''|\phi|^2\right)\desphere\\
 	+&\int_{\mathbb{S}^2} \left\{ 2f'(|\snabla_L \phi - \snabla_\lbar\phi|^2) + |r \snabla \phi|^2\left[-4f \left(\frac{1-\mu}{r^2} \right)' \right] \right. \\ &\left.\phantom{\frac 1 2 }+ |\phi|^2(-4fV'-2f''')\right\}\desphere= 0.
\end{aligned}
\end{equation}
We now proceed to integrate Equation~(\ref{mora}) on spacetime against the form $\de u \de v$. By Lemma~\ref{lem:poinca}, we note that for the bulk term to be positive, it suffices that there exists a $c > 0$ such that
\begin{equation}\label{posicon}
	-2 \frac{(V+\frac 1 {r^2} (1-\mu))'}{1-\mu}f - \frac{f'''}{1-\mu} \geq \frac c {r^3}
\end{equation}
and that $f'>0$.
By choosing $f(r) := \left(1+\frac{M}{r} \right) \left(1-\frac{3M}{r} \right)$, let us calculate, as in \cite{linearized},
\begin{align*}
 f'&= (1-\mu)\left(\frac{2M}{r^2}+ \frac{6M^2}{r^3}\right), \\
 f'' &= (1-\mu)\partial_r(f') = \frac{2 M \left(-48 M^3+30 M^2 r+M r^2-2 r^3\right)}{r^6},\\
 f''' &= (1-\mu)\partial_r(f'') = \frac{4 M (r-2M) \left(144 M^3-75 M^2 r-2 M r^2+3 r^3\right)}{r^8}.
\end{align*}
 Multiplying inequality (\ref{posicon}) by $-\frac 1 4 r^3$, we obtain that inequality (\ref{posicon}) is achieved if and only if
\begin{equation*}
	\frac{144 M^4-93 M^3 r-8 M^2 r^2+13 M r^3-2 r^4 }{r^4}< c,
\end{equation*}
which is always the case after $r = 2M$. 

We therefore obtain the following estimate, making use of the positivity of the angular terms:

\begin{equation}
\begin{aligned}
 \int_{u_1}^{u_2}\int_{v_1}^{v_2}\int_{\mathbb{S}^2} &
\left\{\frac 1 {r^2}|\snabla_L \phi - \snabla_\lbar \phi|^2+ \frac{(r-3M)^2}{r^3}|\snabla \phi|^2+\frac 1 {r^3} |\phi|^2\right\} (1-\mu) \de u \de v \desphere \\ &\leq 
 C(F^T_{u_1}[\phi](v_1,v_2)+ F^T_{v_1}[\phi](v_1,v_2)).
\end{aligned}
\end{equation}
We can recover the missing derivative by integrating Equation~(\ref{en1}) with a monotonically increasing $f$, which vanishes of third order at $r = 3M$, and get:

\begin{equation}
\begin{aligned}
 \int_{u_1}^{u_2}\int_{v_1}^{ v_2}\int_{\mathbb{S}^2} 
&\left\{\frac 1 {r^2}|\snabla_L \phi - \snabla_\lbar \phi|^2+ \frac{(r-3M)^2}{r^3}\left(|\snabla \phi|^2 + \frac 1 {r} |\snabla_L \phi + \snabla_\lbar \phi|^2\right)\right. \\
&\left. +\frac 1 {r^3} |\phi|^2\right\} (1-\mu) \de u \de v  \desphere\\ &\leq
 C(\fara^T_{v_1}[\phi](u_1,u_2)+ \fara^T_{u_1}[\phi](v_1,v_2)).
\end{aligned}
\end{equation}
This is the claim.
\end{proof}

\subsection{The redshift estimate}

\begin{lemma}[The redshift estimate]\label{lem:redsh}
There exists a positive constant $C$ such that the following holds. Let $\phi$ be a smooth solution to the \fackip Equation~(\ref{eq:rw1}) on~$\{u~\geq~u_0\}~\cap \{v~\geq~v_0\}$. Let $u_2 \geq u_1 \geq u_0$ and $v_2 \geq v_1 \geq v_0$. Let $\chi(r)$ be a smooth cutoff function such that $\chi_{\mathcal{H}^+}(r) = 1$ for $r \in (2M, 5/2 M)$ and $\chi_{\mathcal{H}^+}(r) = 0$ for $r \geq 3M$.
Defining the fluxes as in (\ref{eq:fdefu}), (\ref{eq:fdefv}), we have that
\begin{align}\label{eq:enconse}
\fara^\infty(u_2, v_2) \leq C \fara^\infty(u_1, v_1),
\end{align}
\begin{equation}\label{eq:redsh}
\begin{aligned}
 \int_{u_1}^{u_2}&\int_{v_1}^{ v_2}\int_{\mathbb{S}^2} 
\left\{\frac 1 {r^2}|\snabla_L \phi - \snabla_{\lbar} \phi|^2\right. \\& 
\left.+ \frac{(r-3M)^2}{r^3}\left(|\snabla \phi|^2 + \frac 1 {r} |\snabla_L \phi + \snabla_{\lbar} \phi|^2\right)+\frac 1 {r^3} |\phi|^2\right\} (1-\mu) \de u \de v  \desphere\\ 
+\int_{u_1}^{u_2}&\int_{v_1}^{ v_2}\int_{\mathbb{S}^2} \chi_{\mathcal{H}^+}(r) |\snabla_{(1-\mu)^{-1}\lbar}\phi|^2 (1-\mu) \de u \de v \desphere
\\
&\leq
 C(\fara^N_{v_1}[\phi](u_1,u_2)+ \fara^T_{u_1}[\phi](v_1,v_2)).
 \end{aligned}
\end{equation}
\end{lemma}

\begin{proof}
Let $2M < r_c < 3M$. Let $h(r)$ be a smooth radial function such that $h(r) = (1-\mu)^{-1}$ for $r \in (2M, r_c)$, $h(r) = 0$ for $r \in [3M, +\infty)$. The following relation can be deduced from the \fackip equation:
\begin{equation}\label{eq:redest}
 \begin{aligned}
  &L\left\{h(r)|\snabla_\lbar \phi|^2 \right\} + \lbar\left\{(1-\mu)h(r)|\snabla \phi|^2 \right\} + \lbar\left\{\frac{(1-\mu)}{r^2}h(r)|\phi|^2 \right\} \\
&- L\left\{h(r)\right\}|\snabla_\lbar \phi|^2 - \lbar\left\{\frac{(1-\mu)}{r^2}h(r)\right\}|\phi|^2 
-r^2 \lbar\left\{\frac{(1-\mu)}{r^2}h(r)\right\}|\snabla \phi|^2 \stackrel{\mathbb{S}^2}{=} 0.
 \end{aligned}
\end{equation}
Here, $\stackrel{\mathbb{S}^2}{=}$ denotes equality after integration on $\mathbb{S}^2$ against the form $\desphere$.
We proceed to integrate (\ref{eq:redest}) against the form $\de u \de v$ in the region $\{u_1 \leq u \leq u_2 \} \cap \{v_1 \leq v \leq v_2\}$. The second claim is achieved absorbing the error terms arising from $h$ with estimate (\ref{eq:fullen}) (note that we only require one derivative on the initial slice).

Analysing the boundary terms arising from integration of (\ref{eq:redest}) and adding a multiple of (\ref{eq:encons}) we obtain the first claim.
\end{proof}

\subsection{The estimates required for the $r^p$ method}

\begin{definition}\label{moraint}
Let $\chi_{\mathcal{H}^+}$ be the smooth cutoff function of Lemma~\ref{lem:redsh}. We define $m[\phi]$ to be the integrand in the left hand side of Equation~(\ref{eq:redsh}):
\begin{equation}\label{def:mphi}
\begin{aligned}
m[\phi] : = &\frac 1 {r^2}|\snabla_L \phi - \snabla_{\lbar} \phi|^2+ \frac{(r-3M)^2}{r^3}\left(|\snabla \phi|^2 + \frac 1 {r} |\snabla_L \phi + \snabla_{\lbar} \phi|^2\right)+\frac 1 {r^3} |\phi|^2\\
+& \chi_{\mathcal{H}^+}(r) |\snabla_{(1-\mu)^{-1}\lbar}\phi|^2.
\end{aligned}
\end{equation}
Furthermore, we let
\begin{equation}\label{eq:malldef}
m_{\text{all}}[\phi] := m[\phi] + m[\snabla_T \phi].
\end{equation}
\end{definition}
\begin{remark}
Let us notice that the \fackip Equation implies that there exists a constant $C > 0$, depending only on $M$, such that
\begin{equation*}
\begin{aligned}
 &\int_{u_1}^{\infty} \int_{v_1}^{\infty}\int_{\mathbb{S}^2}  \mall (1-\mu) \de u \de v  \desphere \\
 &\geq C \int_{u_1}^{\infty} \int_{v_1}^{\infty}\int_{\mathbb{S}^2} \left(\frac 1 {r^2} (|\snabla_L \phi|^2+|\snabla_{(1-\mu)^{-1}\lbar} \phi|^2) + \frac 1 {r^3} |\phi|^2 + \frac 1 r |\snabla \phi|^2 \right)(1-\mu) \de u \de v  \desphere,
\end{aligned}
\end{equation*}
hence $\mall$ controls all non-degenerate derivatives of $\phi$.
\end{remark}

\begin{remark}
From the commutation relations: $[\snabla_T, \snabla_L] = [\snabla_T, r\snabla_\lbar] = [\snabla_T, \snabla_A] = 0$, we obtain that $\snabla_T \phi$ satisfies the \fackip Equation, if $\phi$ does. The reasoning which led to the proof of Lemma~\ref{lem:redsh} then yields
\begin{equation}\label{eq:morader}
\begin{aligned}
 \int_{u_1}^{\infty} \int_{v_1}^{\infty}\int_{\mathbb{S}^2}  \mall (1-\mu) \de u \de v  \desphere
  \leq C(\fara^\infty[\phi] (u_1, v_1) + \fara^\infty[\snabla_T \phi] (u_1, v_1)).
\end{aligned}
\end{equation}
\end{remark}

\begin{lemma}[$p$-hierarchy]
There exists a positive number $R_*$ and a positive constant $C$ such that the following holds. Let $\phi$ be a smooth solution to the \fackip Equation~(\ref{eq:rw1}) on $\{u \geq u_0\} \cap \{v \geq v_0\}$, let $u_2 \geq u_1 \geq u_0$, $v_2 \geq v_1 \geq v_0$. Let $P(u)$ be $(u,2R_* + u).$ Defining fluxes as in (\ref{eq:fdefu}), (\ref{eq:fdefv}), we have that
\begin{equation}\label{pdue}
\begin{aligned}
\int_{\conplus_{u_2}\cap\{r \geq R\}} (r^2 |\snabla_L \phi|^2) \de v \desphere+
\int_{\duu} \left[r |\snabla_L \phi|^2 + |\snabla \phi|^2 + r^{-2}|\phi|^2\right] (1-\mu) \de u  \de v \desphere \\ \leq 
C \left(\fara^\infty[\phi](P(u_1))\right) + \int_{\conplus_{u_1}\cap\{r \geq R\}} (r^2 |\snabla_L \phi|^2) \de v \desphere.
\end{aligned}
\end{equation}
Also,
\begin{equation}\label{puno}
\begin{aligned}
\int_{\conplus_{u_2}\cap\{r \geq R\}}  (r |\snabla_L \phi|^2)\de v \desphere +
\int_{\duu}\left[ |\snabla_L \phi|^2 + |\snabla \phi|^2 + r^{-2}|\phi|^2\right] (1-\mu) \de u  \de v \desphere \\ \leq 
C \left(\fara^\infty[\phi](P(u_1))\right) + \int_{\conplus_{u_1}\cap\{r \geq R\}}  (r |\snabla_L \phi|^2)\de v\desphere.
\end{aligned}
\end{equation}

\end{lemma}

\begin{proof}
Let $p, k \in \R$. Let us consider the identity, which follows from the \fackip Equation~(\ref{eq:rw1}):
\begin{equation}\label{eq2.15}
 \begin{aligned}
  \lbar\left\{\frac{r^p}{(1-\mu)^k}|\snabla_L \phi|^2 \right\} + L\left\{\frac{r^p}{(1-\mu)^{k-1}}|\snabla \phi|^2 \right\} + L\left\{V\frac{r^p}{(1-\mu)^k}|\phi|^2 \right\} \\
- \lbar\left\{ \frac{r^p}{(1-\mu)^k}\right\}|\snabla_L \phi|^2 - L\left\{V\frac{r^p}{(1-\mu)^k}\right\}|\phi|^2 \\
+ \left[(2-p)r^{p-1}(1-\mu)^{2-k}+ r^p(k-1)(1-\mu)^{-k+1}\frac{2M}{r^2} \right]|\snabla \phi|^2 \stackrel{\mathbb{S}^2}{=} 0.
 \end{aligned}
\end{equation}
Here, $\spheq$ denotes equality after integration on $\mathbb{S}^2$ with respect to the volume form $\desphere$.
We fix $R_*>0$ big enough so that the following holds in the region $r_* \geq R_*$:
\begin{equation}
- \lbar \left(\frac{r^p}{(1-\mu)^k}\right) \geq \frac 1 2 r^{p-1},
\end{equation}
when $p$ is either $1$ or $2$ and $k$ is either $1$ or $2$.

Let us furthermore calculate
\begin{equation}
 - L \left(\frac{Vr^p}{(1-\mu)^{k}}\right) = (1-\mu)^{2-k} (2-p) r^{p-3} - 2Mr^{p-4}(1-k)(1-\mu)^{-k+1}.
\end{equation}
Recall the definition of $\mathfrak{D}$:
\begin{equation*}
 \mathfrak{D}_{u_1}^{u_2} := \left\{r \geq R, u \in [u_1, u_2] \right\}.
\end{equation*}
By possibly choosing $R$ bigger,
we assume that $R-M > 3M$. Let $f$ be a smooth radial function such that
\begin{equation*}
f(r) := \left\{
\begin{array}{ll}
1 & \text{for } r \geq R,\\
0 & \text{for } r \leq R-M.
\end{array}
\right.
\end{equation*}
We then integrate the expression
\begin{equation}\label{eq:pweiavg}
 \begin{aligned}
  &\lbar\left\{\frac{r^p}{(1-\mu)^k}f|\snabla_L \phi|^2 \right\} + L\left\{\frac{r^p}{(1-\mu)^{k-1}}f|\snabla \phi|^2 \right\} \\+
  &L\left\{V\frac{r^p}{(1-\mu)^k}f|\phi|^2 \right\} 
- \lbar\left\{ \frac{r^p}{(1-\mu)^k}f\right\}|\snabla_L \phi|^2 - L\left\{V\frac{r^p}{(1-\mu)^k}f\right\}|\phi|^2
 \end{aligned}
\end{equation}
on the region $\future(S_{u_1, v_1})$, with $v_1 = 2R_* + u_1$. 

We choose $p = 2$, $k = 2$, and we use the Morawetz estimate (\ref{eq:fullen}) in order to bound the spacetime error term in the strip $r \in [R-M, R]$. We obtain
\begin{equation}
\begin{aligned}
\int_{\conplus_{u_2}\cap\{r \geq R\}} (r^2 |\snabla_L \phi|^2) \de v \desphere+
\int_{\duu} \left[r |\snabla_L \phi|^2 + |\snabla \phi|^2 + r^{-2}|\phi|^2\right] (1-\mu) \de u  \de v \desphere \\ \leq 
C \fara^\infty[\phi](P(u_1))+ \int_{\conplus_{u_1}\cap\{r \geq R\}} (r^2 |\snabla_L \phi|^2) \de v \desphere.
\end{aligned}
\end{equation}
Choosing $p = 1$, $k = 1$, we obtain, instead
\begin{equation}
\begin{aligned}
\int_{\conplus_{u_2}\cap\{r \geq R\}}  (r |\snabla_L \phi|^2)\de v \desphere +
\int_{\duu}\left[ |\snabla_L \phi|^2 + |\snabla \phi|^2 + r^{-2}|\phi|^2\right] (1-\mu) \de u  \de v \desphere \\ \leq 
C \fara^\infty[\phi](P(u_1)) + \int_{\conplus_{u_1}\cap\{r \geq R\}}  (r |\snabla_L \phi|^2)\de v\desphere.
\end{aligned}
\end{equation}
This proves the Lemma.
\end{proof}

\subsection{Application of the $r^p$ method: decay of null fluxes}
We now apply the $r^p$-method of Dafermos and Rodnianski to prove integrated decay for $\phi$. We prove the following Lemma.
\begin{lemma}\label{lem:fluxdec}
There exists a positive number $R_*$ and a positive constant $C$ such that the following holds. Let $\phi$ be a smooth solution to the \fackip Equation~(\ref{eq:rw1}) on $\{u \geq u_0\} \cap \{v \geq v_0\}$, let $u \geq u_0$. Let $P(u)$ be $(u,2R_* + u).$ We have the decay of the flux:
\begin{equation}\label{eq:decaydue}
\begin{aligned}
& \fara^\infty[\phi](P(u)) \leq C u^{-2} \snorm{\phi}{2}{2}{0}.
\end{aligned}
\end{equation}
Here, we used the definition of norm in (\ref{eq:defsobo}).
\end{lemma}

\begin{proof}
We define the sequence $u_n := 2^n (|u_0|+M)$. Inequality (\ref{pdue}) now yields:
\begin{equation}
\int_{\conplus_{u_n}\cap\{r \geq R\}}  r^2 |\snabla_L \phi|^2 \de v \desphere \leq C \left(\fara^\infty[\phi](P(u_0))\right) + \int_{\conplus_{u_0}\cap\{r \geq R\}} (r^2 |\snabla_L \phi|^2)\de v \desphere.
\end{equation}
Also, 
\begin{equation}
\begin{aligned}
(u_{n+1}-u_n)  \int_{\conplus_{\tilde u_n} \cap \{r \geq R\}} r |\snabla_L \phi|^2 \de v \desphere \leq \int_{\mathfrak{D}^{u_{n+1}}_{u_n}} r |\snabla_L \phi|^2 (1-\mu) \de u \de v
 \desphere \\ \leq
C \left(\fara^\infty[\phi](P(u_0))\right) + \int_{\conplus_{u_n} \cap \{r \geq R\}}(r^2 |\snabla_L \phi|^2) \de v \desphere.
\end{aligned}
\end{equation}
This holds for some $\tilde{u}_n \in (u_n, u_{n+1})$. All in all, we have, for this new sequence $\tilde{u}_n$,
\begin{equation}\label{eq:decayun}
\begin{aligned}
&\int_{\conplus_{\tilde{u}_n} \cap \{r \geq R\}}  r |\snabla_L \phi|^2 \de v \desphere \\
&\leq
C \tilde{u}_n^{-1}\fara^\infty[\phi](P(u_0)) + 
\tilde{u}_n^{-1}\int_{\conplus_{u_0}\cap \{r \geq R\}}  r^2 |\snabla_L \phi|^2 \de v \desphere.
\end{aligned}
\end{equation}
Recall the definition of the spacetime regions: 
\begin{equation*}
\begin{aligned}
\mathfrak{E}_{v_1}^{v_2} &:= \left\{r \leq R, v \in [v_1, v_2] \right\},\\
 \mathfrak{F}_{u_1}^{u_2} &:= \mathfrak{D}_{u_1}^{u_2} \cup \mathfrak{E}_{v_1}^{v_2}, \mbox{ such that } v_1-u_1 = 2 R_* \mbox{ and } v_2-u_2 = 2 R_*.
\end{aligned}
\end{equation*}
Recall furthermore the definition of $m[\phi]$, as in (\ref{def:mphi}).
We employ Equation~(\ref{puno}) and the Morawetz estimate (recall: Equation~(\ref{eq:redsh}) and Definition \ref{moraint}), to obtain the following:
\begin{equation}\label{gro1}
\begin{aligned}
 \int_{\mathfrak{E}_{v_1}^{v_2}} m[\phi](1-\mu)\de u \de v \desphere +  \int_{\conplus_{u_2}\cap\{r\geq R \}} r |\snabla_L \phi|^2 \de v \desphere +\\ 
\int_{\duu} \left[ |\snabla_L \phi|^2 + |\snabla \phi|^2 + r^{-2}|\phi|^2\right](1-\mu) \de u  \de v \desphere \\ \leq 
C \fara^\infty[\phi](P(u_1)) + \int_{\conplus_{u_1}\cap \{r \geq R\}}  r |\snabla_L \phi|^2 \de v \desphere.
\end{aligned}
\end{equation}
Recall now that $m_{\text{all}}[\phi] := m[\phi] + m[\snabla_T \phi]$.
Plugging now the previous sequence $\tilde{u}_n$ into the formula (\ref{gro1}), summing it with estimate (\ref{eq:morader}) and using Fubini's Theorem, we obtain a second sequence $\tilde{\tilde u}_n$ (with $\tilde{\tilde{u}}_n+\tilde{\tilde{v}}_n = 2 R^*$) such that
\begin{equation}\label{gro2}
\begin{aligned}
  \ttildeu \int_{\mathbb{S}^2} \int_{\conminus_{\ttildev} \cap \{r \leq R\}} \mall \de u \desphere 
  \\
  + \ttildeu \int_{\mathbb{S}^2} \int_{\conplus_{\ttildeu} \cap \{r \geq R\}}
\left[ |\snabla_L \phi|^2 + |\snabla \phi|^2 + r^{-2}|\phi|^2\right] \de v \desphere \\ \leq 
 C \int_{\mathfrak{E}_{\tilde{v}_1}^{\tilde{v}_2}} \mall \de u \de v \desphere+ C \int_{\mathfrak{D}_{\tilde{u}_1}^{\tilde{u}_2}} \left[ |\snabla_L \phi|^2 + |\snabla \phi|^2 + r^{-2}|\phi|^2\right](1-\mu) \de u  \de v \desphere \\ \leq 
C \left(\fara^\infty[\phi](P(\tilde{u}_n))+ \fara^\infty[\snabla_T \phi](P(\tilde{u}_n))\right) + \int_{\conplus_{\tilde u_n} \cap \{r \geq R\}} r |\snabla_L \phi|^2 \de v \desphere  \\  \leq 
 C \left(\fara^\infty[\phi](P(\tilde{u}_n))+ \fara^\infty[\snabla_T \phi](P(\tilde{u}_n))\right)
+ C \tilde{u}_n^{-1}\fara^\infty[\phi](P(u_0)) \\ + \tilde{u}_n^{-1}\int_{\conplus_{u_0}} r^2 |\snabla_L \phi|^2 \de v \desphere .
\end{aligned}
\end{equation}
We commute the \fackip Equation~(\ref{eq:rw1}) with $\snabla_T$, and we see that 
$$
\fara^\infty[\snabla_T \phi](P(u))
$$
decays, by the same reasoning for the decay of $\fara^\infty[\phi](P(u))$. We use the monotonicity of energy given in (\ref{eq:enconse}) to eliminate the restriction to the dyadic sequence. We finally obtain the bound
\begin{equation*}
\begin{aligned}
 \fara^\infty[\phi](P(u)) \leq
C u^{-2} \sum_{i = 0}^2 \left( \fara^\infty [(\snabla_T)^i \phi](P(u_0))+\int_{\conplus_{u_0} \cap \{r \geq R\}}r^2 |\snabla_L (\snabla_T)^{\min\{i,1\}} \phi|^2\de v \desphere\right).
\end{aligned}
\end{equation*}
This is the claim.
\end{proof}
\begin{cor} Assume the hypotheses of Lemma~\ref{lem:fluxdec}. Consider the Lie derivative $\slie$ induced from the connection $\snabla$. Recall the definition of the set $\leo := \{\Omega_1, \Omega_2, \Omega_3\}$. The following estimate holds, for $j \in \{1,2,3\}$:
\begin{equation}\label{eq:decaylie}
\begin{aligned}
 \fara^\infty [\slie_{\Omega_j} \phi](P(u))\leq
C u^{-2} \snorm{\phi}{2}{2}{1}.
\end{aligned}
\end{equation}
Here, we used the definition of norm in (\ref{eq:defsobo}).
\end{cor}

\begin{proof}
Let $\eta \in \Gamma(\mathcal{B}^*)$. An easy calculation implies that we have the commutation relations:
\begin{equation*}
[\slie_{\Omega_i}, \snabla_L]\eta = 0, \qquad [\slie_{\Omega_i}, \snabla_\lbar]\eta = 0.
\end{equation*}
Furthermore, since the $\Omega_i$'s are Killing vectors for the induced metric $\gbar$ on the spheres, by Equation 3.25 in \cite{globalnon}, we have that
\begin{equation*}
[\slie_{\Omega_i}, \snabla ] = 0.
\end{equation*}
We  therefore obtain that $\slie_{\Omega_i} \phi$ also satisfies the \fackip equation (\ref{eq:rw1}). The proof then proceeds as in Lemma~\ref{lem:fluxdec}.
\end{proof}

\subsection{The $r^p$-method revisited}
We will now derive decay for a null flux, crucial for the decay of the extreme components.
\begin{lemma}\label{lem:fluxa}
There exists a positive number $R_*$ and a positive constant $C$ such that the following holds. Let $\phi$ be a smooth solution to the \fackip Equation~(\ref{eq:rw1}) on $\{u \geq u_0\} \cap \{v \geq v_0\}$. We define the following flux:
\begin{equation}\label{eq:fluxincdef}
\fara^\smallsetminus_p[\eta](\tilde u, \tilde v) := \int_{\conminus_{ \tilde v} \cap \{r \geq R\}\cap \{u \geq \tilde u\}} |\eta|^2( u, \tilde v, \omega) (r(u,  \tilde v))^p \de u \desphere.
\end{equation}
Let $u \geq u_0$, $v \geq v_0$. We have the decay estimates:
\begin{align}\label{eq:fluxa1}
&\fluxob_0[\phi](u_0, v) \leq \  \ \, \ \ \, \ \, C \snorm{\phi}{2}{0}{0}^2,   \qquad \fluxob_2[\snabla \phi](u_0, v) \leq C \snorm{\phi}{2}{0}{0}^2,
\\\label{eq:fluxa2}
&\fluxob_{-1}[\phi](v/2, v) \leq C v^{-1} \snorm{\phi}{2}{1}{0}^2,  \qquad \fluxob_{1}[\snabla \phi](v/2, v) \leq C  v^{-1} \snorm{\phi}{2}{1}{0}^2,
\\\label{eq:fluxa3}
&\fluxob_{-2}[\phi](v/2, v) \leq C v^{-2} \snorm{\phi}{2}{2}{0}^2,  \qquad \fluxob_{0}[\snabla \phi](v/2, v) \leq C  v^{-2} \snorm{\phi}{2}{2}{0}^2.
\end{align}
\end{lemma}

\begin{proof} We choose $R$ such that $R- M > 3M$.
We choose a smooth radial function $f$ with the following requirements:
\begin{equation*}
f(r) = \left\{
\begin{array}{cc}
1 & \text{for } r \geq R, \\
0 & \text{for } r \leq R-M.
\end{array}
\right.
\end{equation*}
Let $v_1 \geq v_0$, $u_1 \geq u_0$. We define the spacetime region
$$\mathfrak{W}(u_1, R, R_1, v_1) := \{u\geq u_1\} \cap \{r \geq R\} \cap \{2(R_1)_* + u_1\leq v \leq v_1\}.$$
\begin{figure}
\centering
\begin{tikzpicture}	

\node (I)    at ( 0,0) {};

\path  
  (I) +(90:4)  coordinate[label=90:$i^+$]  (top)
       +(-90:4) coordinate[label=-90:$i^-$] (bot)
       +(0:4)   coordinate                  (right)
       +(180:4) coordinate (left)
       ;
       
\path 
	(top) + (180:4) coordinate  (acca)
		+ (-45: 3) coordinate (nulluno)
		+ (-45: 4) coordinate (nulldue)
	;

\path 
	(left) + (-45: 3) coordinate (correspuno)
		+ (-45: 4) coordinate (correspdue)
	;

\path (top) + (-135: 0.5) coordinate (w1); 

\path (w1) + (-45: 6) coordinate (w2);

\path (w1) + (-135:1) coordinate (w3);
\path (w2) + (-135:1) coordinate (w4);
\draw [name path = wbuilda, opacity = 0] (w3) -- (w4);

\draw [name path = rconst] (top) to [bend left = 15](bot);
\draw [name path = nc, opacity= 0] (nulluno) to (correspuno);
\draw [name path = nc1, opacity= 0] (nulldue) to (correspdue);
\draw [name path = wbuild, opacity = 0] (w1) to (w2);

\coordinate [
 name intersections={of=rconst and wbuild}] (wwa) at (intersection-1);
\coordinate [
 name intersections={of=nc1 and wbuild}] (wwb) at (intersection-1);
\coordinate [ 
name intersections={of=nc1 and nc}] (wwc) at (intersection-1);

\coordinate [name intersections={of=rconst and wbuilda}
] (wwdum) at (intersection-1);

\coordinate [
name intersections={of=nc1 and wbuilda}] (wwe) at (intersection-1);

\coordinate [ 
 name intersections={of=nc and wbuilda}] (wwf) at (intersection-1);

\path (top) + (-135: 2.2) coordinate (horiuno)
			+ (-135: 3.2) coordinate (horidue);

\draw (wwf) to (wwdum) to node[midway, sloped, above]{\tiny $\{v = v_1\}$} (wwb);

\draw (left) -- 
          node[midway, above left, sloped]    {$\mathcal{H}^+$}
      (top) --
          node[midway, above, sloped] {$\mathcal{I}^+$}
      (right) -- 
          node[midway, above, sloped] {$\mathcal{I}^-$}
      (bot) --
          node[midway, above, sloped]    {$\mathcal{H}^-$}    
      (left) -- cycle;

\draw[decorate,decoration=zigzag] (top) -- (acca)
      node[midway, above, inner sep=2mm] {$r=0$};

\fill[color=gray!10] (wwe)--(wwf)--(wwdum)--(wwb)--cycle;
\fill[color=gray!10] (wwa)--(wwf)--(wwe)--cycle;

\path
	(nulluno) + (-125:1.5) coordinate[label=90:$\mathfrak{W}_{1}$] (d12);

\draw [name path = rconst] (top) to [bend left = 15] node[midway, left]{\tiny $\{r = R\}$} (bot);
\draw [name path = nc1, opacity= 1] (nulldue) to node[midway, sloped, below right]{\tiny $\{u = u_1\}$} (correspdue);
\draw [name path = wbuild, opacity = 0] (w1) to (w2);
\draw (wwf) to (wwdum) to node[midway, sloped, above]{\tiny $\{v = v_1\}$} (wwb);
\end{tikzpicture}
\caption{Penrose diagram of the regions considered.}\label{figdue}
\end{figure}
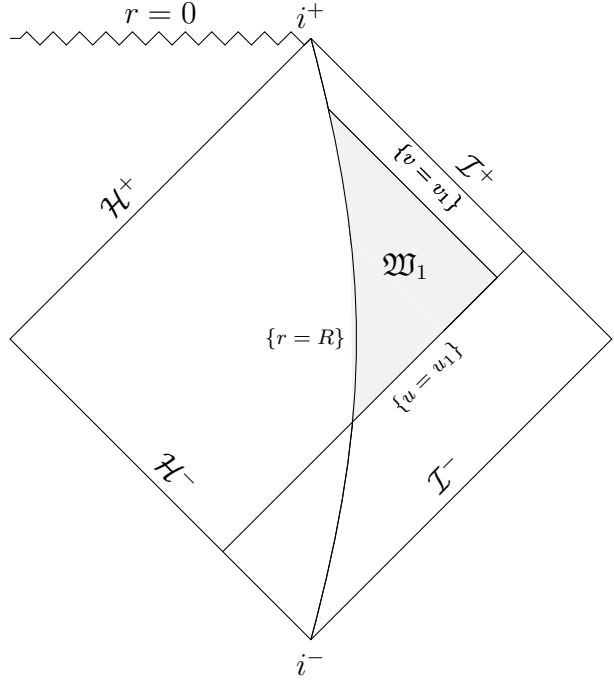
To shorten notation, we let $\mathfrak{W}_1 := \mathfrak{W}(u_1, R, R, v_1)$, and $\mathfrak{W}_2 := \mathfrak{W}(u_1, R-M, R, v_1)$.
We integrate the expression (\ref{eq:pweiavg}) on the region $\mathfrak{W}_2$ against the form $\de u \de v \desphere$. 
We obtain, using (\ref{eq2.15}),

\begin{equation}\label{eq:pbis}
 \begin{aligned}
 \int_{\{v = v_1\} \cap \mathfrak{W}_1}\left\{\frac{r^p}{(1-\mu)^{k-1}}|\snabla \phi|^2 +V\frac{r^p}{(1-\mu)^k}|\phi|^2 \right\} \de u\desphere\\+
\int_{\mathfrak{W}_1} \left\{- \lbar\left\{ \frac{r^p}{(1-\mu)^k}\right\}|\snabla_L \phi|^2 - L\left\{V\frac{r^p}{(1-\mu)^k}\right\}|\phi|^2\right\}\de u \de v \desphere\\
+\int_{\mathfrak{W}_1} \left\{(2-p)r^{p-1}(1-\mu)^{2-k}+ r^p(k-1)(1-\mu)^{-k+1}\frac{2M}{r^2} \right\}|\snabla \phi|^2 \de u \de v \desphere\\ \lesssim
  \int_{\{u = u_1\} \cap \mathfrak{W}_1} \left\{\frac{r^p}{(1-\mu)^k}|\snabla_L \phi|^2 \right\}\de v \desphere
  + \fara^\infty[\phi](u_1, u_1+2R_*).
 \end{aligned}
\end{equation}
Here, we used the Morawetz estimate (\ref{eq:redsh}) as well as the conservation of energy (\ref{eq:enconse}) to bound the errors (both boundary and spacetime) arising from the cutoff $f$ (here is where we use that $R > 4M$ not to lose derivatives).

First, recall: $V = \frac{1-\mu}{r^2}$. Also, by possibly increasing $R$,
\begin{equation}
- \lbar \left\{ \frac{r^p}{(1-\mu)^k}\right\} \geq \frac 1 2 r^{p-1}
\end{equation}
on $\mathfrak{W}_1$, for $p = 0,1,2$.

We now plug $p = 2$, $k =1$ in the previous equation (\ref{eq:pbis}). We obtain that there exists a positive constant $K > 0$ such that the following holds:
\begin{equation}\label{eq3.1}
 \begin{aligned}
 \int_{\{v = v_1\} \cap \mathfrak{W}_1}\left\{{r^2}|\snabla \phi|^2 +|\phi|^2 \right\} \de u\desphere+
\int_{\mathfrak{W}_1} K r|\snabla_L \phi|^2 \de u \de v \desphere \\
   \lesssim
     \int_{\{u = u_1\} \cap \mathfrak{W}_1} \frac{r^2}{(1-\mu)}|\snabla_L \phi|^2\de v \desphere + \fara^\infty[\phi](u_1, u_1+2R_*).
 \end{aligned}
\end{equation}
In particular, this means that the flux
\begin{equation*}
\int_{\{v = v_1\} \cap \mathfrak{W}_1}\left\{{r^2}|\snabla \phi|^2 +|\phi|^2 \right\} \de u\desphere \leq C \snorm{\phi}{2}{0}{0}^2.
\end{equation*}
Furthermore, using again Equation~(\ref{eq:pbis}), with $p = 1$, $k =1$, we obtain:
\begin{equation}\label{eq:pundue}
 \begin{aligned}
 \int_{\{v = v_1\} \cap \mathfrak{W}_1}\left\{r|\snabla \phi|^2 +r^{-1}|\phi|^2 \right\} \de u\desphere\\ \lesssim
  \int_{\{u = u_1\} \cap \mathfrak{W}_1} r|\snabla_L \phi|^2 \de v \desphere
  + \fara^\infty[\phi](u_1, u_1+2R_*).
 \end{aligned}
\end{equation}
Using inequality (\ref{eq:pundue}) with the choice $u_1 = v/2$, $v_1 = v$, together with (\ref{gro2}), (\ref{eq:decayun}), (\ref{puno}) and the monotonicity of $\fara^\infty[\phi]$ (\ref{eq:enconse}) to eliminate the restriction to the dyadic sequence, we obtain
\begin{equation}
\fluxob_{-1}[\phi](v/2, v)+\fluxob_{1}[\snabla \phi](v/2, v) \leq C v^{-1} \snorm{\phi}{2}{1}{0}^2.
\end{equation}
Furthermore, using again Equation~(\ref{eq:pbis}), with $p = 1$, $k =1$, together with inequality (\ref{gro2}) commuted once with $\snabla_T$, we obtain similarly
\begin{equation}
\fluxob_{-2}[\phi](v/2, v)+\fluxob_{0}[\snabla \phi](v/2, v) \leq C v^{-2} \snorm{\phi}{2}{2}{0}^2.
\end{equation}
This finishes the proof of the Lemma.
\end{proof}

\section{Decay for solutions to the spin $\pm$1 Teukolsky equations}\label{sec:decayteu}
In this section, we assume that $\alpha$ and $\alphabar$ are solutions to the spin $\pm1$ Teukolsky Equations, and we prove items (\ref{eq:da1}), (\ref{eq:dab1}), (\ref{eq:da3}), (\ref{eq:dab2}) of Theorem~\ref{prop:decayteu}. We postpone the proof of item \ref{eq:da4} to the next section, which is self-contained and in particular does not depend on the estimates for the incoming fluxes in Lemma~\ref{lem:fluxa}.

\begin{proof}[Proof of estimates (\ref{eq:da1}), (\ref{eq:dab1}), (\ref{eq:da3}), (\ref{eq:dab2}) in Theorem~\ref{prop:decayteu} ]
We divide the proof in four steps.
\subsection*{Step 1: estimates on $\alpha$, region of unbounded $r$}
We recall the definition of $\phi$:
\begin{equation}
\phi_A = \frac{r^2}{1-\mu} \snabla_\lbar (r\alpha_A).
\end{equation}
Let $p \in \R$. It follows that
\begin{equation}\label{eq:opa}
\snabla_\lbar (r^{1+p}\alpha_A) = - p (1-\mu)r^p \alpha_A + (1-\mu)r^{p-2} \phi_A.
\end{equation}
Let $2p > \varepsilon > 0$, an application of the Cauchy-Schwarz inequality implies
\begin{equation}\label{eq:opaa}
\begin{aligned}
\snabla_\lbar(r^{2+2p}|\alpha|^2) = 2r^{1+p}\alpha^A(-p(1-\mu)r^p \alpha_A+(1-\mu)r^{p-2}\phi_A)\\
 \leq -(2p-\varepsilon)(1-\mu)r^{1+2p}|\alpha|^2+ \varepsilon^{-1}(1-\mu)r^{2p-3}|\phi|^2
\end{aligned}
\end{equation}
A straightforward calculation implies that $[r\snabla, \snabla_\lbar] = 0$. Commuting Equation~(\ref{eq:opa}) twice with $r \snabla$, we obtain
\begin{equation*}
\snabla_\lbar (r^{3+p}\snabla_C \snabla_B \alpha_A) = - p (1-\mu)r^{p+2} \snabla_C \snabla_B \alpha_A + (1-\mu)r^{p} \snabla_C \snabla_B \phi_A.
\end{equation*}
This implies:
\begin{equation}\label{eq:opb}
\snabla_\lbar(r^{2+2p}|r^2 \snabla \snabla \alpha|^2) 
 \leq -(2p-\varepsilon)(1-\mu)r^{1+2p}|r^2 \snabla \snabla\alpha|^2+ \varepsilon^{-1}(1-\mu)r^{2p-3}|r^2 \snabla \snabla\phi|^2.
\end{equation}
We let $\tilde v \geq v_0$, $\tilde u \geq u_0$, and $\tilde v-\tilde u = 2(\tilde r)_*$. Let $r_0$ such that $\tilde v-u_0 = 2(r_0)_*$. We now integrate Equation~(\ref{eq:opb}) and (\ref{eq:opaa}) on $\conminus_{\tilde v} \cap \{r\geq \tilde r\} \cap \{u \geq u_0\}$. We obtain:
\begin{equation}\label{eq:masteradec}
\begin{aligned}
\int_{\mathbb{S}^2}{\tilde r}^{2+2p}(|{\tilde r}^2\snabla \snabla \alpha|^2(\tilde u, \tilde v, \omega)+ |\alpha|^2(\tilde u, \tilde v, \omega))\desphere(\omega) \\  \lesssim 
{r_0}^{2+2p}\int_{\mathbb{S}^2}(|r_0^2\snabla \snabla \alpha|^2(u_0, \tilde v, \omega)+ |\alpha|^2(\tilde u, \tilde v, \omega)) \desphere(\omega)\\
+ \varepsilon^{-1} \int_{\conminus_{\tilde v} \cap \{r\geq \tilde r\} \cap \{u \geq u_0\}} r^{2p-3}(|\phi|^2(u,\tilde v)+|r^2 \snabla \snabla\phi|^2(u, \tilde v))\desphere \de u
\end{aligned}
\end{equation}
Recall the definition of the cutoff $\chi$. It is a smooth function such that $\chi(r) =1$ for $r \geq 3M$, and $\chi(r) = 0$ for $r \in [2M, 5/2M]$.

Letting $\overline \Psi := \chi(r)(1-\mu)^{-1} r^3 \alpha$, we trivially bound the initial boundary term by data:
\begin{equation}\label{eq:boundain}
|\alpha| \leq C r^{-\frac 5 2} \snorm{\overline \Psi}{0}{0}{2}.
\end{equation}
Furthermore, we bound the integrals of $\phi$ appearing in the RHS of (\ref{eq:masteradec}) by Lemma~\ref{lem:fluxa}. We first write, using the definition of Lie derivative:
\begin{align*}
\int_{\conminus_{\tilde v} \cap \{r\geq \tilde r\} \cap \{u \geq u_0\}} r^{2p-3}(|\phi|^2(u,\tilde v)+|r^2 \snabla \snabla\phi|^2(u, \tilde v))\desphere \de u \\
 \leq C \int_{\conminus_{\tilde v} \cap \{r\geq \tilde r\}\cap \{u \geq u_0\}} r^{2p-3}(|\phi|^2(u,\tilde v)+ r^2|\snabla\phi|^2(u, \tilde v) + r^2 \sum_{j=1}^3|\snabla \slie_{\Omega_j} \phi|^2(u, \tilde v))\desphere \de u.
\end{align*}
From here, we consider two cases:
\begin{itemize}
\item If $\tilde u \leq \tilde v/2$, estimate (\ref{eq:fluxa1}) is sufficient to conclude. Let $\tilde r$ such that $\tilde r_* = \frac 1 2(\tilde v- \tilde u)$. Then,
\begin{equation}\label{eq:mastermod}
\begin{aligned}
\int_{\mathbb{S}^2}{\tilde r}^{5}(|{\tilde r}^2\snabla \snabla \alpha|^2(\tilde u, \tilde v, \omega)+ |\alpha|^2(\tilde u, \tilde v, \omega))\desphere(\omega)  \\ \lesssim 
{r_0}^{5}\int_{\mathbb{S}^2}(|r_0^2\snabla \snabla \alpha|^2(u_0, \tilde v, \omega)+ |\alpha|^2(\tilde u, \tilde v, \omega)) \desphere(\omega)\\
+ \varepsilon^{-1} 
\int_{\conminus_{\tilde v} \cap \{r\geq \tilde r\} \cap \{u\geq u_0\} } (|\phi|^2(u,\tilde v)+ r^2|\snabla\phi|^2(u, \tilde v) + r^2 \sum_{j=1}^3|\snabla \slie_{\Omega_j} \phi|^2(u, \tilde v))\desphere \de u
\\ \lesssim \snorm{\overline \Psi}{0}{0}{2}^2+\fluxob_0[\phi](u_0, \tilde v) + \fluxob_2[\snabla \phi](u_0, \tilde v) + \sum_{i=1}^3\fluxob_2[\snabla \slie_{\Omega_i}\phi](u_0, \tilde v) \\ \lesssim \snorm{\overline \Psi}{0}{0}{2}^2+\snorm{\phi}{2}{0}{0}^2 + \sum_{i=1}^3\snorm{\slie_{\Omega_i}\phi}{2}{0}{0}^2 \\ \lesssim \snorm{\overline \Psi}{0}{0}{2}^2+\snorm{\phi}{2}{0}{1}^2.
\end{aligned}
\end{equation}
Note that, if $\tilde u \leq \tilde v/2$, then $\tilde r_* = \frac 1 2 (\tilde v - \tilde u) \geq \frac 1 4 v$.
The Sobolev embedding for one-forms (see Lemma~\ref{lem:sobone} in the Appendix) then implies the bound
\begin{equation*}
|\alpha(u,v,\omega)| \leq C (\snorm{\overline \Psi}{0}{0}{2}+\snorm{\phi}{2}{0}{1}) v^{-\frac 5 2}, 
\end{equation*}
in the region $\{r \geq R_*\} \cap \{u \leq v/2\}$.
\item If $\tilde u \geq \tilde v/2$, we consider again Equation~(\ref{eq:masteradec}). Notice that $r_0 \gtrsim (r_0)_* = \frac 1 2 (\tilde v - u_0) \gtrsim \tilde v/2$. We estimate
\begin{equation}
\begin{aligned}
\int_{\mathbb{S}^2}{\tilde r}^{2+2p}(|{\tilde r}^2\snabla \snabla \alpha|^2(\tilde u, \tilde v, \omega)+ |\alpha|^2(\tilde u, \tilde v, \omega))\desphere(\omega)  \\
\lesssim {r_0}^{2p+2-5}\underbrace{\int_{\mathbb{S}^2} r_0^5(|r_0^2\snabla \snabla \alpha|^2(u_0, \tilde v, \omega)+ |\alpha|^2(u_0, \tilde v, \omega)) \desphere(\omega)}_{(i)}
\\
\underbrace{+\fluxob_{2p-3}[\phi](\tilde v /2, \tilde v) + \fluxob_{2p-1}[\snabla \phi](\tilde v /2, \tilde v) + \sum_{i=1}^3\fluxob_{2p-1}[\snabla \slie_{\Omega_i}\phi](\tilde v /2, \tilde v)}_{(ii)}
\\\underbrace{+\int_{\conminus_{\tilde v} \cap \{\tilde v /2 \geq u \geq u_0\} } r^{2p-3}|\phi|^2(u,\tilde v)+ r^{2p-1}|\snabla\phi|^2(u, \tilde v) \desphere \de u}_{(iii)} \\+\underbrace{\int_{\conminus_{\tilde v} \cap \{\tilde v /2 \geq u \geq u_0\} } r^{2p-1} \sum_{j=1}^3|\snabla \slie_{\Omega_j} \phi|^2(u, \tilde v)\desphere \de u}_{(iv)}.
\end{aligned}
\end{equation}
Now, by inequality (\ref{eq:boundain}),
$$
(i) \lesssim \snorm{\overline \Psi}{0}{0}{2}^2.
$$
Concerning $(ii)$, we have, by Lemma~\ref{lem:fluxa}, that
\begin{align*}
\fluxob_{2p-3}[\phi](\tilde v /2, \tilde v)+\fluxob_{2p-1}[\snabla \phi](\tilde v /2, \tilde v) &\lesssim v^{2p-3} \snorm{\phi}{2}{3-2p}{0}^2,\\
\sum_{i=1}^3\fluxob_{2p-1}[\snabla \slie_{\Omega_i}\phi](\tilde v /2, \tilde v) &\lesssim v^{2p-3} \snorm{\slie_{\Omega_i}\phi}{2}{3-2p}{1}^2.
\end{align*}
Concerning $(iii)$, we have that the $r$-coordinate of the point of coordinates $(\tilde v, \tilde v/2)$ is such that $r(\tilde v, \tilde v/2) \gtrsim \tilde v, \tilde v/2 = \tilde v /2$. Then, if $p \in \{3/2,1,1/2\}$, $2p-3 \leq 0$, hence
\begin{equation*}
\begin{aligned}
(iii) \lesssim (r(\tilde v, \tilde v/2))^{2p-3} \int_{\conminus_{\tilde v} \cap \{\tilde v /2 \geq u \geq u_0\} }( |\phi|^2(u,\tilde v)+ r^{2}|\snabla\phi|^2(u, \tilde v)) \desphere \de u \\
\lesssim (\tilde v)^{2p-3} \snorm{\phi}{2}{0}{0}.
\end{aligned}
\end{equation*}
Similarly,
\begin{equation*}
(iv) \lesssim (\tilde v)^{2p-3} \snorm{\phi}{2}{0}{1}.
\end{equation*}
Using the Sobolev embedding, Lemma~\ref{lem:sobone} in the Appendix, setting $q = -2p+3$, we finally have
\begin{equation*}
\begin{aligned}
|\alpha(u,v,\omega)| \leq C (\snorm{\overline \Psi}{0}{0}{2}+ \snorm{\phi}{2}{q}{1}) v^{-q/2} r^{-3/2}
\end{aligned}
\end{equation*}
in the region $\{r \geq R_*\} \cap \{u \geq v/2\}$, for $q \in \{0, 1, 2\}$.
\end{itemize}
Summarizing, we have the claim (\ref{eq:da3}). This concludes Step 1.

\subsection*{Step 2: estimates on $\alpha$, region of bounded $r$}
We let $2M < r_c < R$, $v \geq v_0$ $u \geq u_0$, and we integrate equations (\ref{eq:opaa}) and (\ref{eq:opb}) on
\begin{equation*}
\conminus_v \cap \{r_c \leq r \leq R\}.
\end{equation*}
We already know by estimate (\ref{eq:mastermod}) in Step 1 that there exists a constant $C$ depending only on $R$ such that
\begin{equation*}
\int_{\mathbb{S}^2}(|\snabla \snabla \alpha|^2(u, v, \omega)+ |\alpha|^2(u, v, \omega))\desphere \leq C_{R} v^{-2}(\snorm{\overline \Psi}{0}{0}{2}^2+\snorm{\phi}{2}{2}{1}^2).
\end{equation*}
Here, $u$ and $v$ are such that $v-u = 2R_*$. From the estimates on the flux of $\phi$ (\ref{eq:decaydue}), (\ref{eq:decaylie}), and the Sobolev embedding \ref{lem:sobone} we finally obtain
\begin{equation}
|\alpha|(u,v,\omega) \leq C v^{-1}(\snorm{\overline \Psi}{0}{0}{2}+\snorm{\phi}{2}{2}{1}),
\end{equation}
where $(u,v,\omega) \in \{u \geq u_0\} \cap \{v \geq v_0\} \cap \{r \leq R\}$.
\subsection*{Step 3: estimates on $\alphabar$, region of bounded $r$}
We set 
\begin{equation*}
\widetilde{\alphabar}:= (1-\mu)^{-1}\alphabar.
\end{equation*}
It follows that
\begin{equation}\label{eq:abarnorm}
\begin{aligned}
\snabla_L((1-\mu)^{-2}r^2|\alphabar|^2)  = 
L ((1-\mu)^{-2}) r^2 |\alphabar|^2+ 2(1-\mu)^{-2} r \alphabar^A \snabla_L(r \alphabar_A )\\ = 
- 2 (1-\mu)^{-2} \frac{2M}{r^2} r^2 |\alphabar|^2 + 2(1-\mu)^{-2} r \alphabar^A \snabla_L(r \alphabar_A ).
\end{aligned}
\end{equation}
Recall the definition of $\phibar$:
$$
\phibar_A = \frac{r^2}{1-\mu}\snabla_L(r \alphabar_A).$$
This implies:
\begin{equation}
\begin{aligned}
\snabla_L(r^2|\widetilde \alphabar|^2) = 
- 4M |\widetilde \alphabar|^2 + 2 r^{-1} \widetilde{\alphabar}^A\phibar_A.
\end{aligned}
\end{equation}
From this, it follows that
\begin{equation}\label{eq7.1}
\snabla_L(r^2|\widetilde \alphabar|^2 ) + 4M |\widetilde \alphabar|^2 \leq 4M \varepsilon |\widetilde \alphabar|^2 + \frac {1} {4M r^2 \varepsilon} |\phibar|^2.
\end{equation}
Defining now $A^2(u,v) := \int_{\mathbb{S}^2} |\widetilde \alphabar|^2(u,v,\omega) \desphere(\omega)$, we have the same equation for $A^2$:
\begin{equation}
\snabla_L(r^2A^2) + 4M A^2 \leq 4M \varepsilon A^2 + \frac {1} {4M r^2 \varepsilon} \int_{\mathbb{S}^2}|\phibar|^2\desphere.
\end{equation}
We integrate this inequality on the interval $[v_0, v]$. We obtain
\begin{equation}\label{eq:togronwall}
\begin{aligned}
L \left(\exp\left(\int_{v_0}^v \frac {4M (1-\varepsilon)}{r(u,v')^2} \de v'\right) r(u,v)^2A(u,v)^2\right) \\ \leq
\frac{1}{4M r(u,v)^2 \varepsilon} \exp\left(\int_{v_0}^v \frac {4M (1-\varepsilon)}{r(u,v')^2} \de v'\right) \int_{\mathbb{S}^2}|\phibar|^2\desphere.
\end{aligned}
\end{equation}
Consider $u$ as fixed. Since we restrict to the region $\{r \leq R\}$, we have that the function
\begin{equation}
F(u,v) := \int_{v_0}^v \frac {4M (1-\varepsilon)}{r(u,v')^2} \de v'
\end{equation}
is monotonically increasing, and satisfies the inequalities
\begin{equation}\label{eq:estsonf}
\frac{4M (1-\varepsilon)}{R^2}(v-v_0) \leq F(u,v) \leq M^{-1}(1-\varepsilon)(v-v_0).
\end{equation}
By integrating inequality (\ref{eq:togronwall}), we obtain that there exists a constant $C_R$ such that
\begin{equation}\label{eq7.4}
\begin{aligned}
\exp(F(v))r(u,v)^2 A(u,v)^2  - r(u,v_0)^2 A(u,v_0)^2 \\
\leq C_R \int_{\mathbb{S}^2}\int_{v_0}^v (r(u,v'))^{-2} \exp(F(u,v'))|\phibar|^2 \de v' \desphere.
\end{aligned}
\end{equation}
We split the integral on the right hand side in two: 
\begin{equation}
\begin{aligned}
\int_{\mathbb{S}^2}\int_{v_0}^v \exp(F(u,v'))(r(u,v'))^{-2}|\phibar|^2 \de v'\desphere \\
= \int_{\mathbb{S}^2}\int_{v_0}^{v/{A_0}} \exp(F(u,v'))(r(u,v'))^{-2}|\phibar|^2 \de v'\desphere \\ + 
\int_{\mathbb{S}^2}\int_{v/{A_0}}^v \exp(F(u,v'))(r(u,v'))^{-2}|\phibar|^2 \de v'\desphere,
\end{aligned}
\end{equation}
with $A_0 > 1$.
We subsequently claim that:
\begin{equation}
\exp(-F(u,v))\int_{\mathbb{S}^2}\int_{v_0}^{v/{A_0}} \exp(F(u,v'))(r(u,v'))^{-2}|\phibar|^2 \de v' \desphere \leq C \snorm{\phibar}{0}{0}{0}^2 v^{-2},\label{eq:claimu}
\end{equation}
and
\begin{equation}\label{eq:claimdu}
\exp(-F(u,v))\int_{\mathbb{S}^2}\int_{v/{A_0}}^v \exp(F(u,v'))(r(u,v'))^{-2}|\phibar|^2 \de v' \desphere \leq C \snorm{\phibar}{2}{2}{1} v^{-2}.
\end{equation}

\begin{itemize}
\item In order to prove (\ref{eq:claimu}), we recall the inequalities (\ref{eq:estsonf}), which we write in the following way:
\begin{equation}
\begin{aligned}
A_1 v + A_2 \leq F(u,v) \leq A_3 v + A_4,
\end{aligned}
\end{equation}
with $A_1 < A_3$ positive constants. 
Choose $A_0 > 0$ so that 
\begin{equation*}
-A_1+A_3/{A_0} < 0.
\end{equation*}
Note that $A_0$ depends only on the value of $R$.
Then,
\begin{equation}
\begin{aligned}
\exp(-F(u,v))\int_{\mathbb{S}^2}\int_{v_0}^{v/{A_0}} \exp(F(u,v'))(r(u,v'))^{-2}|\phibar(u,v')|^2 \de v' \desphere  \\  \leq
\exp(-F(u,v) + F(u,v/{A_0})) \int_{\mathbb{S}^2}\int_{v_0}^\infty (r(u,v'))^{-2}|\phibar(u,v')|^2 \de v'  \desphere \\ \lesssim 
\exp((- A_1+A_3/{A_0})v ) \int_{\mathbb{S}^2}\int_{v_0}^\infty (r(u,v'))^{-2}|\phibar(u,v')|^2 \de v'  \desphere.
\end{aligned}
\end{equation}
We now use the fact that $\phibar$ satisfies the \fackip equation: we have the energy conservation statement (\ref{eq:enconse}). We obtain the first claim (\ref{eq:claimu}).

\item For the claim (\ref{eq:claimdu}), we first estimate, given that, for fixed $u$, $F(u,v)$ is non-decreasing in $v$,
\begin{equation}
\begin{aligned}
\exp(-(F(v)))\int_{\mathbb{S}^2}\int_{v/{A_0}}^v \exp(F(v'))(r(u,v'))^{-2}|\phibar|^2 \de v' \desphere \\ \leq  \int_{\mathbb{S}^2}\int_{v/{A_0}}^v(r(u,v'))^{-2}|\phibar|^2 \de v'\desphere.
\end{aligned}
\end{equation}
We consequently notice that, by the energy conservation statement (\ref{eq:encons}), the resulting flux satisfies
\begin{equation}
\int_{\mathbb{S}^2}\int_{v/{A_0}}^v(r(u,v'))^{-2}|\phibar|^2 \de v' \desphere \leq C \fara^{\infty}[\phibar](Q(v)).
\end{equation}
where the point $Q(v) = ((A_0)^{-1}v-2R_*,(A_0)^{-1}v)$.
We now use the bound in Lemma~\ref{lem:fluxdec} in order to obtain
\begin{equation}
\int_{\mathbb{S}^2}\int_{v/{A_0}}^v(r(u,v'))^{-2}|\phibar|^2 \de v' \desphere \leq C v^{-2} \snorm{\phibar}{2}{2}{1}^2.
\end{equation}
This is claim (\ref{eq:claimdu}).
\end{itemize}
We now use previous Equation~(\ref{eq7.4}), together with the fact that 
$$
r(u,v_0)^2 A(u,v_0)^2 \leq C_R \snorm{\widetilde{\alphabar}}{0}{0}{0}
$$
on the region $\{v \geq v_0\} \cap \{r \leq R\}$. We obtain finally that $A$ satisfies, in the region $\mathcal{J}^+(C_{u_0, v_0}) \cap\{r \leq R\}$,
\begin{equation*}
A \leq C v^{-1}(\snorm{\widetilde{\alphabar}}{0}{0}{0} + \snorm{\phibar}{2}{2}{0}).
\end{equation*}
We notice that Equation~(\ref{eq7.1}) holds with $\widetilde \alphabar$ replaced by $r^2 \snabla \snabla \widetilde \alphabar$ (this follows by taking the $r\snabla$ derivative of the defining relation of $\phibar$). Using the Sobolev embedding (Lemma~\ref{lem:sobone}), we obtain the decay for $\alpha$ in the region $\{r \leq R\}$:
\begin{equation}
|\alphabar| \leq C (1-\mu) v^{-1}(\snorm{\phibar}{2}{2}{1}+\snorm{(1-\mu)^{-1}\alphabar}{0}{0}{2}).
\end{equation}

\subsection*{Step 4: estimates on $\alphabar$, region of unbounded $r$.}

Setting $p=0$ in previous Equation~(\ref{eq:abarnorm}) we obtain, using Cauchy-Schwarz on the right:

\begin{equation}\label{eq5.2}
\snabla_L(r^2A^2) + 4M |\widetilde \alphabar|^2 \leq 4M \varepsilon |\widetilde \alphabar|^2 + \frac {1} {r^2 4M \varepsilon} |\phibar|^2.
\end{equation}
We integrate this equation on cones of constant $u$ coordinate, starting from $\{r = R\}$. We use the Sobolev embedding \ref{lem:sobone}, and the estimates (\ref{eq:decaydue}). We finally obtain:
\begin{equation}
|\alphabar| \leq \frac{C}{(|u|+1) r}(\snorm{\phibar}{2}{2}{1}+\snorm{(1-\mu)^{-1} \alphabar}{0}{0}{2}),
\end{equation}
on the region $\{r \geq R\}\cap \mathcal{J}^+(u_0, v_0)$. This concludes the proof of the Proposition.
\end{proof}

\section{Improved decay for $\alpha$} \label{sec:alphaimproved}

In this Section, we prove estimate (\ref{eq:da4}) of Theorem~\ref{prop:decayteu}. We state again the result, as it is of independent interest. It is essentially an application of the $r^p$ method to the spin $+1$ Teukolsky equation.

\begin{proposition}\label{prop:morealpha}
There exists $R_* > 0$ and a constant $C > 0$ such that the following holds. Let $\alpha$ satisfy the spin $+1$ Teukolsky Equation on $\{ u \geq u_0\} \cap \{v \geq v_0 \}$,  with $v_0 - u_0 = 2R_*$.
Let $\chi$ be a smooth cutoff function such that $\chi(r) = 1$ for $r \geq 3M$, and $\chi(r) = 0$ for $r \in [2M, 3/2M]$. 

We let $\overline \Psi := \chi(r) (1-\mu)^{-1}r^3 \alpha$. In these conditions, we have the following bound, valid in the region $\{r \geq R\}$:
\begin{equation*}
\begin{aligned}
|\alpha| \leq C (|u|+1)^{-\frac 1 2} r^{-3} (\snorm{\overline \Psi}{2}{0}{2}+\snorm{\phi}{2}{2}{2} ).
\end{aligned}
\end{equation*}
\end{proposition}

In order to prove the Proposition, we make use of the following Lemma. The Lemma gives decay estimates on the boundary terms at $\{r = R\}$. 
\begin{lemma}\label{lem:prevmoredec}
Under the assumptions of Proposition~\ref{prop:morealpha}, we have the following inequality, valid for $u_1 \geq u_0$:
\begin{equation}\label{eq:masterrconst}
\begin{aligned}\int_{\conplus_{u_1}\cap \{r \geq R\}} (r^2|\alpha|^2 + r^2|r \snabla \alpha|^2)  \de v \desphere\\
\int_{\{r = R \} \cap \{ u \geq u_1\}} (|\snabla \alpha|^2+|\alpha|^2)\de \mathcal{T}
 + \int_{\mathfrak{D}_{u_1}^{\infty}} (|r \snabla \alpha|^2 + |\alpha|^2) \de u \de v \desphere\\
 \leq (1+|u_1|)^{-2} \sum_{i = 0}^2 \left( \fara^\infty [(\snabla_T)^i \phi](P(u_0))+\int_{\conplus_{u_0} \cap \{r \geq R\}}r^2 |\snabla_L (\snabla_T)^{\min\{i,1\}} \phi|^2\de v \desphere\right) \\
 + (1+|u_1|)^{-2}  \int_{\conplus_{u_0}\cap \{r \geq R\}} (r^4|\alpha|^2 + r^4|r \snabla \alpha|^2)  \de v \desphere.
\end{aligned}
\end{equation}
Here, as usual, $\de \mathcal{T}$ is the induced volume form on the hypersurface $\{r = R\}$. 

Recall now the definition of angular multi-indices and repeated Lie derivative of (\ref{eq:replie}). Letting $I \in \iota^\Omega_{\leq 2}$, the estimate (\ref{eq:masterrconst}) holds when all the occurrences of the symbol $\alpha$ in (\ref{eq:masterrconst}) are replaced with $\slie^I \alpha$, and all the occurrences of the symbol $\phi$ in (\ref{eq:masterrconst}) are replaced with $\slie^I \phi$.
\end{lemma}

\begin{proof}[Proof of Lemma~\ref{lem:prevmoredec}]
Let us notice that the quantity $\phi$ satisfies the following equation, upon substitution in the Teukolsky equation for $\alpha$ (\ref{eq:teua}):
\begin{equation}\label{eq:teuaref}
\snabla_L \phi_A - r^2 \sdelta (r\alpha_A) + r \alpha_A = 0.
\end{equation}
Let $f(r)$ be a smooth radial function. We now multiply Equation~(\ref{eq:teuaref}) by $(1-\mu) f(r) \phi^A$, and we obtain the identity, valid upon integration on $\mathbb{S}^2$:
\begin{equation}\label{eq:tomonot}
\begin{aligned}
\frac 1 2 L \left\{(1-\mu)f(r)|\phi|^2 \right\} - \frac 1 2 L \left\{(1-\mu)f(r) \right\}|\phi|^2\\
+\frac 1 2 \lbar \left\{f(r)r^2 |r^2 \snabla \alpha|^2   \right\} - \frac 1 2 \lbar\left(f(r)r^2 \right) |r^2 \snabla \alpha|^2\\
+ \frac 1 2 \lbar \left\{f(r)r^2 |r\alpha|^2   \right\} - \frac 1 2 \lbar \left(f(r) r^2 \right)|r \alpha|^2 \stackrel{\mathbb{S}^2}{=} 0.
\end{aligned}
\end{equation}
Let us then integrate with respect the form $\de u \de v$ in the spacetime region $\{r \geq R \} \cap \mathfrak{D}_{u_1}^{u_2} \cap \{v \leq v_{\max}\}$. We obtain, taking $v_\text{max} \to \infty$, and averaging in $\phi$, 
\begin{equation}\label{eq:fortau}
\begin{aligned}
\int_{\conplus_{u_2}\cap \{r \geq R\}} f(r)r^2(|r\alpha|^2 + |r^2 \snabla \alpha|^2)  \de v \desphere\\
-\int_{\conplus_{u_1}\cap \{r \geq R\}} f(r)r^2(|r\alpha|^2 + |r^2 \snabla \alpha|^2)  \de v \desphere\\
+ \int_{\{r = R \} \cap \{u_1 \leq u \leq u_2 \}} f(R)R^2 (|R^2 \snabla \alpha|^2+|R\alpha|^2)\de \mathcal{T}\\
 - \int_{\mathfrak{D}_{u_1}^{u_2}} \lbar(r^2 f(r) )(|r^2 \snabla \alpha|^2 + |r\alpha|^2) \de u \de v \desphere\\
 \leq C(f,R) \int_{\{R \leq r \leq R+M \} \cap \{u \geq u_1\}} (|\phi|^2 + |\snabla_L \phi|^2)\de u \de v \desphere \\
 + \int_{\mathfrak{D}_{u_1}^{u_2}} L((1-\mu)f(r))|\phi|^2 \de u \de v \desphere.
\end{aligned}
\end{equation}
Here, we supposed $f$ to be positive and smooth. $\de \mathcal{T}$ is, as before, the induced volume form on the hypersurface $\{r = R\}$.

We choose now $f(r) = (1-\mu)^{-1}$ in inequality (\ref{eq:fortau}), and we obtain
\begin{equation}\label{eq:alphahigh}
\begin{aligned}
\int_{\conplus_{u_2}\cap \{r \geq R\}} (r^4|\alpha|^2 + r^4|r \snabla \alpha|^2)  \de v \desphere\\
+ \int_{\{r = R \} \cap \{u_1 \leq u \leq u_2 \}} (| \snabla \alpha|^2+|\alpha|^2)\de \mathcal{T}
 + \int_{\mathfrak{D}_{u_1}^{u_2}} r(|r^2 \snabla \alpha|^2 + |r\alpha|^2) \de u \de v \desphere\\
\leq C \int_{\conplus_{u_1}\cap \{r \geq R\}} (r^4|\alpha|^2 + r^4|r \snabla \alpha|^2)  \de v \desphere + C \fara^\infty[\phi](P(u_1)).
\end{aligned}
\end{equation}
Recall that, here, $P(u) = (u, u + 2R_*)$. This implies that, along a dyadic sequence $u_n$, we have
\begin{equation}\label{eq:decalphamid}
\begin{aligned}
\int_{\conplus_{u_n}\cap \{ r \geq R\}} r(|r\alpha|^2 + |r^2 \snabla \alpha|^2)\de v \desphere \\
\leq C u_n^{-1} \int_{\conplus_{u_0}\cap \{r \geq R\}} (r^4|\alpha|^2 + r^4|r \snabla \alpha|^2)  \de v \desphere + u_n^{-1}\fara^\infty[\phi](P(u_0)).
\end{aligned}
\end{equation}
We choose now $f(r) = r^{-1}(1-\mu)^{-1}$ in inequality (\ref{eq:fortau}), and we obtain, discarding the last term in (\ref{eq:fortau}) (indeed, $L( (1-\mu) f) < 0$):
\begin{equation}\label{eq:alphaprefin}
\begin{aligned}
\int_{\conplus_{u_2}\cap \{r \geq R\}} (r^3|\alpha|^2 + r^3|r \snabla \alpha|^2)  \de v \desphere\\
+ \int_{\{r = R \} \cap \{u_1 \leq u \leq u_2 \}} (|\snabla \alpha|^2+|\alpha|^2)\de \mathcal{T}
 + \int_{\mathfrak{D}_{u_1}^{u_2}} (|r^2 \snabla \alpha|^2 + |r\alpha|^2) \de u \de v \desphere\\
\leq C\int_{\conplus_{u_1}\cap \{r \geq R\}} (r^3|\alpha|^2 + r^3|r \snabla \alpha|^2)  \de v \desphere \\
+ C\int_{\{R \leq r \leq R+M \} \cap \{u \geq u_1\}} (|\phi|^2 + |\snabla_L \phi|^2)\de u \de v \desphere
\end{aligned}
\end{equation}
From inequality (\ref{gro2}) we now have, along a dyadic sequence $\tilde u_n$, that
\begin{equation*}
\tilde u_n \fara^\infty[\phi](P(\tilde u_n)) \leq C \underbrace{ \left(\fara^\infty[\phi](P(u_0))+\fara^\infty[\snabla_T \phi](P(u_0)) +  \int_{\conplus_{u_0} \cap \{r \geq R\}} |\snabla_L \phi|^2 r \de v \desphere  \right)}_{:= W[\phi](P(u_0))}.
\end{equation*}
By the Morawetz estimate for $\phi$, then, we have (without loss of generality, we can assume $r > 3M$):
\begin{equation}
\begin{aligned}
\int_{\{R \leq r \leq R+M \} \cap \{u \geq u_1\}} (|\phi|^2 + |\snabla_L \phi|^2)\de u \de v \desphere \\
\leq (1+|u_1|)^{-1} W[\phi](P(u_0)).
\end{aligned}
\end{equation}
We therefore have the inequality, valid for $u_2 \geq u_n$:
\begin{equation}\label{eq:decalphalowpre}
\begin{aligned}
\int_{\conplus_{u_2}\cap \{r \geq R\}} (r^3|\alpha|^2 + r^3|r \snabla \alpha|^2)  \de v \desphere\\
+ \int_{\{r = R \} \cap \{u_n \leq u \leq u_2 \}} (|\snabla \alpha|^2+|\alpha|^2)\de \mathcal{T}
 + \int_{\mathfrak{D}_{u_n}^{u_2}} (|r^2 \snabla \alpha|^2 + |r\alpha|^2) \de u \de v \desphere\\
\leq u_n^{-1} W[\phi](P(u_0)) + C u_n^{-1} \int_{\conplus_{u_0}\cap \{r \geq R\}} (r^4|\alpha|^2 + r^4|r \snabla \alpha|^2)  \de v \desphere.
\end{aligned}
\end{equation}
From the last display, in particular, it follows that, for $u_1 \geq u_0$,
\begin{equation*}
\begin{aligned}
\int_{\conplus_{u_1}\cap \{r \geq R\}} (r^3|\alpha|^2 + r^3|r \snabla \alpha|^2)  \de v \desphere + \int_{\{r = R \} \cap \{u \geq u_1\}} (|\snabla \alpha|^2+|\alpha|^2)\de \mathcal{T}\\
\leq C u_1^{-1} W[\phi](P(u_0)) + C u_1^{-1} \int_{\conplus_{u_0}\cap \{r \geq R\}} (r^4|\alpha|^2 + r^4|r \snabla \alpha|^2)  \de v \desphere.
\end{aligned}
\end{equation*}
From (\ref{eq:decalphalowpre}) it furthermore follows, along a dyadic sequence $\bar u_n$,
\begin{equation}\label{eq:decalphlow}
\begin{aligned}
\int_{\conplus_{\bar u_n}\cap \{ r \geq R\}} (|r\alpha|^2 + |r^2 \snabla \alpha|^2)\de v \desphere \\
\leq C \bar u_n^{-2} \int_{\conplus_{u_0}\cap \{r \geq R\}} (r^4|\alpha|^2 + r^4|r \snabla \alpha|^2)  \de v \desphere + \bar u_n^{-2}W[\phi](P(u_0)).
\end{aligned}
\end{equation}
We choose now $f(r) = r^{-2}(1-\mu)^{-1}$ in inequality (\ref{eq:fortau}), and we obtain finally
\begin{equation}\label{eq:alphafin}
\begin{aligned}
\int_{\conplus_{u_2}\cap \{r \geq R\}} (r^2|\alpha|^2 + r^2|r \snabla \alpha|^2)  \de v \desphere\\
+ \int_{\{r = R \} \cap \{u_1 \leq u \leq u_2 \}} (|\snabla \alpha|^2+|\alpha|^2)\de \mathcal{T}
 + \int_{\mathfrak{D}_{u_1}^{u_2}} (|r \snabla \alpha|^2 + |\alpha|^2) \de u \de v \desphere\\
\leq \int_{\conplus_{u_1}\cap \{r \geq R\}} (r^2|\alpha|^2 + r^2|r \snabla \alpha|^2)  \de v \desphere \\
+ \int_{\{R \leq r \leq R+M \} \cap \{u \geq u_1\}} (|\phi|^2 + |\snabla_L \phi|^2)\de u \de v \desphere
\end{aligned}
\end{equation}
We now use the Morawetz estimate for $\phi$, as well as the flux decay in Lemma~\ref{lem:fluxdec}, to obtain
\begin{equation}\label{eq:prelprel}
\begin{aligned}
\int_{\{r = R \} \cap \{u \geq \bar u_n\}} (|\snabla \alpha|^2+|\alpha|^2)\de \mathcal{T}
 + \int_{\mathfrak{D}_{\bar u_n}^{\infty}} (|r \snabla \alpha|^2 + |\alpha|^2) \de u \de v \desphere\\
 \leq \bar u_n^{-2} \sum_{i = 0}^2 \left( \fara^\infty [(\snabla_T)^i \phi](P(u_0))+\int_{\conplus_{u_0} \cap \{r \geq R\}}r^2 |\snabla_L (\snabla_T)^{\min\{i,1\}} \phi|^2\de v \desphere\right) \\
 + \bar u_n^{-2} \int_{\conplus_{u_0}\cap \{r \geq R\}} (r^4|\alpha|^2 + r^4|r \snabla \alpha|^2)  \de v \desphere.
\end{aligned}
\end{equation}
It is trivial to remove the restriction to the dyadic sequence, due to the monotonicity of the fluxes considered on the left hand side of (\ref{eq:prelprel}).

Similarly, it is straightforward to deduce the decay estimate, valid for all $u_1 \geq u_0$:
\begin{equation}\label{eq:prelfin}
\begin{aligned}
\int_{\conplus_{u_1}\cap \{r \geq R\}} (r^2|\alpha|^2 + r^2|r \snabla \alpha|^2)  \de v \desphere\\
\leq (1+|u_1|)^{-2} \sum_{i = 0}^2 \left( \fara^\infty [(\snabla_T)^i \phi](P(u_0))+\int_{\conplus_{u_0} \cap \{r \geq R\}}r^2 |\snabla_L (\snabla_T)^{\min\{i,1\}} \phi|^2\de v \desphere\right) \\
 + (1+|u_1|)^{-2} \int_{\conplus_{u_0}\cap \{r \geq R\}} (r^4|\alpha|^2 + r^4|r \snabla \alpha|^2)  \de v \desphere.
\end{aligned}
\end{equation}
The Lemma is thus proved by combining the $\slie^I$-commuted versions of displays (\ref{eq:prelprel}) and (\ref{eq:prelfin}).
\end{proof}

\begin{proof}[Proof of Proposition~\ref{prop:morealpha} and of (\ref{eq:da4})]
Let us consider the Teukolsky Equation for $\alpha$ (\ref{eq:teua}), and write it in the following way:
\begin{equation}
\begin{aligned}
 &\snabla_L \left(\frac{r^2}{1-\mu}\snabla_\lbar(r \alpha_A) \right) - r^2 \sdelta (r \alpha_A) + r \alpha_A = 0 \\
 \iff &\snabla_L \snabla_\lbar \left(\frac{r^2}{1-\mu}r \alpha_A \right) - \snabla_L\left(\lbar\left(\frac{r^2}{1-\mu} \right) r \alpha_A\right) - r^2 \sdelta(r \alpha_A) + r \alpha_A = 0 \\
 \iff &\snabla_L \snabla_\lbar \left(\frac{r^2}{1-\mu}r \alpha_A \right) - \snabla_L\left(\frac{1-\mu}{r^2}\lbar\left(\frac{r^2}{1-\mu} \right) \frac{r^2}{1-\mu}r \alpha_A\right) \\
 &- r^2 \sdelta(r \alpha_A) + r \alpha_A = 0.
\end{aligned}
\end{equation}
Letting $\Psi_A := (1-\mu)^{-1}r^3 \alpha_A$, we have the following equation for $\Psi_A$:
\begin{equation}
\snabla_L \snabla_\lbar \Psi_A - \snabla_L \left(\frac{1-\mu}{r^2} \lbar\left(\frac{r^2}{1-\mu}\right)\Psi_A \right) - (1-\mu) \sdelta \Psi_A + \frac{1-\mu}{r^2} \Psi_A = 0,
\end{equation}
which implies
\begin{equation}
\begin{aligned}
&\snabla_\lbar \snabla_L \Psi_A - L \left(\frac{1-\mu}{r^2} \lbar\left(\frac{r^2}{1-\mu}\right)\right)\Psi_A - \left(\frac{1-\mu}{r^2} \lbar\left(\frac{r^2}{1-\mu}\right)\right)\snabla_L \Psi_A \\
&- (1-\mu) \sdelta \Psi_A + \frac{1-\mu}{r^2} \Psi_A = 0.
\end{aligned}
\end{equation}

Let $f = f(u,v)$ be a smooth function. Multiply the last display through by $f(u,v) \snabla_L \Psi^A$. We obtain
\begin{align*}
&\hspace{15pt} f(u,v) \snabla_L \Psi^A \snabla_\lbar \snabla_L \Psi_A = \frac 1 2 \lbar(f(u,v) |\snabla_L \Psi|^2) - \frac 1 2 \lbar (f(u,v)) |\snabla_L \Psi|^2,\\[10pt]
& - f(u,v) \snabla_L \Psi^A  L \left(\frac{1-\mu}{r^2} \lbar\left(\frac{r^2}{1-\mu}\right)\right)\Psi_A  = - f(u,v)\frac 1 2  L \left(\frac{1-\mu}{r^2} \lbar\left(\frac{r^2}{1-\mu}\right)\right) L |\Psi|^2 \\
&= - \frac 1 2 L \left\{ f(u,v) L \left(\frac{1-\mu}{r^2} \lbar\left(\frac{r^2}{1-\mu}\right)\right) |\Psi|^2  \right\} \\ 
&\hspace{15pt}+ \frac 1 2 L \left\{ f(u,v) L \left(\frac{1-\mu}{r^2} \lbar\left(\frac{r^2}{1-\mu}\right)\right)  \right\}|\Psi|^2,\\[10pt]
&- f(u,v) \left(\frac{1-\mu}{r^2} \lbar\left(\frac{r^2}{1-\mu}\right)\right)\snabla_L \Psi_A \snabla_L \Psi^A = - f(u,v) \left(\frac{1-\mu}{r^2} \lbar\left(\frac{r^2}{1-\mu}\right)\right)|\snabla_L \Psi|^2,\\[10pt]
&- (1-\mu) f(u,v) \sdelta \Psi^A \snabla_L \Psi^A = -f(u,v) \frac{1-\mu}{r^2} r^2 \sdelta \Psi^A \snabla_L \Psi^A \\
&\stackrel{\mathbb{S}^2}{=} f(u,v) \frac{1-\mu}{r^2} \snabla_L (r \snabla_B \Psi_A ) (r \snabla^B \Psi^A)\\
 &= \frac 1 2 L \left\{ f(u,v) (1-\mu) |\snabla \Psi|^2 \right\} - \frac 1 2 r^2 L\left( f(u,v) \frac{1-\mu}{r^2}\right) |\snabla \Psi|^2,
\end{align*}
\begin{align*}
& f(u,v)  \frac{1-\mu}{r^2} \Psi^A \snabla_L \Psi_A = \frac 1 2 f(u,v) \frac{1-\mu}{r^2} L|\Psi|^2 \\
&= \frac 1 2 L \left( f(u,v) \frac{1-\mu}{r^2} |\Psi|^2 \right) - \frac 1 2 L \left(f(u,v) \frac{1-\mu}{r^2} \right)|\Psi|^2.
\end{align*}

We now proceed to integrate the resulting identity on the region $\mathfrak{D}_{u_1}^{u_2} \cap \{v \leq v_{\text{max}}\}$.
We notice the following, from the Poincar\'e inequality for one-forms (Lemma~\ref{lem:poinca}):
\begin{equation*}
\begin{aligned}
&- f(u,v) L\left\{\frac{1-\mu}{r^2}\lbar \left(\frac{r^2}{1-\mu} \right) \right\} |\Psi|^2+ f(u,v) \frac{1-\mu}{r^2} |\Psi|^2 + f(u,v) (1-\mu) |\snabla \Psi|^2  \\
&\geq f(u,v) \left\{  - L\left\{\frac{1-\mu}{r^2}\lbar \left(\frac{r^2}{1-\mu} \right) \right\} + 2 \frac{1-\mu}{r^2}  \right\}|\Psi|^2= 12 M r^{-3} f(u,v)(1-\mu) |\Psi|^2 \geq 0.
\end{aligned}
\end{equation*}
Hence, if $f$ is a positive function, we can discard the corresponding incoming null flux on $v_{\text{max}}$. 

We now choose $f(u,v) = (1-\mu)^{-1}v^2$. It is easy to verify that there exists a value $u_{\text{in}}$ such that the following holds, for $u \geq u_{\text{in}}$ and $u+v \geq 2R_*$:
\begin{equation*}
- L \left(\frac{v^2}{r^2}\right) \geq 0, \qquad - L \left(\frac{v^2}{r^3}\right) \geq 0
\end{equation*}
We now use the previous display, along with the Poincar\'e estimate for one-forms (Lemma~\ref{lem:poinca}) to obtain positivity of the bulk terms in $\Psi$, for $u \geq u_{\text{in}}$:
\begin{equation*}
\begin{aligned}
\frac 1 2 L \left\{ f(u,v) L \left(\frac{1-\mu}{r^2} \lbar\left(\frac{r^2}{1-\mu}\right)\right)  \right\} |\Psi| ^2 
- \frac 1 2 L \left(\frac{v^2}{r^2}\right) |\Psi|^2 - \frac 1 2 L \left(\frac{v^2}{r^2}\right) |\snabla \Psi|^2 \\
\geq \frac 1 2 L \left\{-2 \frac{v^2}{r^2}+ f(u,v) L \left(\frac{1-\mu}{r^2} \lbar\left(\frac{r^2}{1-\mu}\right)\right)  \right\} |\Psi| ^2 = \frac 1 2 L \left\{- 12 M \frac{v^2}{r^3} \right\} \geq 0.
\end{aligned}
\end{equation*}
We therefore obtain the following estimate, valid for $u_1, u_2 \geq u_\text{in}$:
\begin{equation}\label{eq:vweight}
\begin{aligned}
\int_{\conplus_{u_2} \cap \{r \geq R\}} v^2 |\snabla_L \Psi|^2 \de v \desphere 
- \int_{\conplus_{u_1} \cap \{r \geq R\}} v^2 |\snabla_L \Psi|^2 \de v \desphere \\
\leq C \int_{\{r = R\}\cap \{u_1 \leq u \leq u_2\}} (|\Psi|^2 + |\snabla \Psi|^2) \de \mathcal{T}.
\end{aligned}
\end{equation}
Similarly, we have the commuted version of the previous bound:
\begin{equation}\label{eq:vweightcom}
\begin{aligned}
\sum_{I \in \iota^\Omega_{\leq 2}} \int_{\conplus_{u_2} \cap \{r \geq R\}} v^2 |\snabla_L \slie^I \Psi|^2 \de v \desphere 
- \sum_{I \in \iota^\Omega_{\leq 2}} \int_{\conplus_{u_1} \cap \{r \geq R\}} v^2 |\snabla_L \slie^I \Psi|^2 \de v \desphere \\
\leq C \sum_{I \in \iota^\Omega_{\leq 2}} \int_{\{r = R\}\cap \{u_1 \leq u \leq u_2\}} (|\slie^I \Psi|^2 + |\snabla \slie^I \Psi|^2) \de \mathcal{T}.
\end{aligned}
\end{equation}
We now choose $f(u,v) = (1-\mu)^{-1}r^2$.
We notice that the only spacetime term remaining in either $|\Psi|^2$ or $|\snabla \Psi|^2$ is
\begin{equation}\label{eq:poszeroth}
\int_{\mathfrak{D}_{u_1}^{u_2}} L \left\{ f(u,v) L \left(\frac{1-\mu}{r^2} \lbar\left(\frac{r^2}{1-\mu}\right)\right)  \right\} |\Psi|^2\de u \de v \desphere.
\end{equation}
With our choice of $f$, we have
\begin{equation}
 L \left\{ f(u,v) L \left(\frac{1-\mu}{r^2} \lbar\left(\frac{r^2}{1-\mu}\right)\right)  \right\} = 12 M \frac{1-\mu}{r^2}
\end{equation}
Hence we obtain the following estimate, possibly restricting to $R$ large enough,
\begin{equation}\label{eq:rpalphad}
\begin{aligned}
\int_{\conplus_{u_2} \cap \{r \geq R\}} r^2 |\snabla_L \Psi|^2 \de v \desphere 
- \int_{\conplus_{u_1} \cap \{r \geq R\}} r^2 |\snabla_L \Psi|^2 \de v \desphere\\
+ \int_{\mathfrak{D}_{u_1}^{u_2}} r |\snabla_L \Psi|^2 \de u \de v \desphere 
+ \int_{\mathfrak{D}_{u_1}^{u_2}} r^{-2} |\Psi|^2 \de u \de v \desphere \\
\leq C_R \int_{\{r = R\}\cap \{u_1 \leq u \leq u_2\}} (|\Psi|^2 + |\snabla \Psi|^2) \de \mathcal{T}.
\end{aligned}
\end{equation}
Similarly, we obtain the commuted estimate
\begin{equation}\label{eq:rpalphadcom}
\begin{aligned}
\sum_{I \in \iota^\Omega_{\leq 2}} \left\{\int_{\conplus_{u_2} \cap \{r \geq R\}} r^2 |\snabla_L \slie^I\Psi|^2 \de v \desphere 
- \int_{\conplus_{u_1} \cap \{r \geq R\}} r^2 |\snabla_L \slie^I\Psi|^2 \de v \desphere \right.\\
\left.+ \int_{\mathfrak{D}_{u_1}^{u_2}} r |\snabla_L\slie^I \Psi|^2 \de u \de v \desphere 
+ \int_{\mathfrak{D}_{u_1}^{u_2}} r^{-2} |\slie^I\Psi|^2 \de u \de v \desphere \right\} \\
\leq C_R \sum_{I \in \iota^\Omega_{\leq 2}}\int_{\{r = R\}\cap \{u_1 \leq u \leq u_2\}} (|\slie^I\Psi|^2 + |\snabla \slie^I\Psi|^2) \de \mathcal{T}.
\end{aligned}
\end{equation}
By an analogous reasoning, we obtain the following estimate, choosing $f(r) = (1-\mu)^{-1} r$:
\begin{equation}\label{eq:rpalphau}
\begin{aligned}
\int_{\conplus_{u_2} \cap \{r \geq R\}} r |\snabla_L \Psi|^2 \de v \desphere 
- \int_{\conplus_{u_1} \cap \{r \geq R\}} r |\snabla_L \Psi|^2 \de v \desphere\\
+ \int_{\mathfrak{D}_{u_1}^{u_2}} |\snabla_L \Psi|^2 \de u \de v \desphere 
+ \int_{\mathfrak{D}_{u_1}^{u_2}}  r^{-3} |\Psi|^2 \de u \de v \desphere \\
\leq C_R \int_{\{r = R\}\cap \{u_1 \leq u \leq u_2\}} (|\Psi|^2 + |\snabla \Psi|^2) \de \mathcal{T}.
\end{aligned}
\end{equation}
We also obtain the corresponding commuted estimate:
\begin{equation}\label{eq:rpalphaucom}
\begin{aligned}
\sum_{I \in \iota^\Omega_{\leq 2}}\int_{\conplus_{u_2} \cap \{r \geq R\}} r |\snabla_L \slie^I \Psi|^2 \de v \desphere 
- \sum_{I \in \iota^\Omega_{\leq 2}}\int_{\conplus_{u_1} \cap \{r \geq R\}} r |\snabla_L \slie^I \Psi|^2 \de v \desphere\\
+ \sum_{I \in \iota^\Omega_{\leq 2}}\int_{\mathfrak{D}_{u_1}^{u_2}} |\snabla_L \slie^I \Psi|^2 \de u \de v \desphere 
+ \sum_{I \in \iota^\Omega_{\leq 2}}\int_{\mathfrak{D}_{u_1}^{u_2}}  r^{-3} |\slie^I \Psi|^2 \de u \de v \desphere \\
\leq C_R \sum_{I \in \iota^\Omega_{\leq 2}} \int_{\{r = R\}\cap \{u_1 \leq u \leq u_2\}} (|\slie^I \Psi|^2 + |\snabla \slie^I \Psi|^2) \de \mathcal{T}.
\end{aligned}
\end{equation}

We now choose $f(u,v) = (1-\mu)^{-1}$. We look again at the resulting combination of spacetime terms in either $|\Psi|^2$ or $|\snabla \Psi|^2$, and use the Poincar\'e inequality for one-forms:
\begin{equation}
\begin{aligned}
&\frac 1 2 L \left\{ (1-\mu)^{-1} L \left(\frac{1-\mu}{r^2} \lbar\left(\frac{r^2}{1-\mu}\right)\right) -r^{-2}\right\}|\Psi|^2 - \frac 1 2 r^2L(r^{-2})|\snabla \Psi|^2\\
&\geq \frac 1 2 L \left\{ (1-\mu)^{-1} L \left(\frac{1-\mu}{r^2} \lbar\left(\frac{r^2}{1-\mu}\right)\right) - 2 r^{-2}\right\}|\Psi|^2.
\end{aligned}
\end{equation}
We now notice:
\begin{equation}
\begin{aligned}
 L \left\{ (1-\mu)^{-1} L \left(\frac{1-\mu}{r^2} \lbar\left(\frac{r^2}{1-\mu}\right)\right) - 2 r^{-2}\right\} \\
 = (1-\mu)\partial_r \left\{-2 r^{-2} - \partial_r \left(\frac{(1-\mu)^2}{r^2} \partial_r \left(\frac{r^2}{1-\mu}\right)\right)  \right\}= (1-\mu)\frac{36M}{r^4}.
 \end{aligned}
\end{equation}
Hence, we obtain the following estimate:
\begin{equation}\label{eq:rpalphaz}
\begin{aligned}
\int_{\conplus_{u_2} \cap \{r \geq R\}} |\snabla_L \Psi|^2 \de v \desphere 
- \int_{\conplus_{u_1} \cap \{r \geq R\}} |\snabla_L \Psi|^2 \de v \desphere\\
+ \int_{\mathfrak{D}_{u_1}^{u_2}}r^{-1} |\snabla_L \Psi|^2 \de u \de v \desphere 
+ \int_{\mathfrak{D}_{u_1}^{u_2}} r^{-4} |\Psi|^2 \de u \de v \desphere \\
\leq C_R \int_{\{r = R\}\cap \{u_1 \leq u \leq u_2\}} (|\Psi|^2 + |\snabla \Psi|^2) \de \mathcal{T} .
\end{aligned}
\end{equation}
Also, we obtain the following commuted version of the previous estimate:
\begin{equation}\label{eq:rpalphazcom}
\begin{aligned}
\sum_{I \in \iota^\Omega_{\leq 2}} \left\{\int_{\conplus_{u_2} \cap \{r \geq R\}} |\snabla_L \slie^I \Psi|^2 \de v \desphere 
- \int_{\conplus_{u_1} \cap \{r \geq R\}} |\snabla_L \slie^I \Psi|^2 \de v \desphere\right. \\
\left.+ \int_{\mathfrak{D}_{u_1}^{u_2}}r^{-1} |\snabla_L \slie^I \Psi|^2 \de u \de v \desphere 
+ \int_{\mathfrak{D}_{u_1}^{u_2}} r^{-4} |\slie^I \Psi|^2 \de u \de v \desphere \right\}\\
\leq C_R \sum_{I \in \iota^\Omega_{\leq 2}} \int_{\{r = R\}\cap \{u_1 \leq u \leq u_2\}} (|\slie^I \Psi|^2 + |\snabla \slie^I \Psi|^2) \de \mathcal{T} .
\end{aligned}
\end{equation}
Notice that, in particular, from Lemma~\ref{lem:prevmoredec}, we have the following estimate:
\begin{equation}\label{eq:tfluxfirst}
\begin{aligned}
\sum_{I \in \iota^\Omega_{\leq 2}} \int_{\{u \geq u_0\} \cap \{ r = R \}}  (|\snabla \slie^I \alpha|^2 + |\slie^I \alpha|^2)\de v \desphere \\
\leq \sum_{I \in \iota^\Omega_{\leq 2}} \int_{\conplus_{u_0} \cap\{r \geq R\}} (r^{-4}|\slie^I\Psi|^2+ r^{-2}|\snabla \slie^I \Psi|^2) \de v \desphere + \sum_{I \in \iota^\Omega_{\leq 2}} \fara^\infty[\slie^I \phi](P(u_0)).
\end{aligned}
\end{equation}
Now, using inequalities (\ref{eq:vweightcom}), (\ref{eq:rpalphadcom}) and (\ref{eq:tfluxfirst}), we have the following uniform bound for the flux in $r^2 |\snabla_L \Psi|$:
\begin{equation} \label{eq:unibddue}
\sum_{I \in \iota^\Omega_{\leq 2}} \int_{\conplus_{u}\cap \{r \geq R\} } v^2 |\snabla_L \slie^I \Psi|^2 \de v \desphere \leq C( \snorm{\overline \Psi}{2}{0}{2}^2 + \snorm{\phi}{0}{0}{2}^2).
\end{equation}
Here, we used the definition $\overline \Psi := \chi(r) \Psi$, where $\chi$ is a smooth cutoff function as in the statement of the Proposition, i.e. such that $\chi(r) = 0$ for $r \in [2M, 5/2M]$, and $\chi(r) =1$ for $r \in [3M, \infty)$.

Now, from (\ref{eq:rpalphadcom}) it follows that there exists a dyadic sequence $u_n$ such that
\begin{equation}\label{eq:decpsiun}
\begin{aligned}
\sum_{I \in \iota^\Omega_{\leq 2}}\int_{\conplus_{u_n} \cap \{r \geq R\}} \left(r |\snabla_L \slie^I \Psi|^2 + r^{-2}|\slie^I \Psi|^2\right) \de v \desphere \\
 \leq C u_n^{-1}\left\{ \sum_{I \in \iota^\Omega_{\leq 2}} \int_{\conplus_{u_0} \cap\{r \geq R\}} (r^{-4}|\slie^I\Psi|^2+ r^{-2}|\snabla \slie^I \Psi|^2) \de v \desphere \right. \\
 \left.+  \sum_{I \in \iota^\Omega_{\leq 2}}\int_{\conplus_{u_0} \cap \{r \geq R\}} r^2 |\snabla_L \slie^I\Psi|^2 \de v \desphere + \sum_{I \in \iota^\Omega_{\leq 2}} \fara^\infty[\slie^I \phi](P(u_0))\right\}.
\end{aligned}
\end{equation}
Recall that, from inequality (\ref{eq:alphahigh}), the following bound holds for all $u \geq u_0$:
\begin{equation}\label{eq:decafin}
\begin{aligned}
 &\sum_{I \in \iota^\Omega_{\leq 2}} \int_{\conplus_{u} \cap \{r \geq R\}}(|\snabla \slie^I\Psi|^2 +  r^{-2}|\slie^I \Psi|^2) \de v \desphere \\
\leq & C (1+|u|)^{-1} \sum_{I \in \iota^\Omega_{\leq 2}} \left( W[\slie^I \phi](P(u_0)) \phantom{\int}\right.\\
&\left. + \int_{\conplus_{u_0}\cap \{r \geq R\}} (r^{-2}|\slie^I \Psi|^2 + |\snabla \slie^I \Psi|^2+r^2|\snabla_L \slie^I \Psi|^2)  \de v \desphere\right)
\end{aligned}
\end{equation}

We plug the sequence $\{u_n\}$ in estimate (\ref{eq:rpalphaucom}), use (\ref{eq:decafin}) to bound the terms on the right hand side, and we obtain that there exists a second dyadic sequence $\{\bar u_n\}$ such that there holds
\begin{equation}\label{eq:decpsidu}
\begin{aligned}
\sum_{I \in \iota^\Omega_{\leq 2}}\int_{\conplus_{\bar u_n} \cap \{r \geq R\}} \left( |\snabla_L \slie^I \Psi|^2 + r^{-3}|\slie^I \Psi|^2\right) \de v \desphere \\
 \leq C (|\bar u_n|+1)^{-2}\left\{ \sum_{I \in \iota^\Omega_{\leq 2}} \int_{\conplus_{u_0} \cap\{r \geq R\}} (r^{-4}|\slie^I\Psi|^2+ r^{-2}|\snabla \slie^I \Psi|^2) \de v \desphere \right. \\
 \left.+  \sum_{I \in \iota^\Omega_{\leq 2}}\int_{\conplus_{u_0} \cap \{r \geq R\}} r^2 |\snabla_L \slie^I\Psi|^2 \de v \desphere + \sum_{I \in \iota^\Omega_{\leq 2}} \fara^\infty[\slie^I \phi](P(u_0))\right.\\
\left. \phantom{\int} +\sum_{I \in \iota^\Omega_{\leq 2}} W[\slie^I \phi](P(u_0))\right\}.
\end{aligned}
\end{equation}
We now wish to remove the restriction to the dyadic sequence on the integral
\begin{equation*}
\int_{\conplus_{\bar u_n} \cap \{r \geq R\}} \left( |\snabla_L \slie^I \Psi|^2 + r^{-3}|\slie^I \Psi|^2\right) \de v \desphere.
\end{equation*}
Concerning the term in $\snabla_L \Psi$, we have that, from inequality (\ref{eq:decpsidu}), (\ref{eq:rpalphaz}) and (\ref{eq:masterrconst}), the following bound holds for all $u\geq u_0$:
\begin{equation}\label{eq:declfin}
\begin{aligned}
\int_{\conplus_{u} \cap \{r \geq R\}} |\snabla_L \Psi|^2 \de v \desphere \\
\leq (1+|u|)^{-2} \sum_{i = 0}^2 \left( \fara^\infty [(\snabla_T)^i \phi](P(u_0))+\int_{\conplus_{u_0} \cap \{r \geq R\}}r^2 |\snabla_L (\snabla_T)^{\min\{i,1\}} \phi|^2\de v \desphere\right) \\
+ C (1+ |u|)^{-2} \int_{\conplus_{u_0}\cap \{r \geq R\}} (r^{-2}|\Psi|^2 + |\snabla \Psi|^2+r^2|\snabla_L \Psi|^2)  \de v \desphere
\end{aligned}
\end{equation}
Similarly, considering the corresponding commuted estimates, if $I \in \iota^\Omega_{\leq 2}$, we have, for $u \geq u_0$,
\begin{equation}\label{eq:declfincom}
\begin{aligned}
\int_{\conplus_{u} \cap \{r \geq R\}} |\snabla_L \slie^I \Psi|^2 \de v \desphere \leq C(|u|+1)^{-2} (\snorm{\Psi}{2}{0}{2}^2+\snorm{\phi}{2}{2}{2}^2).
\end{aligned}
\end{equation}
Now, from the commuted version of (\ref{eq:masterrconst}), we obtain, if $I \in \iota^\Omega_{\leq 2}$ and $u \geq u_0$:
\begin{equation}\label{eq:decaypsifin}
\begin{aligned}\int_{\conplus_{u}\cap \{r \geq R\}} (r^{-4}|\slie^I\Psi|^2 + r^{-2}|\snabla \slie^I \Psi|^2)  \de v \desphere\\
 \leq C(|u|+1)^{-2} (\snorm{\Psi}{2}{0}{2}^2+\snorm{\phi}{2}{2}{2}^2).
\end{aligned}
\end{equation}

We now have, using Lemma~\ref{lem:sobpsi} in the Appendix, letting $(u,v,\omega)$ a point in $(u,v)$-coordinates:
\begin{equation*}
\begin{aligned}
|\Psi(u,v,\omega)|^2 \leq C (|u|+1)^{-1} (\snorm{\Psi}{2}{0}{2}^2+\snorm{\phi}{2}{2}{2}^2),
\end{aligned}
\end{equation*}
if $v-u \geq 2R_*$. This implies the claim.
\end{proof}

\section{Decay estimates for $\sigma$ and $\rho$}\label{sec:decaymid}
In this section, we suppose that $F \in \Lambda^2(\mathcal{S}_e)$ is a solution to the full Maxwell system, and we prove Theorem~\ref{prop:decaymax}.

\begin{proof}[Proof of Theorem~\ref{prop:decaymax}]
Let $F \in \Lambda^2(\mathcal{S}_e)$ satisfy the Maxwell Equations (\ref{mw1}) to (\ref{mw6}). Then, the extreme components $\alpha$ and $\alphabar$ satisfy the spin $\pm 1$ Teukolsky Equations ((\ref{eq:teua}) and (\ref{eq:teuabar})), and hence we have the required decay rates for $\alpha$ and $\alphabar$ from Theorem~\ref{prop:decayteu}.

Hence the proof reduces to proving decay for the middle components $\sigma$ and $\rho$. As noticed in Remark~\ref{rm:secret}, we have
\begin{align*}
\phi_A &= r^3(\snabla_A \rho + \svol_{AB} \snabla^B \sigma), \\
\phibar_A &= r^3(-\snabla_A \rho + \svol_{AB} \snabla^B \sigma).
\end{align*}
Hence it is clear that estimates (\ref{eq:decaydue}) and (\ref{eq:decaylie}) still hold with $\phi$ replaced by either $r^3 \snabla \rho$ or $r^3 \snabla \sigma$.

Let us now define $\rho_s(u,v)$ and $\sigma_s(u,v)$:
\begin{equation}\label{eq:rhosdefge}
\rho_s(u,v) := \frac {R^2} {4 \pi} \int_{\omega \in \mathbb{S}^2} \rho(u, v, \omega) \desphere(\omega), \qquad \sigma_s(u,v) := \frac {R^2} {4 \pi} \int_{\omega \in \mathbb{S}^2} \sigma(u, v, \omega) \desphere(\omega).
\end{equation}
We notice that, by integrating each of the Maxwell Equations (\ref{mw3}) -- (\ref{mw6}) on $\mathbb{S}^2$, we have, for all $(u,v) \in \{u \geq u_0\} \cap \{v \geq v_0\}$,
\begin{equation}
\sigma_s (u,v) = \sigma_s (u_0, v_0), \qquad \rho_s(u,v) = \rho_s(u_0, v_0).
\end{equation}
Let us restrict our attention to the estimates for $\rho$. The estimates for $\sigma$ can be obtained in a very analogous manner.

\subsection*{Step 1: region of bounded $r$}
We first consider the region $\{r \leq R\} \cap \mathcal{J}^+(C_{u_0, v_0})$.
Estimates (\ref{eq:decaydue}) and (\ref{eq:decaylie}) imply:
\begin{equation}\label{eq:rhofdec}
\fara^\infty[\snabla \rho](P(u))+ \sum_{i=1}^3\fara^\infty[\snabla \slie_{\Omega_i} \rho](P(u)) \leq C u^{-2} (M_{\rho,\sigma})^2.
\end{equation}
Here, $M_{\rho,\sigma}$ is as in Equation~(\ref{eq:mrhosdef}).

Let $2M < r_c < R$. We now use the Sobolev Lemma in the Appendix \ref{lem:sobforrho} to obtain, if $\tilde v \geq v_0$,
\begin{equation}
\begin{aligned}
\sup_{( u,  \tilde v, \omega) \in \conminus_{\tilde v} \cap \{r_c \leq r \leq R\}}|\rho(u, \tilde v, \omega)-(r(u, \tilde v))^{-2}\rho_s ( u, \tilde v)|^2  \\ \leq C \int_{\conminus_{\tilde v} \cap \{r_c \leq r \leq R\}}(|\snabla \snabla \rho|^2 +|(1-\mu)^{-1}\snabla_{\lbar} \snabla \snabla \rho|^2) (1-\mu)\desphere \de u.
\end{aligned}
\end{equation}
By writing the expression of the Lie derivative, we obtain that there exist positive constants $C_0$, $C_1$, $C_2$ such that, in the region $\{r \leq R\}$,
\begin{equation}
\sum_{i=1}^3 |(1-\mu)^{-1}\snabla_{\lbar}\slie_{\Omega_i} \snabla \rho|^2 \geq |(1-\mu)^{-1}\snabla_\lbar \snabla \snabla \rho| - C_1 |\snabla \rho|^2 - C_2 |\snabla \snabla \rho|^2
\end{equation}

This implies that, possibly renaming $C_1$ and $C_2$, the following inequality holds (recall: $[\slie_{\Omega_i}, \snabla] = 0$):
\begin{equation*}
\begin{aligned}
\int_{\conminus_{\tilde v} \cap \{r_c \leq r \leq R\}}(1-\mu)(|\snabla \snabla \rho|^2 + |(1-\mu)^{-1}\snabla_{\lbar} \snabla \snabla \rho|^2) \desphere \de u \\
 \leq C_1 \fara^\infty[\snabla \rho](P(\tilde u))+ C_2\sum_{i=1}^3\fara^\infty[\snabla \slie_{\Omega_i} \rho]P(\tilde u),
\end{aligned}
\end{equation*}
with $\tilde v - \tilde u = 2R_*$. (\ref{eq:rsnear}) now easily follows from the previous display, along with (\ref{eq:rhofdec}), in the region $\{u \geq u_0\} \cap \{v \geq v_0\}$.
\subsection*{Step 2: region of unbounded $r$}
Let $v \geq v_1 \geq v_0$, $u = u_1 \geq u_0$, and let $v_1 - u_1 = 2R_*$, and $v-u = 2r_*$.
We begin by noticing, by Lemma~\ref{lem:sobsphere} and the definition of $\phi$:
\begin{equation}
\begin{aligned}
|\rho- r^{-2}\rho_s|^2  \lesssim \int_{\mathbb{S}^2} r^4 |\snabla \snabla \rho|^2 \de \mathbb{S}^2 = \int_{\mathbb{S}^2} r^{-2} |\snabla \phi|^2 \de \mathbb{S}^2.
\end{aligned}
\end{equation}
Now, the definition of Lie derivative yields:
\begin{equation*}
\slie_{\Omega_i} \phi_A = \snabla_{\Omega_i} \phi_A + \phi(\snabla_A \Omega_i).
\end{equation*}
This implies the pointwise bound:
\begin{equation}\label{eq:liespez}
\begin{aligned}
|\snabla_L \snabla_{\Omega_i}\phi|^2 = \snabla_L \snabla_{\Omega_i}\phi_A \snabla_L \snabla_{\Omega_i}\phi^A \\
= \snabla_L (\slie_{\Omega_i} \phi_A - \phi(\snabla_A \Omega_i)) \snabla_L (\slie_{\Omega_i} \phi^A - \phi(\snabla^A \Omega_i)) \lesssim |\snabla_L \slie_{\Omega_i}\phi|^2 +  |\snabla_L \phi|^2.
\end{aligned}
\end{equation}
Furthermore, we have
\begin{equation}\label{eq:mixder}
|\snabla_L \snabla \phi|^2 \lesssim \sum_{i=1}^3 \frac 1 {r^2} |\snabla_L \snabla_{\Omega_i}\phi^2|+ \frac 1 {r^2}|\snabla \phi|^2.
\end{equation}
We now use the Cauchy--Schwarz inequality, as well as the computation in Equation~(\ref{eq:liespez}) to obtain the following chain of estimates.
\begin{equation*}
\begin{aligned}
r\int_{\mathbb{S}^2} |\snabla \phi|^2(u,v,\omega) \de \mathbb{S}^2(\omega) \\ \lesssim 
\int_{\conplus_{u}\cap \{r \geq R\}} |\snabla \phi|^2(u,\tilde v,\omega) \de \mathbb{S}^2 \de \tilde v + 
\int_{\conplus_{u}\cap \{r \geq R\}}r(u, \tilde v) \left(|\snabla \phi| |\snabla_L \snabla \phi| \right)(u,\tilde v,\omega)\de \mathbb{S}^2(\omega)\de \tilde v \\ 
\stackrel{(\ref{eq:mixder})}{\lesssim}
\int_{\conplus_{u}\cap \{r \geq R\}}  \sum_{i=1}^3 \left(|\snabla \phi| |\snabla_L \snabla_{\Omega_i} \phi| \right)\de \mathbb{S}^2(\omega) \de \tilde v + \int_{\conplus_{u}\cap \{r \geq R\}} |\snabla \phi|^2\de \mathbb{S}^2(\omega) \de \tilde v 
\\ \stackrel{(\ref{eq:liespez})}{\lesssim}
(M_{\rho,\sigma})^2(|u|+1)^{-2}
+  \left( \int_{\conplus_{u}\cap \{r \geq R\}} |\snabla \phi|^2 \de \mathbb{S}^2 \de \tilde v \right)^{\frac 1 2}  \left( \int_{\conplus_{u}\cap \{r \geq R\}} \sum_{i=1}^3 |\snabla_L  \slie_{\Omega_i }\phi|^2 \de \mathbb{S}^2 \de \tilde v \right)^{\frac 1 2}\\
+  \left( \int_{\conplus_{u}\cap \{r \geq R\}}|\snabla \phi|^2 \de \mathbb{S}^2 \de \tilde v \right)^{\frac 1 2}  \left( \int_{\conplus_{u}\cap \{r \geq R\}}  |\snabla_L  \phi|^2 \de \mathbb{S}^2 \de \tilde v \right)^{\frac 1 2}
\\
\lesssim (M_{\rho,\sigma})^2(|u|+1)^{-2}.
\end{aligned}
\end{equation*}
\begin{remark}
Notice that the flux $\fara^\infty$ allows us to estimate not only the $\snabla_L \phi$ term, but also the $\snabla \phi$ term.
\end{remark}
This implies finally that
\begin{equation}\label{eq:decayrhot}
|\rho(u,v,\omega) - r^{-2}\rho_s(u_0,v_0)| \leq C r^{- \frac 3 2} (|u|+1)^{-1}  M_{\rho,\sigma}.
\end{equation}
Similarly, we compute
\begin{equation}\label{eq:phifar}
\begin{aligned}
r^2\int_{\mathbb{S}^2} |\snabla \phi|^2(u,v,\omega) \de \mathbb{S}^2(\omega) \lesssim \int_{\mathbb{S}^2} \sum_{i = 1}^3 \gbar\left(\snabla_{\Omega_i}\phi,\snabla_{\Omega_i}\phi \right)(u,v, \omega)\desphere(\omega) \\\lesssim 
 \int_{\conplus_{u}\cap \{r \geq R\}} |\snabla\phi|^2(u,\tilde v,\omega)\de \tilde v \desphere(\omega) \\
 +  \int_{\conplus_{u}\cap \{r \geq R\}} \snabla_L( \sum_{i=1}^3\gbar\left(\snabla_{\Omega_i}\phi,\snabla_{\Omega_i}\phi \right)) (u,\tilde v,\omega) \de \tilde v \desphere(\omega).
\end{aligned}
\end{equation}
The first term in the right hand side of the last display is estimated by (\ref{eq:decaydue}).
We again use Equation~(\ref{eq:liespez}) to obtain, for the second term in (\ref{eq:phifar}), the following chain of estimates.
\begin{equation*}
\begin{aligned}
\left| \int_{\conplus_{u}\cap \{r \geq R\}} \snabla_L( \sum_{i=1}^3\gbar\left(\snabla_{\Omega_i}\phi,\snabla_{\Omega_i}\phi \right)) \de \tilde v\desphere(\omega)\right| \\ \lesssim 
\sum_{i=1}^3 \left( \int_{\conplus_{u}\cap \{r \geq R\}} r^{-2}|\snabla_{\Omega_i}\phi|^2 \desphere(\omega)\de \tilde  v\right)^{\frac 1 2} \times \left( \int_{\conplus_{u}\cap \{r \geq R\}} r^2|\snabla_L \snabla_{\Omega_i}\phi|^2 \de \tilde v\desphere(\omega)\right)^{\frac 1 2} \\ \lesssim 
\sum_{i=1}^3 \left(  \int_{\conplus_{u}\cap \{r \geq R\}} |\snabla \phi|^2 \desphere(\omega)\de \tilde v\right)^{\frac 1 2} \times \left( \int_{\conplus_{u}\cap \{r \geq R\}} r^2 ( |\snabla_L \slie_{\Omega_i}\phi|^2 +  |\snabla_L \phi|^2)\de \tilde v\desphere(\omega)\right)^{\frac 1 2} \\
\lesssim  (|u|+1)^{-1} (M_{\rho,\sigma})^2.
\end{aligned}
\end{equation*}
The last inequality follows by the decay estimates (\ref{eq:decaylie}), (\ref{eq:decaydue}), and the uniform boundedness estimate (\ref{pdue}) applied to the flux containing the $L$-derivative.

This implies the decay rate for $\rho$:
\begin{equation}
|\rho(u,v,\omega)-r^{-2}\rho_s(u_0,v_0)| \leq C (|u|+1)^{-\frac 1 2} r^{-2}M_{\rho,\sigma}.
\end{equation}
This concludes the proof of the Theorem.
\end{proof}

\appendix

\section{Derivation of the null decomposition of the Maxwell system and of the spin $\pm$1 Teukolsky equations}\label{sec:dernulldec}

\begin{proof}[Proof of Proposition~\ref{prop:mwnull}] The proof is by calculation in the null frame and Hodge dualization.
\subsection*{Step 1: the full Maxwell system}
Recall: $L:=\partial_{r_\star} + \partial_t = \partial_u$, $\lbar:= \partial_t - \partial_{r_\star} = \partial_u$.
In the following calculations, when an uppercase letter appears, it signifies contraction with one of the basis elements $\partial_{\theta^A}$, $\partial_{\theta^B}$, where $\theta^A$ and $\theta^B$ are local coordinates for $\mathbb{S}^2$.

Now set $e_A := \partial_{\theta^A}$, and $e_B := \partial_{\theta^B}$. Let us now calculate
\begin{equation}\label{eq:projected}
\begin{aligned}
	\nabla_{e_A} e_B &= \slashed{\nabla}_{e_A} e_B - \frac 1 2 (1-\mu)^{-1}\left(g(\nabla_{e_A} e_B, \lbar)L + g(\nabla_{e_A} e_B, L)\lbar \right)\\
	&= \slashed{\nabla}_{e_A} e_B + \frac 1 2 (1-\mu)^{-1}\left(g(e_B, \nabla_{e_A}\lbar)L + g(e_B, \nabla_{e_A}L)\lbar \right)\\
	&= \slashed{\nabla}_{e_A} e_B + \frac 1 2 (1-\mu)^{-1}\left(g(e_B, -\frac{1-\mu}{r}e_A)L + g(e_B, \frac{1-\mu}{r}e_A)\lbar \right) \\
	&= \slashed{\nabla}_{e_A} e_B + \frac 1 {2r} (\lbar - L) \slashed{g}_{AB}.
\end{aligned}
\end{equation}
We begin by calculating:
\begin{equation}\label{eq:abl}
\begin{aligned}
 \nabla_A F_{BL} &= e_A F(e_B,L)-F(\nabla_{e_A} e_B, L)-F(e_B, \nabla_{e_A} L) \\ &=
  e_A F(e_B,L)-F(\snabla_{e_A} e_B, L) -\frac 1 {2r} F((\lbar-L)\slashed{g}_{AB}, L) -F(e_B, \frac{1-\mu}{r}e_A) \\&=
  \snabla_A \alpha_B - \frac{1-\mu}{r} \rho \slashed{g}_{AB}+ \frac{1-\mu}{r}\sigma \slashed{\varepsilon}_{AB}.
\end{aligned}
\end{equation}
Furthermore, we have:
\begin{equation}\label{eq:ablbar}
\begin{aligned}
 \nabla_A F_{B\lbar} &= e_A F(e_B,\lbar)-F(\nabla_{e_A} e_B, \lbar)-F(e_B, \nabla_{e_A} \lbar) \\ &=
  e_A F(e_B,\lbar)-F(\snabla_{e_A} e_B, \lbar) -\frac 1 {2r} F((\lbar-L)\slashed{g}_{AB}, \lbar) -F(e_B, -\frac{1-\mu}{r}e_A) \\&=
  \snabla_A \alphabar_B - \frac{1-\mu}{r} \rho \slashed{g}_{AB}- \frac{1-\mu}{r}\sigma \slashed{\varepsilon}_{AB}.
\end{aligned}
\end{equation}
By the null decomposition of the Hodge dual of $F$, it follows that
\begin{equation*}
\begin{aligned}
 \nabla_A \fdual_{BL} &= \slashed{\varepsilon}_{CB} \snabla_A \alpha^C - \frac{1-\mu}{r} \sigma \slashed{g}_{AB}-\frac{1-\mu}{r}\rho \slashed{\varepsilon}_{AB}\\
 \nabla_A \fdual_{B\lbar} &= -\slashed{\varepsilon}_{CB} \snabla_A \alphabar^C - \frac{1-\mu}{r} \sigma \slashed{g}_{AB}+\frac{1-\mu}{r}\rho \slashed{\varepsilon}_{AB}.
\end{aligned}
\end{equation*}
We also have:
\begin{equation}\label{eq:llbar}
\begin{aligned}
	\frac 1 2 \nabla_A F_{\lbar L} &= \frac 1 2 \nabla_A(F(\lbar, L)) - \frac 1 2 F(\nabla_A \lbar, L) - \frac 1 2 F(\lbar, \nabla_A L)\\
	 &=
	\snabla_A \rho + \frac 1 2 \frac{1-\mu}{r} (\alphabar_A + \alpha_A).
\end{aligned}
\end{equation}
Again, by taking the Hodge dual,
\begin{equation*}
	\frac 1 2 \nabla_A \fdual_{\lbar L} = \snabla_A \rho + \frac 1 2 \frac{1-\mu}{r} (- \slashed{\varepsilon}_{CA }\alphabar^C + \svol_{CA }\alpha^C)
\end{equation*}
Now, use Equations (\ref{eq:projected}) and (\ref{eq:llbar}) to get:
\begin{equation} \label{premw1}
\begin{aligned}
	\nabla_L F_{\lbar A} &+ \nabla_A F_{L\lbar} + \nabla_\lbar F_{AL} = 
	-\nabla_L F_{A\lbar} + \nabla_A F_{L\lbar} + \nabla_{\lbar} F_{AL} \\&=
	- \snabla_L \alphabar_A + \snabla_\lbar \alpha_A - 2(1-\mu) \snabla_A \rho -F(\nabla_A L, \lbar) - F(L, \nabla_A \lbar) \\&=
	- \snabla_L \alphabar_A + \snabla_\lbar \alpha_A - 2(1-\mu) \snabla_A \rho -\frac {1-\mu}{r} \alpha_A - \frac {1-\mu}{r} \alphabar_A  \\ &=
	-\frac 1 r \snabla_L(r \alphabar_A)+\frac 1 r \snabla_\lbar(r \alpha_A)-2(1-\mu) \snabla_A \rho = 0.
\end{aligned}
\end{equation}
By taking the dual of the last equation, we obtain
\begin{equation}
\begin{aligned}
	\nabla_L \fdual_{\lbar A} + \nabla_A \fdual_{L\lbar} + \nabla_\lbar \fdual_{AL} \\
	 = 
	\frac 1 r \snabla_L(r \svol_{BA} \alphabar^B)+ \frac 1 r \snabla_{\lbar}(r \svol_{BA} \alpha^B) - 2 (1 -\mu)\snabla_A \sigma = 0
\end{aligned}
\end{equation}
The last display is equivalent to:
\begin{equation}
\begin{aligned}
	 -\frac 1 r \snabla_L(r \alphabar_A)-\frac 1 r \snabla_\lbar(r \alpha_A)+2(1-\mu)\svol_{AB} \snabla^B \sigma = 0.
\end{aligned}
\end{equation}
We therefore obtain Equations (\ref{mw1}) and (\ref{mw2}):
\begin{equation} \boxed{
\begin{aligned}
	\frac 1 r \snabla_L (r\alphabar_A) + (1-\mu) (\snabla_A \rho - \svol_{AB} \snabla^B \sigma) = 0,\\
	\frac 1 r \snabla_\lbar (r\alpha_A) - (1-\mu) (\snabla_A \rho + \svol_{AB} \snabla^B \sigma) = 0.
\end{aligned}}
\end{equation}
Now, let us calculate, with the aid of (\ref{eq:ablbar}), the following expression:
\begin{equation}
\begin{aligned}
	&\nabla_A F_{B\lbar} + \nabla_\lbar F_{AB} + \nabla_B F_{\lbar A}\\ &=
	\snabla_A \alphabar_B - \frac{1-\mu}{r} \rho \gbar_{AB} - \frac{1-\mu}{r} \sigma \svol_{AB}+\\
	&-\snabla_B \alphabar_A + \frac{1-\mu}{r} \rho \gbar_{AB} - \frac{1-\mu}{r} \sigma \svol_{AB} + \svol_{AB} \snabla_\lbar \sigma=0.
\end{aligned}	
\end{equation}
Contracting the last display with $\svol^{AB}$ we obtain the following equation:
\begin{equation}\boxed{
\begin{aligned}
 	\curl \alphabar - 2 \frac{1-\mu}{r} \sigma + \snabla_\lbar \sigma = 0.
\end{aligned}}
\end{equation}
By taking the dual of the last equation, we obtain furthermore
\begin{equation} \boxed{
\begin{aligned}
 	-\dive \alphabar + 2 \frac{1-\mu}{r} \rho - \snabla_\lbar \rho = 0.
\end{aligned}}
\end{equation}
We finally compute, with the aid of (\ref{eq:abl}), the following expression:
\begin{equation}
\begin{aligned}
	\nabla_A F_{BL} + \nabla_L F_{AB} + \nabla_B F_{LA}=
	\snabla_A \alpha_B - \frac{1-\mu}{r} \rho \gbar_{AB} + \frac{1-\mu}{r} \sigma \svol_{AB}+\\
	-\snabla_B \alpha_B + \frac{1-\mu}{r} \rho \gbar_{AB} + \frac{1-\mu}{r} \sigma \svol_{AB} + \snabla_L \sigma = 
 	\curl \alpha + 2 \frac{1-\mu}{r} \sigma + \snabla_L \sigma = 0.
\end{aligned}
\end{equation}
We obtain:
\begin{equation*}
\boxed{\curl \alpha + 2 \frac{1-\mu}{r} \sigma + \snabla_L \sigma = 0.}
\end{equation*}
Taking the dual of the last expression yields
\begin{equation}\boxed{
\begin{aligned}
 	\dive \alpha - 2 \frac{1-\mu}{r} \rho - \snabla_L\rho= 0.
\end{aligned}}
\end{equation}
This concludes the derivation of the null decomposition of the Maxwell system.

\subsection*{Step 2: spin $\pm$1 Teukolsky equations}
We now turn to the derivation of the spin $\pm$1 Teukolsky equations.
Recall the following facts, which can be checked by explicit calculation:
\begin{equation*}
\begin{aligned}
&\snabla_L \svol_{AB} = 0, \qquad \snabla_\lbar \svol_{AB} = 0,\qquad \snabla_L \gbar_{AB} = 0, \qquad \snabla_\lbar \gbar_{AB} = 0, \\ &[r\snabla_A, \snabla_L] = [r\snabla_A, \snabla_\lbar] = 0.
\end{aligned}
\end{equation*}
Operate now on Equation~(\ref{mw1}) with $\lbar$, in order to obtain
\begin{align*}
\snabla_\lbar \snabla_L (r \alphabar_A) + \lbar(1-\mu) r (\snabla_A \rho - \svol_{AB}\snabla^B \sigma) + \underbrace{(1-\mu) r (\snabla_A \snabla_\lbar \rho - \svol_{AB}\snabla^B \snabla_\lbar \sigma)}_{(*)} = 0.
\end{align*}
Using the expression for $\lbar \rho$ (\ref{mw4}) and $\lbar \sigma$ (\ref{mw3}), we obtain
\begin{align*}
(*) = (1-\mu) r \left(2 \frac{1-\mu}{r}\snabla_A \rho -2 \frac{1-\mu}{r}\svol_{AB}\snabla^B \sigma\right)+(1-\mu)r\underbrace{(-\snabla_A \dive \alphabar + \svol_{AB}\snabla^B \curl \alphabar)}_{:= -(\text{ANG})}.
\end{align*}
Hence, 
\begin{align*}
\snabla_\lbar \snabla_L (r \alphabar_A) -(1-\mu)\frac{2M}{r^2} r (\snabla_A \rho - \svol_{AB} \snabla^B \sigma) + 2(1-\mu)^2 (\snabla_A \rho - \svol_{AB} \snabla^B \sigma)\\
+(1-\mu)r \cdot \text{ANG} = 0.
\end{align*}
This implies, upon substitution using (\ref{mw1}) again,
\begin{align*}
\snabla_\lbar \snabla_L (r \alphabar_A) -\frac 2 r \left(1-\frac{3M}{r}\right)\snabla_L(r \alphabar_A)-(1-\mu)r \cdot \text{ANG} = 0.
\end{align*}
Using now Lemma~\ref{lem:angular}, we obtain the claim for $\alphabar$. The reasoning for $\alpha$ is analogous.
\end{proof}

\begin{lemma}\label{lem:angular}
Let $\mathbb{S}^2$ be endowed with the standard metric $g_{\mathbb{S}^2}$ and denote by $D$ the Levi-Civita connection associated to such metric on $\mathbb{S}^2$. Let $\varepsilon_{AB}$ be the standard volume form. Let $\omega$ be a smooth 1-form on $\mathbb{S}^2$, let $\daiv$ and $\cherl$ the associated covariant divergence and curl: 
$$
\daiv \omega := g^{AB}D_A \omega_B  \qquad \cherl \omega := \varepsilon^{AB} D_A \omega_B.
$$
Let $\Delta:= g_{\mathbb{S}^2}^{AB} D_A D_B$ be the covariant Laplacian. Then, we have:
\begin{equation}\label{eq:angulteu}
D_A \daiv \omega - \varepsilon_{AB} D^B \cherl \omega = \Delta \omega_A - \omega_A.
\end{equation}
\end{lemma}

\begin{proof}
Let $p \in \mathbb{S}^2$. Let us fix a vector $V \in T_p \mathbb{S}^2$, and let us set up coordinates $(\theta, \varphi)$ such that the coordinates of $p$ are $(\pi/2, 0)$, $\partial_\theta|_p = V$, and finally the metric in these local coordinates is represented by the two-form $\de \theta \otimes \!\!\de \theta + \sin^2 \theta \de \varphi \otimes \! \!\de \varphi$.
Let
\begin{equation}
\mathcal{T}_A := D_A \daiv \omega - \varepsilon_{AB} D^B \cherl \omega.
\end{equation} Then, since the left hand side of (\ref{eq:angulteu}) is a tensor, we have
\begin{align*}
\mathcal{T}(\partial_\theta)|_p = \mathcal{T}(V)= \underbrace{D_{\partial_\theta} \daiv \omega}_{(1)} - \underbrace{\varepsilon_{\theta B} D^B \cherl \omega}_{(2)}.
\end{align*}
\begin{align*}(1) = \left.
\partial_\theta \left(\partial_\theta \omega (\partial_\theta)-\omega(D_{\partial_\theta}\partial_\theta)+\frac{\partial_\varphi}{\sin \theta}\omega \left(\frac{\partial_\varphi}{\sin \theta}\right)-\omega\left(D_{\frac{\partial_\varphi }{\sin \theta}} \frac{\partial_\varphi }{\sin \theta}\right)\right)\right|_{(\pi/2, 0)} \\
= \left.\left(\partial_\theta \partial_\theta \omega(\partial_\theta) + \partial_\theta \left(\frac{\partial_\varphi}{\sin \theta}\omega \left(\frac{\partial_\varphi}{\sin \theta}\right) \right) - \omega(\partial_\theta)\right)\right|_{(\pi/2, 0)}
\end{align*}
\begin{align*}-(2) = \left.-\frac{\partial_\varphi}{\sin \theta} \left((D_{\partial_\theta} \omega)\left(\frac{\partial_\varphi}{\sin \theta} \right) -(D_\frac{\partial_\varphi}{\sin \theta} \omega)\left({\partial_\theta} \right)\right) \right|_{(\pi/2, 0)}  \\
= - \left.\frac{\partial_\varphi}{\sin \theta } \left(\partial_\theta \omega \left(\frac{\partial_\varphi}{\sin \theta}\right) + \frac {\cos \theta} {\sin^2 \theta } \omega(\partial_\varphi) \right)\right|_{(\pi/2, 0)} +\left. \frac{\partial_\varphi}{\sin \theta }\frac{\partial_\varphi}{\sin \theta } \omega(\partial_\theta)\right|_{(\pi/2, 0)}.
\end{align*}
Upon summation, 
\begin{equation*}
\mathcal{T}(\partial_\theta)|_p =
\left.\left(\partial_\theta \partial_\theta \omega(\partial_\theta)- \omega(\partial_\theta)\right)\right|_{(\pi/2, 0)}
 +\left. \frac{\partial_\varphi}{\sin \theta }\frac{\partial_\varphi}{\sin \theta } \omega(\partial_\theta)\right|_{(\pi/2, 0)}.
\end{equation*}
A calculation of the covariant Laplacian in coordinates $(\theta, \varphi)$ at $(\pi/2, 0)$ yields the claim.
\end{proof}

\section{A Poincar\'e lemma}
Let us first set some notation for this Section.
\begin{itemize}
\item Consider the sphere $\mathbb{S}^2$, with the standard metric $g_{\mathbb{S}^2}$ and the standard volume form $\varepsilon$.
\item Denote by $D$ the Levi-Civita connection on the sphere $\mathbb{S}^2$ related to the standard metric $g_{\mathbb{S}^2}$.
\item Denote by $\Delta$ the covariant Laplacian associated to $D$, $\Delta := g_{\mathbb{S}^2}^{AB} D_A D_B$.
\item Denote by $\Delta_H$ the Hodge-de-Rham laplacian:
$$
\Delta_H := d \delta + \delta d,
$$
where $\delta := - \star d \star$ is the codifferential. Here, if $\omega$ is a one-form, the Hodge dual $(\star \omega)$ is still a one-form, defined to satisfy
$$
(\star \omega)_A := \varepsilon_{AB} \omega^B.
$$ 
\item Let our convention on the Riemann tensor be
$$
R(X,Y)Z := (D_X D_Y - D_Y D_X - D_{[X,Y]})Z,
$$
where $X, Y, Z$ are smooth vectorfields. Consequently we denote,
$$
\text{Rm}_{ABCD} := g_{\mathbb{S}^2}(R(\partial_{\theta^A}, \partial_{\theta^B})\partial_{\theta^C} , \partial_{\theta^D}).
$$
Here, $(\theta^1, \theta^2)$ is a local coordinate system on $\mathbb{S}^2$, and the previous equation defines uniquely the 4-covariant tensor $\text{Rm}$.

\item We define the Ricci tensor so that $\text{Ric}_{AB} := \text{Rm}\indices{_A_C^C_B}$.

\item Under these conventions, if $\omega_A$ is a one-form, the commutation relation holds:
\begin{equation*}
D_A D_B \omega_C - D_B D_A \omega_C = - \text{Rm}\indices{_{ABC}^D} \omega_D.
\end{equation*}
\end{itemize}
The main purpose of this section is to give a proof of the following elementary inequality.

\begin{lemma}[Poincar\'e inequality for 1-forms on $\mathbb{S}^2$]\label{lem:poinca}
Let $\omega$ be a smooth $1$-form on $\mathbb{S}^2$. We have the inequality
\begin{equation}
\int_{\mathbb{S}^2} |D \omega|^2 \desphere \geq \int_{\mathbb{S}^2} |\omega|^2 \desphere.
\end{equation}
\end{lemma}

We first need a classical result.
\begin{lemma}[Hodge]\label{lem:hodge}
Let $\omega$ be a smooth one-form on $\mathbb{S}^2$, then there exist smooth functions $f$ and $g$ $\in \mathcal{C}^\infty(\mathbb{S}^2)$ such that
\begin{equation*}
\omega = d f + \star d g.
\end{equation*}
\end{lemma}

We deduce a simple case of the Bochner--Weitzenb\"ock identity.
\begin{lemma}
For any one-form $\omega$ on the sphere $\mathbb{S}^2$, we  have
\begin{equation}\label{eq:second}
\int_{\mathbb{S}^2} |D \omega|^2 \desphere = \int_{\mathbb{S}^2}(\Delta_H \omega)_A \omega^A \desphere-\int_{\mathbb{S}^2}|\omega|^2\desphere.
\end{equation}
\end{lemma}

\begin{proof}
We compute
\begin{equation*}
\begin{aligned}
g^{A B} D_A D_B D_C f = g^{A B} D_A D_C D_B f = g^{A B}D_C  D_A D_B f - g^{A B} \text{Rm}\indices{_{ACB}^E} D_E f.
\end{aligned}
\end{equation*}
Now, on functions, $\Delta_H = - \Delta$, and $\text{Rm}\indices{_A_C^A^E} = -\text{Ric}\indices{_C^E}$. Therefore we have, considering $\text{Ric}$ as a map from one-forms to one-forms:
\begin{equation*}
\Delta (d f) = - \Delta_H (d f) + \text{Ric} (d f),
\end{equation*}
for $f \in \mathcal{C}^\infty(\mathbb{S}^2)$.
Since $D_A \varepsilon_{BC} = 0$, the same holds for the dual:
\begin{equation*}
\Delta (\star d g) = - \Delta_H (\star d g) + \text{Ric} (\star d g),
\end{equation*}
for $g \in \mathcal{C}^\infty(\mathbb{S}^2)$.

Consider now a one-form $\omega$, by the Hodge decomposition (Lemma~\ref{lem:hodge}), we then have
\begin{equation}
(\Delta \omega)_A = - (\Delta_H \omega)_A + \text{Ric}\indices{_A^B}\omega_B.
\end{equation}
Now, $\text{Ric} = g_{\mathbb{S}^2}$. We now contract the previous display with $\omega^A$. We integrate by parts on $\mathbb{S}^2$ in order to obtain the claim.
\end{proof}

We then characterize the spectrum of $\Delta_H$ on one-forms.

\begin{lemma}\label{lem:spec}
The smooth one-form $\omega$ is an eigenvector for $\Delta_H$ if and only if it is of the form $\omega = d f + \star d g$, with $f$ and $g$ smooth eigenfunctions of $\Delta_H$ with the same eigenvalue.
\end{lemma}
\begin{proof}
If $\Delta_H f = \lambda f$ and $\Delta_H g = \lambda g$, then
\begin{equation*}
\Delta_H (df) = d \delta d f + 0 =\lambda d f,
\end{equation*}
also
\begin{equation*}
\Delta_H (\star dg) = d \delta d (\star dg) + 0 =\lambda \star d g.
\end{equation*}
On the other hand, if $\omega$ is a one-form on the sphere, we have, by the Hodge Theorem, that there exist $f$ and $g$ such that $\omega = d f + \star d g$. Then, imposing the eigenvalue condition, we have
\begin{equation}
\Delta_H (d f + \star d g) = \lambda (d f + \star d g) \implies d (\Delta_H f-\lambda f) = \star d (\Delta_H g-\lambda g).
\end{equation}
This implies, since the only harmonic functions on the sphere are the constant functions, $\Delta_H f-\lambda f = c_1$, $\Delta_H g-\lambda g=c_2$. By considering $f_1 := f+ c_1/\lambda$ and $g_1 := g + c_2/\lambda$, we obtain two functions $f_1$ and $g_1$ in the conditions.
\end{proof}
We are now ready to prove Lemma~\ref{lem:poinca}.

\begin{proof}[Proof of Lemma~\ref{lem:poinca}]
By Lemma~\ref{lem:spec}, the spectrum of $\Delta_H$ on $\mathbb{S}^2$ on one-forms is the same as the spectrum on functions. 

Let us now take the Hodge decomposition of $\omega$, $\omega = df + \star dg$, and write $f$, $g$ in the spherical harmonic decomposition. We obtain sequences $f_k$, $g_k$ such that
\begin{equation*}
f_k \stackrel{L^2(\mathbb{S}^2)}{\to} f, \qquad g_k \stackrel{L^2(\mathbb{S}^2)}{\to} g.
\end{equation*}
Since $f$, $g$ are smooth, and hence in particular belong to $H^2(\mathbb{S}^2)$, the approximation by spherical harmonics is in $H^2(\mathbb{S}^2)$:
\begin{equation*}
f_k \stackrel{H^2(\mathbb{S}^2)}{\to} f, \qquad g_k \stackrel{H^2(\mathbb{S}^2)}{\to} g.
\end{equation*}
Letting now $\omega^{(k)}_A := d f_k + \star d g_k$, we have, from Equation~(\ref{eq:second}), and from the fact that the smallest eigenvalue of $\Delta_H$ is $2$,
\begin{equation*}
\int_{\mathbb{S}^2} |D \omega^{(k)}|^2 \desphere \geq \int_{\mathbb{S}^2} (2|\omega^{(k)}|^2-|\omega^{(k)}|^2 ) \desphere = \int_{\mathbb{S}^2} |\omega^{(k)}|^2 \desphere.
\end{equation*}
Taking the limit $k \to \infty$, and using the continuity of the norm, we obtain the claim.
\end{proof}

\section{Sobolev lemmas}
These Sobolev lemmas are used throughut the note. The proofs being standard, we do not provide them.
\begin{lemma}[Sobolev estimate for scalar functions on the sphere]\label{lem:sobsphere}
Let $(\mathbb{S}^2, g_{\mathbb{S}^2})$ be the two-sphere with the standard metric, let $D$ be the associated Levi-Civita connection, and let $f$ be a smooth function $f: \mathbb{S}^2 \to \R$. Let $\bar f := \frac 1 {4 \pi} \int_{\mathbb{S}^2} f\desphere$ be the spherical average of $f$. There exists a universal constant $C$ such that
\begin{equation}
\sup_{\mathbb{S}^2}|f - \bar f|^2 \leq C \int_{\mathbb{S}^2} |DD f|^2\desphere.
\end{equation}
\end{lemma}

\begin{lemma}\label{lem:sobone}
There is a universal constant $C$ such that, for any one-form $\eta$ on $\mathbb{S}^2$, the following inequality holds:
\begin{equation}
\sup_{\mathbb{S}^2} |\eta|^2 \leq C \int_{\mathbb{S}^2} (|D D \eta|^2 +|\eta|^2)\desphere.
\end{equation}
\end{lemma}

\begin{lemma}[1-d Sobolev embedding]\label{lem:1dsob}
Let $(a,b) \subset \R$, with $- \infty <a < b < \infty$. Let $f:(a,b) \to \R$, $f \in W^{1,1}(\R)$. Then there holds:
\begin{equation}
\norm{f(x)-(b-a)^{-1}\int_a^b f(y) \de y}_{L^\infty(\R)}\leq \norm{f'}_{L^1(\R)}.
\end{equation}
\end{lemma}
A straightforward application of Lemma~\ref{lem:1dsob} yields the following Lemma.
\begin{lemma}[Sobolev inequality involving only certain derivatives]\label{lem:sobforrho}
 Let $R > 2M$. Let $f: \mathcal{S}_e \to \R$ be a smooth function. Let $v_0 \geq 0$. Let $\bar f$ be the mean of $f$ over the spheres:
\begin{equation}
\bar f(u,v) := \frac 1 {4 \pi} \int_{\mathbb{S}^2} f(u,v,\omega) \desphere(\omega).
\end{equation}
Then, there exists a constant $C = C(R)$ such that, for any $\tilde v \geq v_0$,
\begin{equation}\label{eq:sobosigma}
\begin{aligned}
\sup_{( u,  \tilde v, \omega) \in \conminus_{\tilde v} \cap \{2M \leq r \leq R\}}|f(u, \tilde v, \omega)-\bar f ( u, \tilde v)|^2  \\ 
\leq C \int_{\conminus_{\tilde v} \cap \{2M \leq r \leq R\}}(|\snabla \snabla f|^2 + |(1-\mu)^{-1}\snabla_{\lbar} \snabla \snabla f|^2)(1-\mu) \desphere \de u
\end{aligned}
\end{equation}
\end{lemma}
Finally, we state the following Lemma, which we need in the proof of Proposition~\ref{prop:morealpha}.

\begin{lemma}\label{lem:sobpsi}
For any $R >0$ there exists a constant $C >0$ such that the following holds. Let $\Psi \in \Lambda^1(\mathcal{B})$ be a $1$-form tangent to the spheres of constant $r$-coordinate.
Then, we have
\begin{equation}\label{eq:sobpsi}
\begin{aligned}
|\Psi|^2(\bar u,\bar v,\omega) \leq C \sum_{I \in \iota^\Omega_{\leq 2}}\int_{\conplus_{\bar u} \cap \{R \leq r \leq 2 R \}} |\slie^I \Psi|^2 \de v \desphere \\
+\sum_{I \in \iota^\Omega_{\leq 2}} \left(\int_{\conplus_{\bar u} \cap \{R \leq  r \leq 2R\}} |\slie^I \Psi|^2 \de v \desphere+ \int_{\conplus_{\bar u} \cap \{ r \geq R\}} |\snabla_L \slie^I \Psi|^2 \de v \desphere \right)^{\frac 1 2} \\
\times\left(\int_{\conplus_{\bar u} \cap \{ r \geq R\}} v^2 |\snabla_L \slie^I \Psi|^2 \de v \desphere \right)^{\frac 1 2}
\end{aligned}
\end{equation}
Here, $\bar u$ and $\bar v$ are such that
\begin{equation*}
\bar v - \bar u \geq 2 R_* = 2(R+2M\log(R-2M)-3M-2M\log(M)).
\end{equation*}
\end{lemma}

\begin{proof}[Sketch of proof]
The lemma is a straightforward consequence of the Sobolev inequality on spheres (Lemma~\ref{lem:sobone}), the Cauchy--Schwarz inequality, as well as the $1$-dimensional Hardy inequality:
\begin{equation}
\int_0 ^\infty x^{-2} (F(x))^2 \de x \leq C \int_{0}^\infty \left(\desu{f}{x}\right)^2\de x.
\end{equation}
Here, $f$ is any smooth real-valued function on $\R$, and $F$ is its primitive (in $x$) which vanishes at $0$.
\end{proof}

\section{Comparison with the work by Andersson, B\"ackdahl and Blue.}\label{sec:compare}

Here, we compare the present approach with the work in \cite{newblue}. We don't want to enter deeply into the spinor formalism used in the relevant paper. Let us focus on the result and compare it to what we obtain.

The authors obtain, in paper \cite{newblue}, an estimate of the form
\begin{equation*}
\begin{aligned}
\int_{\future(\{t=t_1\})}\frac{(r-3M)^2}{r^3}(|\beta_{\hat X}|^2+|\beta_{\hat Y}|^2)\\+M \frac{M(r-2M)}{r^3}|\beta_{\hat Z}|^2 + \frac{M(r-3M)^2(r-2M)}{r^5}|\beta_{\hat T}|^2 \de\text{Vol}\\ \lesssim E(\{t= t_1\})
\end{aligned}
\end{equation*}
Here,
\begin{equation*}
E(\{t=t_1\}) := \int_{\{t=t_1\}} (|\beta_{\hat X}|^2+|\beta_{\hat Y}|^2+|\beta_{\hat Z}|^2+|\beta_{\hat T}|^2)r^2 \de r \sin \theta \de \theta \de \phi.
\end{equation*}
And
\begin{align*}
\beta_{\hat T} &= (1-\mu)^{-\frac 1 2}\nabla_{\partial_t}(r\rho+ir\sigma)\\
\beta_{\hat X} &= \nabla_{\partial_\theta}(\rho+i\sigma) \\
\beta_{\hat Y} &= \csc \theta \nabla_{\partial_\phi}(\rho+i \sigma)\\
\beta_{\hat Z} &= \frac 1 r \partial_r (r^2 (\rho + i\sigma)).
\end{align*}
In this case, a first-order Morawetz estimate is achieved which does not ``see'' the non-decaying modes (indeed, $\beta$ vanishes on the non-radiating modes).

If we compare with our Morawetz estimate, see Equation~(\ref{eq:fullen}), we see that we achieve a Morawetz estimate at the level of two derivatives of the field. Hence, in our approach, we require control on more derivatives in the initial data.

\bibliography{max.bib}
\bibliographystyle{plain}

\end{document}